\documentclass[a4wide]{article}

\usepackage{amsthm} 

\usepackage{amssymb}
\usepackage{amsmath}
\usepackage{txfonts}
\usepackage{amssymb}
\usepackage{enumerate}
\usepackage{amsfonts}
\usepackage{times}
\usepackage{mathrsfs}
\usepackage{amscd}
\usepackage{graphicx}
\usepackage{makeidx}
\usepackage{hyperref}
\usepackage{authblk}

\usepackage{circle} 

\usepackage{version}
 
\includeversion{full}

\newcommand{\commentout}[1]{~\\***CUT[\\ #1 ~\\]******ENDCUT\\}

\newcommand{\ATL}{ATL}
\newcommand{\ATEL}{ATEL}

\newcommand{\LTL}{\mbox{LTL}}
\newcommand{\ECTL}{\mbox{ECTL}}

\newcommand{\CTLK}{\mbox{CTLK}}
\newcommand{\CTL}{\mbox{CTL}}
\newcommand{\CTLsK}{\mbox{CTL$^*$K}}
\newcommand{\CTLs}{\mbox{CTL$^*$}}
\newcommand{\ESL}{\mbox{ESL}}

\newcommand{\trans}{\rightarrow}

\newcommand{\R}{{\cal R}}
\newcommand{\C}{{\cal C}}
\newcommand{\K}{{\cal K}}

\newcommand{\until}{U}

\newcommand{\ptrans}[1]{\stackrel{#1}{\longrightarrow}}

\newcommand{\Prop}{Prop}
\newcommand{\SVar}{V\!ar}
\newcommand{\nat}{\mathbb{N}}

\newcommand{\powerset}[1]{{\cal P}(#1) }

\newcommand{\dimp}{\Leftrightarrow}

\newcommand{\ets}{{\cal E}}

 \newtheorem{propn}{Proposition}
  \newtheorem{cor}{Corollary}

      \newtheorem{theorem}{Theorem}

   \newtheorem{lemma}{Lemma}

\newtheoremstyle{examplesty}         
            {}                   
            {}                   
            {\upshape}           
            {}                   
            {\bfseries} 
            {}                   
            {1em}                
            {}                   

\theoremstyle{examplesty}
\newtheorem{example}{Example}

\newcommand{\strat}{\sigma}
\newcommand{\Ags}{\mathit{Ags}} 
\newcommand{\Agse}{\mathit{Ags}^+} 
\newcommand{\I}{\mathcal{I}} 
\newcommand{\atlop}[1]{\langle\hspace{-2pt}\langle #1 \rangle \hspace{-2pt}\rangle } 
\newcommand{\existsg}[1]{\exists #1.} 
\newcommand{\forallg}[1]{\forall #1.} 
\newcommand{\lid}[2]{\mathtt{e}_{#1}(#2)}

\newcommand{\atluop}[1]{[\hspace{-1pt}[ #1]\hspace{-1pt}]} 
 
\newcommand{\rimp}{\Rightarrow}

\newcommand{\skipp}{\mathit{skip}}  
\newcommand{\exploited}{\mathtt{exploited}}  
\newcommand{\cc}{\mathtt{cc}}  
\newcommand{\CCN}{\mathtt{CCN}}  
\newcommand{\done}{\mathtt{done}}

\newcommand{\Env}{E} 
\newcommand{\restrict}{\upharpoonright}

\newcommand{\Var}{V}

\newcommand{\Acts}{\mathit{Acts}}

\newcommand{\be}{\begin{enumerate}} 
\newcommand{\ee}{\end{enumerate}}

\newcommand{\nxt}{\Circle} 
\newcommand{\always}{\Box} 
\newcommand{\sometimes}{\Diamond}

\newcommand{\G}{{\cal G}} 

\newcommand{\detstrat}{\mathit{det}}

 \newcommand{\Iunif}{\I^\unif}

\newcommand{\sgy}{\alpha} 
\newcommand{\unif}{\mathit{unif}} 
\newcommand{\Strat}{\Sigma} 
  
  \newcommand{\Strats}{\Sigma}
    \newcommand{\Cont}{\Gamma}

\newcommand{\ETL}{{\rm ETLK}}

\newcommand{\ksubf}{{\rm maxk}}

\newcommand{\KSAT}{{\rm KSAT}}
\newcommand{\ASAT}{{\rm ASAT}}

\newcommand{\Crypt}{Cr}
\newcommand{\Low}{Lo}
\newcommand{\High}{Hi}

\title{An Epistemic Strategy Logic\footnote{This paper combines results 
from \cite{HM-SR14} {\em An epistemic strategy logic}, X. Huang and R. van der Meyden, 2nd International Workshop on Strategic Reasoning , April 2014, Grenoble, France, 
and \cite{HvdM2014} {\em A temporal logic of strategic knowledge}, X. Huang and R. van der Meyden, Int. Conf. on Principles of Knowledge Representation and Reasoning, Jul 2014, Vienna. 
It extends these works by including full proofs for all results. 
}}

\author[1]{Xiaowei Huang\thanks{Email: \href{mailto: xiaowei.huang@liverpool.ac.uk}{xiaowei.huang@liverpool.ac.uk}. Most of Huang's work was performed when he was at 
UNSW Sydney, Australia}}
\author[2]{Ron van der Meyden\thanks{Email: \href{mailto: meyden@cse.unsw.edu.au}{meyden@cse.unsw.edu.au}}}
\affil[1]{The University of Liverpool, UK}
\affil[2]{UNSW Sydney, Australia}

\date{
}

\begin{document}

\maketitle

\begin{abstract}
This paper presents an extension of temporal epistemic logic with operators that quantify over agent  strategies. 
Unlike previous work on alternating temporal epistemic logic, the semantics works with systems whose 
states explicitly encode the strategy being used by each of the agents.  
This provides a natural way to express  what agents would know were they to be
aware of 
some of 
the strategies  being used by other agents. A number of examples 
that rely upon the ability to express an agent's knowledge about the strategies 
being used by other agents are presented 
to motivate the framework, including reasoning about game theoretic equilibria, 
knowledge-based programs, and information theoretic computer security policies.  Relationships to 
several variants of alternating temporal epistemic logic are 
discussed. The computational complexity of model checking the logic  and several of its 
fragments are also characterized.
\end{abstract}

\section{Introduction}

In distributed and multi-agent systems, agents typically  have a choice of actions to perform, and 
have individual and possibly conflicting goals. This leads agents to act strategically, attempting to select their actions over time so as to guarantee achievement of their goals even 
in the face of other agents' adversarial behaviour. The choice of actions generally needs to be 
made 
on the basis of {\em imperfect} information concerning the state of the 
system. 

These concerns have motivated the development of a variety of modal logics that aim to capture aspects of such settings. 
One of the earliest, dating from the 1980's, was multi-agent \emph{epistemic logic} \cite{HM90,RamParikh84}, which 
introduced
modal operators that deal with imperfect information 
by providing a way to state what  agents {\em know}. Combining such constructs with temporal logic  constructs \cite{Pnueli1977}
gives \emph{temporal epistemic logics}, which support reasoning about how agents' knowledge changes over time. 
Temporal-epistemic logic is an area about which a significant amount is now understood \cite{FHMVbook}. 

Logics dealing with  reasoning about strategies, which started to be developed in the same period \cite{Parikh83}, 
had a slower initial start, but have in recent years become the focus of intense study \cite{Pauly2002,Horty2001,ATLJACM}. 
 {\em Alternating temporal logic} (ATL) \cite{ATLJACM}, which generalizes branching-time 
temporal logic to encompass reasoning about the temporal effects of strategic choices by one group of agents against
all possible responses by their adversaries,  has become a popular basis for work in this area. 

One of the ways in which recent work has extended ATL is to add epistemic operators, yielding an 
{\em alternating temporal 
epistemic logic}, 
e.g., ATEL \cite{ATEL}. 
Many subtle issues  arise concerning  what agents know in  settings 
where multiple agents act strategically. In the process of understanding these issues, there has been a proliferation 
of epistemic extensions of ATL
\cite{Jonker2003,
vOJ2005,Jamroga2003,JA07}.  
Some of the modal operators introduced in this literature are complex, interweaving ideas about the 
knowledge of a group of agents, the strategies available to them,  the effect of playing these strategies 
against strategies available to agents not in the group, 
and the knowledge that other groups of agents may have about these effects. 

Our contribution in this paper is to develop a logic that extends the expressive power
of previous work on logics for knowledge and strategies, while at the same time 
simplifying the syntactic basis by 
identifying a 
small 
set of primitives 
that can be composed to represent the more complex 
constructs for reasoning about 
strategies and knowledge
from prior literature.  
We present examples to show that the logic is useful for a range of
applications, including expressing notions of information
flow security (such as strategic notions of noninterference
and erasure policies), reasoning about implementations of knowledge-based programs, 
and reasoning about game theoretic equilibria. 
We also conduct a detailed analysis of the complexity of model checking a number of 
fragments of the logic. 
Our semantic framework is able to model a range of semantics for knowledge and strategies 
including a ``perfect recall'' interpretation, but since we are interested in model checking complexity
results at the lower end of the complexity spectrum we concentrate on an ``imperfect recall'' or 
``observational'' semantics of knowledge.  (We note that model checking just ATL, even without 
knowledge operators, under an imperfect information and perfect recall semantics is already 
undecidable \cite{DimaATLpr}.)

At the semantic level, the key way in which our logic extends prior work on alternating temporal epistemic logic is by treating 
agents' strategies as first class citizens in the semantics, 
represented as components of the global state of the system at any
moment of time
in a run of the system.  
This is in contrast to 
most prior work 
in the area, where 
strategies are used to generate runs of a system, but the runs
themselves contain no explicit information about the specific strategies used
by the players to produce them. 
Our approach provides a referent for the notion ``the strategy being used by player $i$ '', 
which cannot be expressed in most prior works on alternating temporal epistemic logic. 

We reflect this additional referent at the syntactic level by 
introducing a syntactic notation $\sigma(i)$, that refers to the strategy of agent $i$. 
Since the strategy of agent $i$ is modelled semantically as a component of the global state, 
just like the local state of agent $i$, 
we allow this construct to be used in the same contexts where the agent name $i$ can be used
--- in particular, in operators for knowledge (including distributed and common knowledge). 
An example of what can be expressed with this extension is $D_{\{i,\strat(i)\}} \phi$, which says
that the truth of $\phi$ in all possible futures can be deduced from knowledge of agent $i$'s local state plus the  
strategy being applied by agent $i$. Intuitively, the construct $D_{\{i,\strat(i)\}} $ captures what agent $i$ knows
when it takes into account the strategy it is running.  

We show that 
this 
extension of temporal epistemic logic 
gives a logical approach with broad applicability. 
In particular, 
as we show in Section~\ref{sec:atltrans}, temporal epistemic logic extended with the
indices $\strat(i)$ can express alternating temporal logic constructs 
(both revocable and irrevocable). 
The extension can also 
express many of the subtly different notions that have been 
proposed in the literature on 
alternating temporal epistemic logics. We demonstrate
this 
(in Section~\ref{sec:atelvarn})
by results that show how 
a number of 
such logics can be translated into our setting. 
We also present a number of other 
applications 
including
game theoretic  solution concepts (Section~\ref{sec:games}), 
issues of concern in computer security (Section~\ref{sec:erasure}), 
and  reasoning about 
possible implementations of knowledge-based programs (Section~\ref{sec:KBP}). 

In some 
applications, 
however, some richer expressiveness is required. 
One such application, discussed in Section~\ref{sec:atelvarn},
concerns expressing an operator, combining common knowledge and strategic concerns, 
from an extended alternating temporal epistemic logic of Jamroga and van der Hoek \cite{JvdH2004}.
We address this by 
adding 
to the logic constructs that can be used to express  quantification over strategies.
This leads to a logic which, like~{\em strategy logic} \cite{CHP10,MogaveroMV10}, supports explicit naming and quantification over strategies.
Technically, 
we achieve this in a slightly more general way: we first generalize temporal epistemic logic to include operators 
$\exists x$ for quantification over global states $x$, as well as statements $\lid{i}{x}$ which say that 
component $i$ in the current global state is the same as component $i$ in the global state denoted by $x$. 
Even before the introduction of strategic concerns, this gives a novel extension of temporal epistemic logic 
in the flavour of \emph{hybrid logic} \cite{BS98}. 
(As we show in Section~\ref{sec:etl} this extension enables the expression of 
security notions such as \emph{nondeducibility} \cite{sutherland_86} that cannot be naturally expressed in 
standard temporal epistemic logics.) 
We then apply this generalization to 
a system that includes strategies encoded in the global states and references these using the ``strategic'' 
indices $\sigma(i)$. 
The resulting 
logic can express that agent $i$ knows what strategy agent $j$ is using, by means of the formula 
$$\exists x(\lid{\strat(j)}{x} \land K_i \lid{\strat(j)}{x})$$ 
in which the first occurrence of $\lid{\strat(j)}{x}$ binds $x$ to a global state in which the strategy of agent  $j$ is the same  
as at the current state, and the remainder of the formula states that every global state considered possible by agent $i$ has
the same strategy for agent $j$. (This cannot be expressed in most alternating temporal epistemic logics, e.g., ATEL \cite{ATEL}, 
since their  semantics fails to encode the strategy being run by an agent in the locus of evaluation of formulas.)
The framework 
is able to express the above-mentioned operator from \cite{JvdH2004}, as well as  notions
of information flow security that quantify over agent strategies, such as \emph{nondeducibility on strategies} \cite{WJ90}, 
which we discuss in Section~\ref{sec:nonded}.  

The main theoretical contribution of the paper is a set of results on the complexity of model checking the 
resulting logic, and its fragments. 
We consider several dimensions: does the logic have 
quantifiers, and what is the temporal basis for the logic: branching-time (CTL) or 
linear time (LTL). The richest logics in our spectrum turn out to have 
EXSPACE-complete model checking problems. 
However, we identify a 
number of special cases where model checking is in PSPACE, 
i.e., no more than the complexity of model checking the temporal logic LTL. 
One is the fragment where we allow the constructs $\exists x$ and $\lid{i}{x}$, but 
restrict the temporal operators to be those of the branching-time logic $\CTL$. 
Another is the fragment in which we do not allow $\exists x$ and $\lid{i}{x}$, 
but allow strategic indices $\strat(i)$ in knowledge operators and 
take the temporal operators from the richer branching-time logic $\CTLs$, which extends 
the linear time logic $\LTL$.

The structure of the paper is as follows. 
In Section~\ref{sec:etl}, we first develop an extension of temporal epistemic logic that adds the
ability to quantify over global states and refer to global state components. 
We then present a semantic model 
for the environments in which 
agents choose their actions. Building on this model, we show how 
to construct a model  for temporal epistemic logic 
called {\em strategy space} 
in which runs build 
in information 
about the strategy being used by each of the agents. 
We then define a spectrum of logics defined over the resulting semantics. 
These logics are obtained as fragments of the extended temporal epistemic logic, interpreted in 
 strategy space. 
Section~\ref{sec:apps} deals with applications of the resulting logics. 
In particular, we show that the logics can express reasoning about 
implementations of knowledge-based programs, 
many notions that have been proposed in the area of 
alternating temporal epistemic logic, game theoretic solution 
concepts, and problems from computer security. 
Next, in Section~\ref{sec:mc}, we provide results on the complexity of the model checking problem for the 
various fragments of the logic, identifying fragments with lower complexity than the general problem.  
In Section~\ref{sec:concl}, we conclude with a discussion of related literature.

\section{An extended temporal epistemic logic}\label{sec:etl}

The usual {\em interpreted systems} semantics for temporal epistemic logic \cite{FHMVbook}
deals with runs,  in which each moment of time is associated with a global state that 
is comprised of a local state for each agent in the system. 
We begin by defining the syntax and semantics of an extension of temporal epistemic logic that adds the
ability to quantify over global states and refer to global state components. 
This syntax and semantics will be instantiated in what follows by taking some of the global state 
components to be the strategies being used by agents. 

To quantify over global states, we extend temporal epistemic logic with a set of variables $\SVar$, 
a quantifier $\exists x$ 
and a construct $\lid{i}{x}$, 
where $x$ is a variable. The formula  $\existsg{x}\phi$ says, intuitively, that there exists in the system a global state $x$ such that $\phi$ (a formula that 
may contain uses of the variable $x$) holds
at the current point. The formula $\lid{i}{x}$ 
asserts the equality of the local states of agent $i$ at the current point and in the global state $x$. 

Let $\Prop$ be a set of atomic propositions and let $\Ags$ be a 
finite 
set of agent 
names, excluding the special name $e$, which we use to designate the environment in which the agents operate. 
We write $\Agse$ for the set $\{e\} \cup \Ags$. 
The 
language $\ETL(\Ags,  \Prop,\SVar)$
(or just $\ETL$ when the parameters are obvious) 
has syntax given by the grammar: 
$$\phi \equiv p ~|~\neg \phi~|~\phi_1 \lor \phi_2~|
~A\phi ~|~\nxt\phi ~|~\phi_1 \until \phi_2 ~
|~ \existsg{x}\phi~ |~ \lid{i}{x}~ | ~D_G\phi~|~C_G\phi$$ 
where $p \in \Prop$, $x\in \SVar$, 
$i\in \Agse$, and $G\subseteq \Agse$. 
The construct $D_G\phi$ expresses that agents in $G$ have  distributed knowledge of $\phi$, i.e., could deduce $\phi$ if they pooled their information, 
and $C_G\phi$ says that $\phi$ is common knowledge to group $G$. 
The temporal formulas $\nxt\phi$, $\phi_1 \until \phi_2$, $A\phi$ have the same intuitive meanings as in the temporal logic $CTL^*$ \cite{EH1986}, i.e., 
$\nxt \phi$ says that $\phi$ holds at the next moment of time, $\phi_1 \until \phi_2$ says that $\phi_1$ holds until $\phi_2$ does,  and
$A\phi$ says that $\phi$ holds in all possible evolutions from the present situation. 

Other operators can be defined in the usual way, e.g., $\phi_1\land \phi_2 =  \neg (\neg \phi_1\lor \neg \phi_2)$, $\sometimes\phi = (true\until \phi) $, 
which says that $\phi$ holds eventually, 
 $\always\phi = \neg \sometimes \neg \phi$, 
 which says that $\phi$ always holds, 
 $E\phi=\neg A\neg \phi$, which says that $\phi$ holds on some path from the current point, 
 etc. 
The universal form $\forallg{x}\phi = \neg \existsg{x} \neg \phi$ expresses that $\phi$ holds
for all global states  $x$
that occur in the system.  
For an agent 
$i \in \Agse$, 
we  write $K_i \phi$ for $D_{\{i\}} \phi$ ---  
this expresses that  agent $i$ knows the fact $\phi$. 
The notion of everyone in group $G$ knowing $\phi$ can then 
be expressed as $E_G\phi = \bigwedge_{i\in G} K_i \phi$. 
We write  $\lid{G}{x}$ for  $\bigwedge_{i\in G} \lid{i}{x}$. 
This says that at the current point, 
the agents in $G$  have the same local state as they do at 
the global state named by variable $x$.

We will be interested in a fragment of the logic that restricts the occurrence of the temporal operators
to some simple patterns, in the style of the branching-time temporal logic CTL \cite{CES1986}. 
We write $\ECTL(\Ags,  \Prop,\SVar)$ 
(or just $\ECTL$ when the parameters are obvious)
for 
the fragment of 
the language
$\ETL(\Ags,  \Prop,\SVar)$ in which the temporal operators occur only in the particular forms
$A\nxt \phi$, $E\nxt \phi$, $A\phi_1 \until \phi_2$, and $E \phi_1 \until \phi_2$. 
In the context of temporal logic, these restrictions reduce the complexity of model checking 
from PSPACE to PTIME \cite{CES1986}. It is therefore interesting to study the impact on complexity of a similar restriction in 
the context of our additional operators.

The semantics of $\ETL(\Ags,  \Prop,\SVar)$ builds straightforwardly on 
the following definitions used 
in the standard semantics for temporal epistemic logic \cite{FHMVbook}. 
Consider a system for a set of agents $\Ags$. 
A {\em global state} is an element of the set 
$\G = \Pi_{i\in \Agse} L_i$, 
where  $L_e$ is a set of states of the environment and 
$L_i$ is a  set of {\em local states} for each agent $i\in \Ags$. 
A {\em run} is a mapping $r: \nat\rightarrow \G$ giving a global state at each moment of time. 
For $n \leq m$,   write $r[n\ldots m]$ for the sequence $r(n) r(n+1)  \ldots r(m)$. 
We also write $r[n\ldots]$ for the infinite sequence $r(n) r(n+1)  \ldots $. 
A {\em point} is a pair $(r,m)$ consisting of a run $r$ and a time 
$m\in \nat$.
An {\em interpreted system} is a pair $\I = (\R, \pi)$, where 
$\R$ is a set of runs and $\pi$ is an  {\em interpretation},  mapping each point $(r,m)$ with $r\in \R$ 
to a subset of  $\Prop$. 
Elements of $\R\times \nat$ are called the {\em points} of $\I$. 
For each 
$i\in \Agse$, 
we write $r_i(m)$ for the 
corresponding
component of $r(m)$ in $L_i$, and then
define an equivalence relation on points by $(r,m) \sim_i (r',m')$ if 
$r_i(m) = r'_i(m')$. 
We also define $\sim^D_G\equiv \cap_{i\in G}\sim_i$, and $\sim^E_G \equiv \cup_{i\in G} \sim_i$, and $\sim^C_G \equiv (\cup_{i\in G} \sim_i)^*$  for $G\subseteq \Ags$, 
where $*$ denotes the reflexive transitive closure of a relation. 
We take $\sim^D_\emptyset$ to be the universal relation on points, and 
(for the sake of preserving monotonicity of these relations in these degenerate cases) take 
$\sim^E_\emptyset$ and $\sim^C_\emptyset$ to be the 
identity relation.

To extend this semantic basis for temporal epistemic logic to a semantics for 
$\ETL$, 
we just need
to add a construct that interprets variables as global states. 
A {\em context} for an interpreted system $\I$  is a mapping $\Cont$ from $\SVar$ to global states occurring in $\I$, 
i.e., such that for all $x\in \SVar$ there exists a point $(r,m)$ of $\I$ such that $\Cont(x) = r(m)$. 
When $g$ is a global state and $x\in \Var$, 
we write $\Cont[g/x]$ for the context $\Cont'$ with $\Cont'(x) =g$ and $\Cont'(y) = \Cont(y)$ for all variables $y\neq x$. 
The semantics of the language \ETL\  
is given by a relation $\Cont, \I, (r,m) \models \phi$, representing that 
formula $\phi$ holds at point $(r,m)$ of the interpreted system $\I$, relative to context $\Cont$. 
This is defined inductively on the structure of the formula $\phi$, as follows: 
\begin{itemize}
\item $\Cont, \I, (r,m)\models p$ if   $p\in \pi(r,m)$; 
\item
$\Cont, \I, (r,m)\models \neg \phi$ if not $\Cont, \I, (r,m)\models \phi$; 

\item
$\Cont, \I, (r,m)\models \phi\wedge \psi$ if $\Cont, \I, (r,m)\models \phi$ and $\Cont, \I, (r,m)\models \psi$; 

\item
$\Cont, \I, (r,m)\models A\phi$ if  $\Cont, \I, (r',m)\models \phi$ for all $r' \in \R$ with 
$r[0\ldots m] = r'[0\ldots m]$; 

\item
$\Cont, \I, (r,m)\models \nxt\phi$ if  $\Cont, \I,  (r,m+1)\models \phi$; 

\item
$\Cont, \I, (r,m)\models \phi \until \psi$ if  there exists
$m'\!\geq\! m $ such that $\Cont, \I, (r,m')\models \psi$ and  $\Cont, \I, (r,k)\models \phi$ for all $k$ with $m\leq k < m'$;

\item
$\Cont, \I, (r,m)\models \existsg{x}\phi$ if $\Cont[r'(m')/x], \I, (r,m)\models\phi$ for some point $(r',m')$ of $\I$; 

\item
$\Cont,\I, (r,m)\models \lid{i}{x}$ if $r_i(m) = \Cont(x)_i$;

\item
$\Cont,\I, (r,m)\models D_G\phi$ if 
 $\Cont,\I, (r',m')\models \phi$ for all $(r',m')$ such that 
 $(r',m')\sim^D_G (r,m)$;

\item
$\Cont,\I, (r,m)\models C_G\phi$ if 
$\Cont,\I,(r',m')\models\phi$ for all $(r',m')$ such that $(r',m')\sim_G^C (r,m)$.
\end{itemize}


The definition is standard, except for the constructs  $\existsg{x}\phi$ and $\lid{i}{x}$. 
The clause for the former says that $\existsg{x}\phi$ holds at a point $(r,m)$ if 
there exists a global state $g=r'(m')$ such  
that  $\phi$ holds at the  point $(r,m)$, 
provided we interpret $x$ as 
referring to $g$. Note that it is required that $g$ is attained at some point $(r',m')$, 
so actually occurs in the system $\I$. 
The clause for $\lid{i}{x}$ says that this holds at a point 
$(r,m)$ if the local state of agent $i$, 
i.e., $r_i(m)$, 
is the same as the 
local state $\Cont(x)_i$ of agent $i$ at the global state $\Cont(x)$ 
that interprets the variable $x$ according to $\Cont$.

We remark that these novel constructs introduce some redundancy, in that
the set of epistemic operators $D_G$ could be reduced to the ``universal'' operator $D_\emptyset$, 
since  $ D_G\phi \equiv \existsg{x} (\lid{G}{x} \land D_\emptyset (\lid{G}{x} \rimp \phi))$. 
Evidently, given the syntactic complexity of this formulation,  $D_G$ remains a useful notation.

\begin{example} \label{ex:nondeducibility} 
As an example of a property that can be naturally expressed in ESL, 
but not in most standard temporal epistemic logics (e.g., ESL minus the 
operators $\exists x$ and $\lid{i}{x}$), 
consider information flow security properties
in the spirit of \emph{nondeducibility} \cite{sutherland_86}. Suppose that there are two 
agents $\High$ and $\Low$, representing two information 
security levels High and Low respectively. The High level contains secrets that need to be protected 
from an attacker, represented by the Low level. Nondeducibility security properties, intuitively, assert 
that $\Low$ always has no information  about $\High$. When the information that needs to be protected
is represented in the local state of $\High$, this means that $\Low$ should always consider all 
local states of $\High$ possible. This can be expressed using the formula 
$$\always( \neg \exists x( K_{\Low} (\neg  
\lid{\High}{x}
))~.$$ 
Here,
$  K_{\Low} (\neg  \lid{\High}{x})$ 
expresses that 
$\Low$ has some information about $\High$, because there exists 
some local state of $\High$ that $\Low$ is able to exclude, namely,  the local state $g_{\High}$ where $g$ is the global state
denoted by $x$.  By asserting that it is 
always 
the case
 that there does not exist such a state $x$ whose $\High$-local component $\Low$ is able to exclude, 
we say that $\Low$ never has information about $\High$. 
Equivalently, pushing the outer negation inwards gives the form 
$\always( \forall x( \neg K_{\Low} (\neg  
\lid{\High}{x}
))$
which says that $\Low$ always considers all local states of $\High$ to be possible. 


We remark that the operators $\exists x$ and $\lid{i}{x}$ may be eliminated
from the above formula if the system $\I$ is known, and has a sufficiently 
rich set of atomic propositions that each local state $h$ of 
$\High$
is associated with a 
conjunction $\phi_h$ of literals that is true exactly at global states $g$ with 
$g_{\High}=h$.  
Let $L_{\High}$ be the set of local states of $\High$. 
This gives the equivalence 
   $$ \exists x( K_{\Low}( \neg \lid{\High}{x}))  ~~\equiv~~ \bigvee_{h \in L_{\High}}  K_{\High} ( \neg  \phi_h) $$ 
which is valid in $\I$. However, if the system $\I$  over all systems, then no single 
formula of the logic without the operators $\exists x$ and $\lid{i}{x}$ can 
be equivalent to 
$\exists x( K_{\Low}( \neg \lid{\High}{x}))$, 
because a fixed set of propositions cannot distinguish an arbitrarily large set of states. 
\qed
\end{example}

\subsection{Strategic Environments}\label{sec:stratenv}

In order to semantically represent settings in which  agents operate  by strategically choosing their
actions, 
we introduce 
\emph{environments}, 
a type of transition system that models the available actions
and their effects on the state.  
This modelling is long established in the literature on reasoning about knowledge \cite{MeydenTARK96}, 
and is similar to models used in the tradition of alternating temporal logic \cite{ATLJACM}. 
From an environment and a class of strategies, we construct
an instance of the 
interpreted systems semantics defined in the previous section. 
One of the innovations in this
construction is to introduce new names, that refer to global state components that
represent the strategies being used by the agents. 

An {\em environment} for agents $\Ags$  is a tuple 
$\Env =  \langle S, I, \{\Acts_i\}_{i\in \Ags}, \trans, \{O_i\}_{i\in \Ags}, \pi\rangle$, where 
\be \item 
$S$ is a set of states, 
\item 
$I$ is a subset of $S$, representing the initial states, 
\item 
for each $i \in \Ags$, component $\Acts_i$ is a  nonempty set of actions
that may be performed by agent $i$; 
we define $\Acts = \Pi_{i\in Ags} \Acts_i$ to be the set of 
joint actions, 
\item 
$\trans \subseteq S \times \Acts \times S$ is a transition relation, 
labelled by joint actions, 
\item 
for each $i\in \Ags$, 
component $O_i: S\rightarrow L_i$ is 
an observation function, 
and 
\item 
$\pi: S\rightarrow \mathcal{ P}(\Prop)$ is a propositional assignment. 
\ee
Here the range $L_i$ of the observation function $O_i$ is any set, what will matter in the semantics
is an equivalence relation derived from this function. 

An environment is said to be finite if all its components, i.e., $S, \Ags, \Acts_i, L_i$ and $\Prop$, are finite. 
Intuitively, a joint action $a\in \Acts$ represents a choice of action $a_i\in \Acts_i$ for each agent $i\in \Ags$, 
performed simultaneously, and the transition relation resolves this into an effect on the 
state. We assume that $\trans$ is serial in the sense that for all $s\in S$ and 
$a \in \Acts$ there exists $t\in S$ such that $(s,a, t)\in\trans$. 
We also write $s\ptrans{a}t$ for $(s,a, t)\in\trans$.

\begin{example} \label{ex:environment} 
We describe an environment for a secure message transmission 
problem, which models a sender agent $HS$ at a High security level that has a bit of  information to be  transmitted 
to a receiver agent $HR$, also at a High security level, via a channel represented by an agent $\Low$ 
at the Low security level (e.g., the internet). The transmission is 
handled by an agent $\Crypt$ that models cryptography that may be applied to the message before transmission. 
Thus, we take $\Ags = \{HS,\Crypt, HR,\Low\}$. The environment has the following components: 
\begin{itemize}
\item The set of states $S$ is the set of assignments to the following variables: 
\begin{itemize} 
\item $\mathtt{s}$, representing the sender's secret bit, with value in $\{0,1\}$
\item $\mathtt{k}$, representing a secret encryption key, with value in $\{0,1\}$
\item $\mathtt{c}$, representing the unsecured communication channel, with value in $\{0,1,\bot\}$
\end{itemize} 
We represent a state in $S$ in the format $\langle s,k,c\rangle$, corresponding to the values of the three variables. 

\item The set $I$ of initial states is the set $\langle s, k, \bot\rangle $ where $s,k\in \{0,1\}$. 
That is, the value of the channel $\mathtt{c}$ is initially $\bot$, representing that no 
message has yet been sent.  

\item We associate the following sets of actions with the agents: 
$\Acts_{HS} = \Acts_{HR} = \Acts_{\Low} = \{ \mathtt{skip} \}$ and 
$\Acts_{\Crypt} = \{ ~\mathtt{c}:=\mathtt{s} \oplus \mathtt{k}, ~ \mathtt{c}:=\mathtt{s} \oplus \overline{\mathtt{k}}~\}$. 
Thus, agents $HS,HR, \Low$ are inert, they can perform only the action $ \mathtt{skip}$, which has
no effect on their local states. The only active agent is $\Crypt$, which 
has two actions, each of which encrypts the message bit $\mathtt{s}$ using the key $k$ and 
places the result in the channel $\mathtt{c}$. Encryption is done 
by computing the exclusive-or $\oplus$ of the message with information from the key. 
The two actions correspond to taking the information from the key to be either 
the key bit $\mathtt{k}$ itself, or its complement $\overline{\mathtt{k}}$. 
Since agents $HS, HR, \Low$ always perform skip, we may, for brevity,  name 
joint actions using the action names for agent $\Crypt$, 
i.e., if $a$ is one of $\Crypt$'s actions, then we denote a joint action $\langle \mathtt{skip}, a, \mathtt{skip}, \mathtt{skip}\rangle$ 
in $\Acts = \Acts_{HS}\times \Acts_{\Crypt}\times \Acts_{HR}\times \Acts_{\Low}$ as just $a$. 

\item The transition relation resolves joint actions denoted as $a\in \Acts_{\Crypt}$ as follows: 
$$\langle s,k,c\rangle \ptrans{a} \langle s',k',c'\rangle$$ 
if either $a$ is $\mathtt{c}:=\mathtt{s} \oplus \mathtt{k}$ and $s' = s$, $k'=k$ and $c' = s\oplus k$, 
or  $a$ is $\mathtt{c}:=\mathtt{s} \oplus \overline{\mathtt{k}}$ and $s' = s$, $k'=k$ and $c' = s\oplus \overline{k}$. 

\item We define the observation functions for each of the agents on states $\langle s,k,c\rangle\in S$ as follows: 
\begin{itemize} 
\item Agent $HS$ observes just the bit to be transmitted, i.e., $O_{HS}(\langle s,k,c\rangle) = s$. 
\item Agent $\Crypt$ observes both the bit to be transmitted and the value of the encryption key, i.e., $O_{HS}(\langle s,k,c\rangle) = \langle s,k\rangle$. 
\item Agent $HR$ observes the communication channel and the value of the encryption key  i.e., $O_{HR}(\langle s,k,c\rangle) = \langle k,c\rangle$. 
\item Agent $\Low$ observes just the communication channel,  i.e., $O_{\Low}(\langle s,k,c\rangle) = c$. 
\end{itemize} 

\item  We do not need any propositions in our later uses of this environment, 
so we take $\Prop = \emptyset$ and $\pi : S\rightarrow \Prop$ to be the trivial assignment. \qed
\end{itemize} 

\end{example}


A {\em strategy} for agent $i\in \Ags$ in an environment $\Env$
is a function $\sgy_i: S \rightarrow \mathcal{ P}(\Acts_i) \setminus \{ \emptyset\}$, 
selecting a 
nonempty 
set of actions of the agent at each state.%
\footnote{More generally, 
a strategy could be a function of the history, but we focus here
on strategies that depend only on the final state.} 
We call these 
actions
{\em enabled} at the state
for agent $i$. 
A {\em group strategy}, or {\em strategy profile}, for a group $G$ is a 
tuple $\sgy_G = \langle \sgy_i\rangle_{i \in G}$ where each $\sgy_i$ is a strategy for agent $i$. 
A \emph{joint} strategy is a group strategy for the group $\Ags$ of all agents. 
If $\sgy= \langle \sgy_i\rangle_{i \in G}$ is a group strategy for group $G$, and $H\subseteq G$, we write 
$\sgy\restrict H$ for the restriction $ \langle \sgy_i\rangle_{i \in H}$ of $\sgy$ to $H$. 

A strategy $\sgy_i$ for agent $i$ is {\em deterministic} if $\sgy_i(s)$ is a singleton
for all $s$. 
 A strategy $\sgy_i$ for agent $i$ is {\em uniform} if for all states $s,t$, if $O_i(s) = O_i(t)$, then 
 $\sgy_i(s)  = \sgy_i(t)$. 
 Intuitively, uniformity captures the constraint that agents' actions are 
 chosen using no more information than they obtain from their observations.%
 \footnote{Recall that we work in this paper with agents with imperfect recall. For agents with 
 perfect recall, we would use a notion of uniformity that allows agents choice of action to 
 depend on all their past observations.} 
 A strategy $\sgy_G = \langle \sgy_i\rangle_{i \in G}$ for a group $G$ 
 is {\em locally uniform (deterministic)} if $\sgy_i$ is uniform (respectively, deterministic) for each agent $i\in G$.%
 \footnote{We prefer the term ``locally uniform'' to just ``uniform'' in the case of groups, 
 since we could say a strategy $\sgy$ for group $G$ is {\em globally uniform} if for all states $s,t$, if $O_i(s) = O_i(t)$ for all $i\in G$, then 
$\sgy_i(s)  = \sgy_i(t)$ for all $i\in G$.  
While we do not pursue this in the present paper, this notion would be interesting in 
settings where agents share information to collude on their choice of move.} 
 Given an environment $\Env$,   we write
 $\Strat^\detstrat(\Env)$ for the set of deterministic 
 joint 
 strategies, 
 $\Strat^\unif(\Env)$ for the set of all 
 locally uniform joint strategies, and  $\Strat^{\unif, \detstrat}(\Env)$ for the set of all deterministic 
 locally uniform joint strategies.   
 
\begin{example} \label{ex:strategy} 
We present some joint strategies in the environment of Example~\ref{ex:environment}. 
For agents $i \in \{HS,HR, \Low\}$, the only available action is $\mathtt{skip}$, so
all joint strategies $\alpha$ have $\alpha_i(s) = \{\mathtt{skip}\}$ for all $s\in S$. 
Thus, each joint strategy $\alpha$ is determined by its component $\alpha_{\Crypt}$, the strategy of the 
encryption agent. 

The encryption agent  could always choose the action $\mathtt{c} := \mathtt{s}\oplus \mathtt{k}$, giving 
the strategy $\alpha_{\Crypt}^0$ defined by $\alpha_{\Crypt}^0(s) = \{\mathtt{c} := \mathtt{s}\oplus \mathtt{k}\}$ for all states $s$. 
This strategy is both locally uniform and deterministic. 

If the encryption agent chooses its action non-deterministically, 
we have the strategy $\alpha_{\Crypt}^1$ defined by $\alpha_{\Crypt}^1(s) = \Acts_{\Crypt}$ for all states $s$. 
This strategy is locally uniform, but not deterministic. 

An alternate strategy for the encryption agent is to choose its action based on the 
values it observes. Consider the strategy $\alpha_{\Crypt}^2$ defined by letting $\alpha_{\Crypt}^2(\langle s,k,c\rangle)$ be the singleton set 
$\{\mathtt{c} := \mathtt{s}\oplus \mathtt{k}\}$ if $k=0$ and 
the action $\{\mathtt{c} := \mathtt{s}\oplus \overline{\mathtt{k}}\}$ 
otherwise. This strategy is deterministic. Also, since the value $k$ is always part of
the agent's observation, this strategy is locally uniform. \qed 
\end{example}

\subsection{Strategy Space}
 \label{sec:generatedIS}
 
 We now define an interpreted system, 
called  the \emph{strategy space}  of an environment, 
 that contains all the possible runs generated 
when agents $\Ags$ behave by choosing a strategy 
 from some set $\Sigma$ of joint strategies in the context of an environment $\Env$. 
 To enable reference to the strategy being used by agent $i\in \Ags$, we 
  introduce  the notation 
 ``$\strat(i)$'' as a name referring to agent $i$'s strategy. 
 For $G\subseteq \Ags$, we write $\strat(G)$ for the set  $\{\strat(i)~|~ i\in  G\}$. 

 Technically, $\strat(i)$ will be treated as 
 if it were 
 an agent in the context of 
 temporal epistemic logic, 
 in the sense that it will be the index of a local state component of the global state. 
 In particular, we 
  take the value of the local state
 at index 
 $\sigma(i)$ to be the strategy in use by agent $i$. 
 We will permit use of the 
 indices 
 $\sigma(i)$ in epistemic operators. 
This provides a way to refer, using distributed knowledge operators $D_G$ where 
$G$ contains the 
strategic indices
$\strat(i)$,  to 
what agents would know, should they take into account not
just their own observations, but also information about other agents' strategies.
For example, the distributed knowledge 
operator $D_{\{i,\strat(i)\}}$ captures the knowledge that  agent $i$ has, 
taking into account the strategy that it is running.  
Operator
$D_{\{i,\strat(i), \strat(j)\}}$ captures what agent $i$ would know, 
taking into account its own strategy and the strategy being used by agent $j$. 
Various applications of the usefulness of this expressiveness are given in 
Section~\ref{sec:apps}. 

 We note, however, that 
 unlike the base agent $i \in \Ags$, the 
 index 
 $\strat(i)$ is not
 one of the agents in the environment $\Env$, and it is not associated with any actions. 
 The 
 index 
 $\strat(i)$ exists only in the interpreted system that we generate from $\Env$. 
(A similar remark applies to the special  agent $e$, which is also 
not associated with any actions.) 
Since the indices $\strat(i)$ are not agents in the same sense as agents $i \in \Ags$, the  
reader may prefer to read $D_G\phi$ with $\strat(i)\in G$ as ``$\phi$ is \emph{deducible} from the 
information contained in state components $G$'' rather than the more standard 
``it is distributed knowledge to agents $G$ that $\phi$''.

 Formally,
 suppose we are 
  given an environment 
 $\Env =  \langle S, I, \{\Acts_i\}_{i\in \Ags}, \trans, \{O_i\}_{i\in \Ags}, \pi\rangle$ for agents $\Ags$, 
 where $O_i : S\rightarrow L_i$ for each $i\in \Ags$, 
 and a set $\Strat \subseteq \Pi_{i\in \Ags} \Strat_i$ of joint strategies for the group $\Ags$.  
 We define the {\em strategy space} interpreted system  $\I(\Env,\Strats) = (\R, \pi')$
 as follows%
 \footnote{The construction given here is for an ``observational'' or ``imperfect recall'' modelling of knowledge that assumes that an agent
 reasons, and chooses its next action, on the basis of its current observation only. It is straightforward to give other constructions such
 as a synchronous perfect recall semantics, where we work with the sequence of observations and  actions of the agent instead. Model checking for such a 
variant would be undecidable, so we do not pursue this here.}. 
The system $\I(\Env,\Strats)$ has global states $\G = S\times \Pi_{i\in \Ags} L_i \times \Pi_{i\in \Ags} \Sigma_i$. 
Intuitively, each global state consists of a state of the environment $E$, a local state for each agent $i$ in $E$, and a strategy for each agent $i$.   
We index the components of this cartesian product by $e$, the elements of $\Ags$ and the elements of $\strat(\Ags)$, respectively. 
We take the set of runs $\R$ of $\I(\Env,\Strats)$ to be the set of all runs $r: \nat \rightarrow \G$
satisfying the following constraints, for all $m\in \nat$ and $i\in \Ags$
\be 
\item $r_e(0) \in I$ and $\langle r_{\strat(i)}(0) \rangle_{i\in \Ags} \in \Strats$, 
\item  $r_i(m) = O_i(r_e(m))$,
\item  $(r_e(m), a,  r_e(m+1)) \in \trans $ for some joint action $a\in \Acts$ such that for all $j\in \Ags$ we have 
$a_j \in \alpha_j(r_e(m))$,
where $\alpha_j = r_{\sigma(j)}(m)$, and 
\item $r_{\sigma(i)}(m+1) = r_{\sigma(i)}(m)$. 
\ee
The interpretation $\pi'$ of $\I(\Env,\Strats)$ is determined from the interpretation  
$\pi$ of $\Env$ by taking $\pi'(r,m) = \pi(r_e(m))$ for all points $(r,m)$. 

The first constraint
on runs says, intuitively, 
 that runs start at an initial state of $E$, and 
the initial strategy profile at time $0$ is one of the profiles in $\Strats$. 
The second constraint states that the agent $i$'s local state  at time $m$ is the observation 
that agent $i$ makes of the state of the environment at time $m$. The third constraint 
says that evolution of the state of the environment is determined at each moment of time by agents choosing 
an action by applying their strategy at that time to 
the state at that time. The joint action resulting from these individual choices is then  
resolved into a transition on the state of the environment using the transition relation from $\Env$.  
The final constraint says that agents' strategies are fixed during the course of a run. 
Intuitively, each agent picks a strategy, and then sticks to it.

Our epistemic strategy logic is now just  an instantiation of the extended temporal epistemic 
logic in the strategy space generated by an environment. That is, we start with an environment $\Env$
and an associated set of strategies $\Strats$, 
and then work with the language 
            $\ETL(\Ags\cup \strat(\Ags), \Prop, \SVar)$
in the interpreted system  $\I(\Env,\Strats)$. 
(Recall that this notation implicitly includes a local state component $e$ to represent the state of the environment.) 
We call this instance of the language $\ESL(\Ags, \Prop, \SVar)$, or just $\ESL$ when the 
parameters are implicit. 

Since interpreted systems are always infinite objects, we use environments to give a finite input for the model checking problem. 
For an environment $\Env$, a set of strategies $\Strat$ for $\Env$, and a context $\Cont$ for  $ \I(\Env, \Strat)$, 
we write $\Cont, E,\Strat\models \phi$ if  $\Cont, \I(\Env, \Strat),(r,0)\models \phi$ 
for all runs $r$ of  $\I(\Env, \Strat)$. 
Often, the formula $\phi$ will be a sentence, i.e., will have all variables $x$ in the scope of an operator $\exists x$. 
In this case the statement $\Cont, \Env,\Strat\models \phi$ is independent of $\Cont$ and we write simply $\Env,\Strat\models \phi$.

We will be interested in a number of fragments of $\ESL$ that turn out to have lower complexity. 
We define $\ESL^-(\Ags, \Prop, \SVar)$, or just $\ESL^-$, 
to be the language 
$$\ECTL(\Ags\cup \strat(\Ags), \Prop, \SVar)~.$$
Another fragment of the language that will be of interest is 
the  language, 
denoted
$$\CTLsK(
\Ags
\cup \strat(\Ags), \Prop, \SVar)$$ 
in which we omit the constructs $\exists$ and $\lid{i}{x}$; this is a standard branching-time temporal epistemic
language except that it contains the 
strategy indices 
$\strat(\Ags)$.


\section{Applications}\label{sec:apps}

We now consider a range of applications of the logic $\ESL$,  
and show how it can represent notions from earlier work on alternating temporal epistemic logic. 
(In a few cases, we prove precise translation results, but due to the large number of operators and distinct semantics underlying 
these logics in the literature, we just sketch intuitive correspondences in most cases.)  

\subsection{Variants of Nondeducibility} \label{sec:nonded} 

We already mentioned the notion of nondeducibility in Example~\ref{ex:nondeducibility}, 
which shows one way that our 
logic extends
the expressiveness of 
previous work on  temporal epistemic logic, by allowing 
quantification over agents' local states to be expressed. 
We continue discussion of this example here, in the context of the environment $\Env$ of 
Example~\ref{ex:environment}. We also show that our logic can represent a related notion from
 the security literature called \emph{nondeducibility on strategies} \cite{WJ90} that involves
an agent reasoning based not just on its local state, but
also using knowledge of the strategy being employed by another agent. 
This demonstrates a further dimension in which we can express more than prior work on 
alternating temporal epistemic logic, and shows the value of allowing the strategic indices 
$\strat(i)$ to occur in epistemic operators. 
(Our discussion in this section loosely follows examples used in \cite{WJ90}
to motivate nondeducibility on strategies.)  

Consider first the instance 
$$\mathit{NonDed}= \always( \neg \existsg{x}( K_{\Low} (\neg  \lid{HS}{x}))~$$ 
of the formula from Example~\ref{ex:nondeducibility},
which expresses that the low level attacker $\Low$ never learns any information about the high level secret 
held in the local state of the high sender $HS$. On the other hand, the formula
$$\mathit{Ded}(G) = \sometimes( \existsg{x}( D_{G} ( \lid{HS}{x}))~$$
 states that group  $G$ does eventually 
 learn the value of the secret held by $HS$.  (Note that the formula $D_{G} ( \lid{HS}{x})$ says that 
 the group $G$ has distributed knowledge that the local state of component $HS$ is the same as 
 the local state of $HS$ in the global state denoted by variable  $x$.
 The formulas $\mathit{NonDed}$ and $\mathit{Ded}(\{\Low\})$ are not 
 opposites,
 as one might expect from the names.
 Actually, the negation of $\mathit{NonDed}$ and $\mathit{Ded}(\{\Low\})$ lead to similar formulas, except that the former has a negation before $\lid{HS}{x}$.) 
  Clearly, for cryptography to be effective, 
 we require that the specification $ \mathit{NonDed} \land \mathit{Ded}(\{HR\})$
 be satisfied, which expresses that the High receiver $HR$ eventually learns the secret, 
 but that the adversary $\Low$ never has any information about the secret. 

In what follows, given a joint strategy $\alpha$, we write
$\Sigma(\alpha)$ for the singleton set of strategies $\{\alpha\}$. 

Suppose first that  encryption is always done 
using the action $\mathtt{c}:=\mathtt{s} \oplus \mathtt{k}$, 
so that the joint strategy is the strategy $\alpha^0$ from Example~\ref{ex:strategy}, with 
$\alpha^0_{\Crypt}(s) = \{\mathtt{c}:=\mathtt{s} \oplus \mathtt{k}\}$ for all states $s$. 
Then we work in the interpreted system generated by the set of strategies $\Sigma(\alpha^0) = \{\alpha^0\}$. 
Note that in $\I(E,\Sigma(\alpha^0))$, it is common knowledge that the strategy being used by $\Crypt$ is $\alpha^0_{\Crypt}$. 
The following result shows that in this case, the system satisfies the specification $\mathit{NonDed} \land \mathit{Ded}(\{HR\})$

\begin{propn} \label{prop:nonded} 
$\Env, \Sigma(\alpha^0) \models \mathit{NonDed} \land \mathit{Ded}(\{HR\})~.$ 
\end{propn} 
\begin{proof} 
We first show that $E, \Sigma(\alpha^0) \models \mathit{NonDed}$. 
Note that, since there is only one joint strategy, and agents' observations are derived from the 
state of the environment, a run $r$ of $\I(\Env, \Sigma(\alpha^0))$ is determined by the
sequence of states of the  environment $r_e[0\ldots] =  r_e(0),r_e(1) \ldots$. These sequences 
all have the form 
$$\langle s,k,\bot\rangle \langle s,k,s\oplus k\rangle^\infty$$ 
for some $s,k \in \{0,1\}$, where $t^\infty$ indicates infinitely many copies of the state $t$. 
For each run $r$ of this form, there exists another run $r'$ with 
$r'_e[0\ldots] = \langle \overline{s},\overline{k},\bot\rangle \langle \overline{s},\overline{k},\overline{s}\oplus \overline{k}\rangle^\infty$.
Now, we have that $(r,n) \sim_{\Low} (r',n)$ for all $n \in \nat$, since 
$$r_{\Low}(0)  = O_{\Low}(\langle s,k,\bot\rangle )  =  \bot = O_{\Low}(\langle \overline{s},\overline{k},\bot\rangle )  = r'_{\Low}(0)$$
and  
\begin{align*} 
r_{\Low}(n) &  = O_{\Low}(\langle s,k,s\oplus k \rangle )\\
&  =  s\oplus k \\
& =  \overline{s}\oplus \overline{k} \\
& = O_{\Low}(\langle \overline{s},\overline{k},\overline{s} \oplus \overline{k}\rangle ) \\
& =  r'_{\Low}(n)
\end{align*}  
for $n \geq 1$. 
Since also $(r,n) \sim_{\Low}(r,n)$, we have that $\Low$ considers both possible values of the local state of agent $HS$ possible, 
so $\I(\Env, \Sigma(\alpha^0)),(r,0) \models  \mathit{NonDed}$.

On the other hand, we have  $\I(\Env, \Sigma(\alpha^0)),(r,0) \models  \mathit{Ded}(\{HR\})$. For,  at time 1, we
have $r_ {HR}(1) = O_{HR}(\langle s, k, s\oplus k\rangle) = \langle k, s\oplus k\rangle $. 
Let $(r',m)$ be any point with $(r,1) \sim_{HR} (r',m)$, and let
$r'_e[0\ldots] = \langle s',k',\bot\rangle \langle s',k',s'\oplus k'\rangle^\infty$. Then $m \geq 1$ and 
$$r'_{HR}(m) =  O_{HR}(\langle s', k', s'\oplus k'\rangle) = \langle k', s'\oplus k' \rangle~.$$ 
Thus, from $r_{HR}(1) = r'_{HR}(m)$, we obtain 
$k = k'$ and $s\oplus k = s' \oplus k'$. Hence also 
 $ r'_{HS}(m) = s' =
 (s'\oplus k')\oplus k' = (s\oplus k)\oplus k = s =  r_{HS}(1) $. 
This shows that  $\I(\Env, \Sigma(\alpha^0)),(r,1)\models  \exists x( K_{HR} ( \lid{HS}{x}))$, 
so $\I(\Env, \Sigma(\alpha^0)),(r,0)\models  \mathit{Ded}(\{HR\}))$. 
\end{proof} 

On the other hand, not every strategy for the encryption agent similarly satisfies the specification. 
Consider the joint strategy $\alpha^2$ from Example~\ref{ex:strategy}.  Here we have 
that $\Low$ and $HR$ both always learn the value of the secret. 

\begin{propn} \label{prop:nonded2}
$\Env, \Sigma(\alpha^2)  \models \mathit{Ded}(\{HR\}) \land \mathit{Ded}(\{\Low\}) $. 
\end{propn} 

\begin{proof} 
Strategy $\alpha^2$ is deterministic. Note that if $k=0$ then $s\oplus k =s$, 
and if $k =1$ then $s\oplus \overline{k} =s$. Thus, the runs of $\alpha^2$ have sequence of environment states 
$r_e[0\ldots ] = \langle s,k,\bot \rangle \langle s,k,s \rangle^\infty$. 
As above, since $\Sigma(\alpha^2) $ is a singleton, this sequence determines the run as a whole. 
Since $O_{\Low}( \langle s,k,s \rangle) = s$ and $O_{HR}( \langle s,k,s \rangle) = \langle k,s\rangle$, 
both $\Low$ and $HR$ directly observe the value of the secret $s$ in the local state of $HS$ from time 1, so know this value. 
\end{proof} 

A corollary of this result is that if we work in a system where all (uniform) strategies for $\Crypt$ are 
possible (represented by the set of strategies $\Sigma^\unif$), then while $\Low$ cannot deduce the secret in general, there are encryption strategies
for $\Crypt$ such that, if $\Low$ knew that this strategy is being applied by $\Crypt$, then $\Low$ would be 
able to deduce the secret. 

\begin{propn} \label{prop:nonded3} 
$\Env, \Sigma^\unif \models  \mathit{NonDed}$, 
but not  $\Env, \Sigma^\unif  \models \neg \mathit{Ded}(\{\Low, \strat(\Crypt)\}) $. 
\end{propn} 

\begin{proof} 
For $\Env, \Sigma^\unif \models  \mathit{NonDed}$, we note that $\Low$ always considers it possible
that $\Crypt$ is running strategy $\alpha^0$ from above, and argue exactly as in Proposition~\ref{prop:nonded}. 
To show that not  $\Env, \Sigma^\unif  \models \neg \mathit{Ded}(\{\Low, \strat(\Crypt)\}) $, 
let  $r$ be a run in which $\Crypt$ runs strategy 
$\alpha^2_{\Crypt}$.  Note that if $(r,0) \sim_{\{\Low, \strat(\Crypt)\}} (r',m)$ then 
$r_{\strat(\Crypt)}(0)= r'_{\strat(\Crypt)}(m)$, i.e.,  $\Crypt$ uses the same strategy in the runs $r$ and $r'$. 
Essentially the same argument as applied in Proposition~\ref{prop:nonded2} 
to show that $\mathit{Ded}(\{\Low\})$ holds then shows that  $\I(\Env, \Sigma^\unif), (r,0) \models  \mathit{Ded}(\{\Low,  \strat(\Crypt)\})$.  
\end{proof} 

By means of a similar example, Wittbold and Johnson \cite{WJ90} argued that 
nondeducibility is too weak a notion of security to capture information flow security 
attacks in which the attacker exploits a covert channel in a system. Intuitively, 
it does not take into account that the attacker may have information about the 
strategies being used by other agents. One example of how such knowledge of another agent's strategy may arise in practice is
when the attacker $\Low$ has succeeded in infiltrating a virus (here represented by the strategy of $\Crypt$) into the system being
attacked (here comprised of components  $HS,HR$ and $\Crypt$,  i.e., the High sender, the High receiver, and the encryption agent, respectively).
When this is the case, a more appropriate modality for the attacker's knowledge is the modality $D_{\{\Low,  \strat(\Crypt)\}}$, 
which captures what $\Low$ can deduce when it also knows the strategy $\strat(\Crypt)$ being employed by $\Crypt$, 
rather than  the modality $D_{\{\Low\}}$ used in $\mathit{Ded}(\{\Low\})$ . 
(The modality $D_{\{\Low,   \strat(\Low), \strat(\Crypt)\}}$ which says that $\Low$ also reasons knowing its 
own strategy would also make sense in general, though in the model under discussion 
it is identical to $D_{\{\Low,  \strat(\Crypt)\}}$ since $\Low$ has only one action to choose from, so all its uniform strategies are the same.)  
Wittbold and Johnson's notion of \emph{nondeducibility on strategies} (NDS) is a definition of security that 
takes into account such reasoning by the attacker. 
For a two-agent system, comprised of Low level agent $\Low$ and High level agent $\High$, 
Wittbold and Johnson define a system to satisfy non-deducibility on strategies if 
every Low view is compatible with every High strategy. 
NDS may be expressed directly in our logic  by the formula%
\footnote{The perfect recall semantics in combination with 
perfect recall strategies would give the interpretation of this formula that is most adequate for 
security applications.}  
$$ D_\emptyset \forallg{x}(\neg K_{\Low}(\neg \lid{\strat(\High)}{x}))~.$$ 
which says that at all points of the system (identifying a $\Low$ view/local state, in particular) 
for all global states $x$  (identifying a High strategy, in particular), $\Low$ considers the High strategy in $x$ to be 
possible. 
This notion cannot be expressed in alternating temporal epistemic logics such as ATEL, 
discussed below, which do not allow reference to what can be deduced 
about other agents' strategies.

\subsection{Revocable and Irrevocable strategies in ATL}  \label{sec:atltrans} 

\emph{Alternating temporal logic} (ATL) \cite{ATLJACM} is a generalization
of the branching-time temporal logic CTL that can express the
capability of 
agents' strategies
 to bring about temporal effects.  
 We show in this section that ESL is able to express several variants of ATL. 
 The following section relates various epistemic extensions of ATL to ESL. 
 
The syntax of \ATL\ formulas $\phi$ is given as follows: 
$$\phi \equiv p ~|~\neg \phi~|~\phi_1 \lor \phi_2~|~\atlop{G}\nxt\phi ~|~~\atlop{G}\always\phi ~|~\atlop{G}(\phi_1 \until \phi_2) $$
where $p \in \Prop$, $i \in \Ags$ and $G\subseteq \Ags$. 
Essentially, each branching construct $A\phi$ 
of CTL is generalized  in ATL 
to an \emph{alternating} construct $\atlop{G} \phi$ for
a group $G$ of agents, where $\phi$ is a ``prefix temporal'' formula such as 
$\nxt \phi'$, $\sometimes \phi'$, $\always\phi'$ or 
$\phi_1 \until \phi_2$, as would be used to construct a CTL 
formula. 
Intuitively, $\atlop{G} \phi$ says that the group $G$ has a strategy for ensuring that $\phi$ holds, 
irrespective of what the other agents do. 

The semantics of ATL is given using \emph{concurrent game structures}, which are very similar to
environments as defined above, with the main differences being the following. For each point of difference, 
we sketch how to view concurrent game structures as equivalent to environments. 
\begin{itemize} 
\item Concurrent game structures lack a set of initial states. 
It is convenient for technical reasons to treat a concurrent
game structure as an environment with all of its states initial.

 \item Concurrent game structures allow that not all actions are available at every state, 
whereas in environments all actions are always available. In environments, we can treat a choice of a non-enabled
action as equivalent to a choice of a default enabled action in the transition relation. 

\item 
The transition relation in concurrent game structures is deterministic, in the sense that for each state $s$ and joint action $a$, there exists 
a unique state $t$ such that $s\ptrans{a} t$. Nondeterminism in environments can be modelled in concurrent game structures, 
by adding an agent that makes the nondeterministic choice through its actions. 

\item ATL's concurrent game structures do not have a notion of observation. Intuitively, all agents
always have perfect information concerning the current state. We may capture this in environments
by taking $O_i(s) =s$ for all agents $i$ and states $s$. 
\end{itemize} 

\newcommand{\GStrats}{\Delta}  

\newcommand{\trl}[1]{#1^*} 
 \newcommand{\compl}[1]{\mathit{comp}(#1)}

Using such correspondences, we  can express the ATL semantics in 
environments $\Env$ as follows. For reasons discussed below, we 
generalize the ATL semantics by parameterizing the definition on a set $\GStrats$ of strategies for groups of agents in the environment $\Env$.
That is, $\GStrats$ is a collection of tuples of agent strategies of the form $\langle \sgy_i\rangle_{i\in G}$, with
both the strategies $\alpha_i$ and 
 the set $G$ of agents varying. 
The semantics uses a relation 
$\Env, s\models^\GStrats \phi$, 
where $\Env =  \langle S, I, \Acts, \trans, \{O_i\}_{i\in \Ags}, \pi\rangle$ is an environment
and 
$s\in S$ is a state of $\Env$, 
and $\phi$ is a formula. 

For the definition, we need the notion of a path in $\Env$: this is a
function $\rho: \nat \rightarrow S$ such that for all $k \in \nat$ there exists a joint action $a$ with $(\rho(k), a, \rho(k+1)) \in \trans$. 
A path $\rho$ is \emph{from} a state $s$ if $\rho(0) = s$. 
A path $\rho$ is {\em consistent}  with a strategy $\sgy= \langle \alpha_i\rangle_{i\in G}$ for a group $G$ if for all   $k \in \nat$ 
there exists a joint action $a$ such that $(\rho(k), a, \rho(k+1)) \in \trans$ and  $a_i \in \sgy_i(\rho(k))$ 
for all $i \in G$. 
It is also convenient to identify the  \emph{path formulas} of ATL as formulas of the form $\nxt \phi$, $\always \phi$ or
$\phi \until \psi$ where $\phi$ and $\psi$ are ATL formulas.

The relation $\Env,s\models^\GStrats \phi$, where $s$ is a state of $\Env$ and $\phi$ is an ATL formula, is defined
by a mutual recursion with the relation $\Env,\rho\models^\GStrats \phi$ where $\rho$ is a path of $\Env$ and $\phi$ is a path
formula, as follows. Note that if $\atlop{G} \phi$ is an ATL formula then $\phi$ is a path formula. 
For evaluation of ATL formulas at a state we have the clauses

\begin{itemize}
\item $\Env,s\models^\GStrats p$ if   $p\in \pi(s)$; 
\item
$\Env,s\models^\GStrats \neg \phi$ if not $\Env,s\models^\GStrats \phi$; 

\item
$\Env,s\models^\GStrats \phi\wedge \psi$ if $\Env,s\models^\GStrats \phi$ and $\Env,s\models^\GStrats \psi$; 
%
%

\item
$\Env,s\models^\GStrats \atlop{G}\phi$ if  there exists a strategy $\sgy_G \in \GStrats$ for group $G$ 
such that for all paths $\rho$ from $s$ that are consistent with $\sgy_G$, 
we have $\Env,\rho \models^\GStrats \phi$;  
\end{itemize}
and for evaluation of a path formula at a path we have the clauses
\begin{itemize}

\item
$\Env,\rho \models^\GStrats \nxt \phi$ if   $\Env,\rho(1)\models^\GStrats \phi$;

\item
$\Env,\rho \models^\GStrats \always\phi$ if  have  $\Env,\rho(k)\models^\GStrats \phi$ for all $k \in \nat$;  

\item
$\Env,\rho\models^\GStrats \phi \until \psi$ if  there exists $m\geq 0$ such that $\Env,\rho(m)\models^\GStrats \psi$, and for all $k<m$, 
we have  $\Env,\rho(k)\models^\GStrats \phi$. 

\end{itemize}
The semantics for ATL given in \cite{ATLJACM} corresponds to the 
instance of this definition with $\GStrats$ equal to the set of perfect recall, perfect information 
group 
strategies, 
but we focus here on the variant where $\GStrats$ contains just imperfect information strategies.

We argue that the ATL construct $\atlop{G}\phi$ can be expressed in 
$\CTLsK(\Prop,\Ags \cup \strat(\Ags))$ as 
$$ \neg K_{e} \neg D_{\{e\} \cup \strat(G)}  \phi~.$$ 
Intuitively, here the outer operator 
$ \neg K_{e} \neg$ 
existentially switches to a 
point that has the same state of the environment as the current state (and hence the same local state
for all agents in $\Ags$), but may have different strategies for any of the agents. 
The inner operator $D_{\{e\} \cup \strat(G)} $ then fixes both the state of the environment and
the strategies selected by the group $G$ but allows all other agents to vary their strategy. 
It quantifies universally over these possibilities. Thus, the formula says that 
the group $G$ has a strategy that achieves $\phi$ from the current state, 
whatever strategy the other agents play. 
(An alternate way to express the formula using the richer expressive 
power of $\ESL$ is as 
$\exists x(K_e(\lid{\strat(G)}{x} \rimp \phi))$.)

More formally, consider the following translation from 
ATL to 
$\CTLsK(\Prop,\Ags \cup \strat(\Ags))$. 
For an ATL formula $\phi$, 
we write $\trl{\phi}$ for the translation of $\phi$, defined inductively on the construction of $\phi$ by 
the following rules 
\begin{align*} 
 & \trl{p} = p \\
 &  \trl{(\neg \phi)} = \neg \trl{\phi} \\
 & \trl{(\phi_1\land \phi_2)} = \trl{\phi_1} \land \trl{\phi_2} \\ 
& \trl{(\atlop{G}\phi)} = \neg K_e \neg D_{\{e\} \cup \strat(G)}  \trl{\phi} \\
& \trl{(\nxt \phi)} = \nxt \trl{\phi} \\ 
& \trl{(\always \phi)} = \always \trl{\phi}\\ 
& \trl{(\phi_1 \until \phi_2)} = \trl{\phi_1} \until \trl{\phi_2} 
\end{align*} 

Note that the semantics of the operators using $\atlop{G}$ quantifies over runs in which the 
agents $G$ run a particular strategy $\sgy_G$, but there is no constraint on the behaviour of the other agents: these are 
not assumed to choose their actions according to any particular strategy. 
A natural alternative to the definition above, would be to use 
the clause 
\begin{itemize} 
\item[]
$\Env,s\models^\GStrats \atlop{G}
\phi$ if 
there exists a strategy $\sgy \in \GStrats$ for group $G$ 
such that for all joint strategies 
$\beta\in \GStrats$ for group $\Ags$ 
with $\beta\restrict G = \sgy$, 
and all paths $\rho$ from $s$ that are consistent with $\beta$, 
we have  $\Env,\rho \models^\GStrats \phi$. 
\end{itemize} 
This variant corresponds more directly to the formula $ \neg D_{e} \neg D_{\{e\} \cup \strat(G)} \always \phi$
than does the ATL semantics. It is reasonable to take the position that it more naturally 
captures a concept of interest in competitive situations where agents are constrained in the 
strategies they are able to use. 

In the original semantics of ATL, where perfect information, perfect recall strategies were considered, 
the two definitions are equivalent, since for any behaviour of the other agents, there is a strategy that matches it. 
However, for the imperfect information, epistemic extension we consider, this does not hold. 
For example, if all strategies in $\GStrats$ are deterministic, then the above variant would not allow paths in 
which some agent in the complement of $G$ chooses an action $a$ at the first occurrence of a state $s$, but some other action $b$ 
at a later occurrence of $s$. On the other hand, such runs are allowed in the ATL semantics given above. 
Since the semantics of $\ESL$ assumes that all runs are generated by all agents running some strategy, 
we need to make some technical assumptions on $\GStrats$ to set up a correspondence with ATL. 

\newcommand{\rand}{\mathit{rand}} 
Define the ``random'' strategy  for agent $i$ to be the strategy $\rand_i$ defined by $\rand_i(s) = \Acts_i$ for all states $s\in S$. 
Given a strategy $\sgy =\langle \sgy_i\rangle_{i\in G}$ for a group of agents $G$ in an environment $\Env$, 
define the \emph{completion} of the strategy to be the joint strategy $\compl{\sgy}= 
\langle \sgy'_i\rangle_{i\in \Ags}
$ with $\sgy'_i = \sgy_i$ for $i\in G$ 
and with $\sgy'_i = \rand_i$ for all $i\in \Ags\setminus G$. 
Intuitively, this operation completes the 
group strategy to a joint strategy for all agents, by adding the ``random'' strategy for all agents not in $G$, 
so that these agents are completely unconstrained in their behaviour. 
Given a set of  strategies $\GStrats$ for groups of agents, 
we define the set of joint strategies $\compl{\GStrats} = \{\compl{\sgy}~|~\sgy\in \Strats\}$. 

A second technicality is needed that results from the way we have used $\GStrats$ as a parameter in a generalization
of the ATL semantics. A constraint on this set is needed to prove our correspondence result. 
Say that a set $\GStrats$ of group strategies is {\em restrictable} if
for every $\sgy\in \GStrats$ 
for group of agents $G$ 
and every group $H\subseteq G$, the restriction $\sgy\restrict H$ of $\sgy$ to agents in $H$ is also in $\GStrats$. 
Say that $\GStrats$ is \emph{extendable} if for every strategy $\sgy$ for a group $H$ and group $G\supseteq H$, 
there exists a strategy $\sgy'\in \GStrats$  for group $G$ whose restriction $\sgy'\restrict H$ to
$H$ is equal to $\sgy$. 
Intuitively, restrictability says that group strategies are closed under formation of subgroups, 
and extensibility says that  a group is not able to prevent any other agent from having
\emph{some} strategy that they are able to follow at the same time as the group 
follows its choice of strategy. 

The requirement that a set $\GStrats$ of group strategies 
be 
restrictable 
 and extendable is  quite mild. 
For example, if $\GStrats_i$ is a set of strategies for agent $i$, for each agent $i \in \Ags$, then 
the natural set of ``cartesian product strategies'' 
$$\GStrats = \{ \langle \alpha_i\rangle_{i\in G}~|~G\subseteq \Ags,~\forall i\in G~(\alpha_i \in \GStrats_i)\}$$ 
is both restrictable and extendable. 
In particular, the set of 
all group strategies, and the set of all locally uniform group strategies, are both restrictable 
and extendable. 
Another example of a collection of strategies satisfying this condition is the set of 
group strategies $\alpha$ in which at most $k$ agents follow a strategy that differs 
from a designated ``correct'' strategy $\strat$. Note that this collection is extendable because
an agent always has the option to choose the correct strategy, even if $k$ others have already deviated. 
This collection models a common assumption in the analysis of fault-tolerant distributed algorithms. 

A final technicality relates to the fact that 
whereas runs of an environment start at an initial state of the environment, and hence an environment may have unreachable states,  
models in the ATL semantics lack a notion of initial state, and formulas may be evaluated at any state. 
As already noted above, we  resolve this difference by viewing ATL models as environments in which 
all states are initial (hence reachable). 

The following result now captures in a precise way that the ATL semantics can be expressed in our logic
as claimed above, provided we allow joint strategies in which some agents run the random strategy.

\begin{theorem} \label{thm:atl}
For every environment $\Env$ in which all states are initial, for every
nonempty set of group strategies $\GStrats$ that is restrictable
and extendable, 
for every
state $s$ of $\Env$ and ATL formula $\phi$, 
we have  $\Env, s \models^\GStrats \phi$ iff for all (equivalently, some)  points $(r,m)$ of  
$\I(\Env,\compl{\GStrats}) $ with $r_e(m) = s$ we have $\I(\Env,\compl{\GStrats}), (r,m)\models \trl{\phi}$.  
\end{theorem} 

\begin{proof} 
For brevity, we write just $\I$ for $\I(\Env,\compl{\GStrats}) $.
For the claim that the quantifiers ``for all'' and ``some'' are interchangeable in the 
right hand side, note that formulas of the form $\trl{\phi}$ are boolean combinations 
of atomic 
propositions 
and formulas of the form 
$K_e\psi$, whose semantics  at 
a point $(r,m)$ depends only on $r_e(m)$. 
This gives the implication from the ``some'' case to the 
``for all'' case. For the implication from the ``for all'' case to the ``some'' case, note that the 
``for all'' case is never trivial because for all states $s$ of $\Env$, there exists 
a point $(r,m)$ of $\I$ with $r_e(m) = s$. This follows from the fact that all states are initial in 
$\Env$ and that the transition relation is serial, so that any group strategy $\sgy$ in $\GStrats$ is consistent with  
an infinite path from initial state $s$.  This corresponds to a run $r$ with $r(0) = (s,\compl{\sgy})$. 

It therefore suffices to show that $\Env, s \models^\GStrats \phi$ iff for all  points $(r,m)$ of  
$\I$ with $r_e(m) = s$ we have $\I, (r,m)\models \trl{\phi}$.  
Additionally, for path subformulas $\phi$ of the  form $\nxt \psi$, $\always \psi$ and $\psi_1 \until \psi_2$ of ATL formulas, we show that 
for all paths $\rho$, we have $\Env, \rho \models^\GStrats \phi$ iff for all  points $(r,m)$ of  
$\I$ with $r_e[m \ldots] = \rho$ we have $\I, (r,m)\models \trl{\phi}$

We proceed by induction on the construction of $\phi$. 
The base case of atomic propositions, as well as the cases for the boolean constructs,  
are trivial. 
The claim concerning path formulas is also straightforward from the semantics of the 
temporal operators and, inductively, the claim concerning state formulas. 

We consider next the case of  $\phi = \atlop{G} \psi$. 
We show that  $\Env, s \models^\GStrats \atlop{G}\psi$ iff 
for all  points $(r,m)$ of  $\I$ with $r_e(m) = s$ we have $\I, (r,m)\models \neg K_e \neg D_{\{e\} \cup \strat(G)}  \trl{\phi}$. 

Suppose first $\Env, s \models^\GStrats \atlop{G} \psi$. 
Let $(r,m)$ be a point of $\I$ with $r_e(m) =s$. 
We show that 
$\I, (r,m)\models \neg K_e \neg D_{\{e\} \cup \strat(G)}   \trl{\psi}$. 
By the ATL semantics, there exists a strategy $\sgy_G\in \GStrats$ for group $G$ such that for all paths $\rho$ of $\Env$ from $s$ that are consistent with $\sgy_G$ we have 
$\Env, \rho \models^\GStrats \psi$. Let $\sgy = \compl{\sgy_G}$ (note that this is in $ \compl{\GStrats}$) 
and 
(using the fact that all states are initial)
let $r'$ be a run of $\I$ with 
$r'(0) = (s,\sgy)$.  
Because 
$r_e(m) = s = r'_e(0)$, we have $(r,m) \sim_e (r',0)$,  and it suffices 
to show that $\I, (r',0)\models  D_{\{e\} \cup \strat(G)}  \trl{\psi}$. 
For this, suppose that $(r'',m'')$ is 
any 
point of $\I$ 
with $(r',0) \sim_{\{e\} \cup \strat(G)} (r'',m'')$. 
We show that $\I, (r'',m'')\models   \trl{\psi}$. 
Now $r''(m'') = (t, \sgy')$ implies that $\sgy'_i = \sgy_i$ for all $i \in G$. 
Thus, the path $\rho = r''_e(m'') r''_e(m''+1) \ldots$ in $\Env$ is consistent with $\sgy_G$, 
and 
$\rho(0) = r''_e(m'') = r'_e(0)  = s$. 
It follows that 
$\Env, \rho  \models^\GStrats \psi$.
Using the induction hypothesis, it follows that 
$\I, (r'',m'')\models   \trl{\psi}$. 
This completes the argument that $\I, (r',0)\models  D_{\{e\} \cup \strat(G)}  \trl{\psi}$.

Conversely, suppose that  for all  points $(r,m)$ of  $\I$ with $r_e(m) = s$ we have $\I, (r,m)\models \neg K_e \neg D_{\{e\} \cup \strat(G)}  \trl{\psi}$. 
We show that $\Env, s \models^\GStrats \atlop{G} \psi$. 
Using the fact that all states are initial, let $r$ be a run of $\I$ with  $r_e(0) = s$, and 
hence $\I, (r,0)\models \neg K_e \neg D_{\{e\} \cup \strat(G)}  \trl{\psi}$.  
Then there exists a point $(r',m')$ of $\I$ such that $r'_e(m')  = s$ and 
$\I, (r',m')\models  D_{\{e\} \cup \strat(G)}  \trl{\psi}$. 
Let $r'(m') = (s,\sgy)$. Then 
there exists a strategy $\beta\in \GStrats$ for some set of agents $G'$ such that 
$\sgy = \compl{\beta}\in  \compl{\GStrats}$. 
Let $H = G\cap G'$. By restrictability, 
we have $\beta\restrict H \in \GStrats$. 
By extendability, there exists a strategy $\gamma \in \GStrats$ for group $G$ 
such that $\gamma \restrict H = \beta \restrict H$.  
Taking $\sgy' = \compl{\beta\restrict H}$, it follows that $\sgy' \in \compl{\GStrats}$. 
Note that $\sgy' \restrict G = \sgy\restrict G$ and $\sgy'_i = \rand_i$ for $i \in \Ags \setminus G$.  
In particular, $\sgy'_i = \rand_i$ for $i \in G \setminus  H$. Thus, 
any path consistent with $\gamma$ is consistent with $\sgy'$.  

To prove that $\Env, s \models^\GStrats \atlop{G} \psi$, we show that 
for every path $\rho$ of $\Env$ from $s$ consistent with 
$\gamma$, 
we have $\Env, \rho   \models^\GStrats   \psi$. 
For this, let $\rho$ be a path from $s$ consistent with the strategy 
$\gamma$ 
for group $G$.
By the conclusion of the previous paragraph, $\rho$ is consistent with the joint strategy $\sgy'$ for all agents. 
Since $s$ is an initial state of $\Env$, there exists a run $r''$ of 
$\I$ with $r''(0) = (s, \sgy')$ and $r''_e[0\ldots \infty] = \rho$. 
Moreover, $(r',m') \sim_{\{e\} \cup \strat(G)} (r'',0)$.
Thus, we obtain from $\I, (r',m')\models  D_{\{e\} \cup \strat(G)}  \trl{\psi}$
that $\I, (r'',0)\models  \trl{\psi}$. 
By the induction hypothesis, 
we obtain that \mbox{$\Env, \rho \models^\GStrats \psi$}. 
\end{proof}

The $\ESL$ interpretation unpacks the alternating double quantification in the semantics of $\atlop{G}\phi$. 
$\ESL$ offers the advantage of being able to express notions that are not expressible in 
ATL. For example,
under assumptions similar to those of Theorem~\ref{thm:atel}, 
$$ \neg D_{e} \neg (( \neg D_{\{e\} \cup \strat(\Ags)} \neg \sometimes p) \land (D_{\{e\} \cup \strat(G)} \always q)) $$ 
says that, from the current state, there is a joint strategy for all agents, such that,
some runs of this joint strategy  satisfy $\sometimes p$, and group $G$'s strategy alone suffices to ensure that 
$\always q$. 


There has been discussion in the literature on ATL about whether strategies should be
\emph{revocable} or \emph{irrevocable}. Consider a formula such as 
$$ \atlop{A} \always ( p \land \atlop{A,B} \sometimes q)~.$$ 
This says that $A$ has a strategy that ensures that it is always the case both that $p$ holds, 
and that $A$ and $B$ together have a strategy that ensures that eventually $q$. 
Under the ATL semantics, the strategy of $A$ used to satisfy the inner 
formula  $\atlop{A,B} \sometimes q$ is allowed to be different from the strategy 
of $A$ 
referred
to by the outer operator. That is, to satisfy the inner formula, $A$ is allowed
to \emph{revoke} the strategy selected by the outer operator. 

This aspect of the ATL semantics has been questioned 
\cite{AgotnesGJ07}, and it has been proposed that the semantics of the formula
$\atlop{G} \phi$ should be defined so that it fixes the strategies of agents in 
the group $G$ and does not allow these to be varied in interpreting operators in the 
formula $\phi$. In such a semantics, the strategy choices are \emph{irrevocable}. 
Using our framework, both revocable and irrevocable interpretations of the formula can be represented. 
We show this with two formulas that 
are almost identical, with the point of difference indicated by use of boldface.
The interpretation allowing strategy revocation would be captured
by translating both operators as described above, yielding the formula 
$$ \neg D_{e} \neg D_{\{e,\strat(A)\} }\,( \always p \land  \neg \mathbf{D_{e}} \neg D_{\{e,\strat(A),\strat(B)\} } \,\sometimes q)~.$$ 
Note that here the outer operator prefix $\neg D_{e} \neg D_{\{e,\strat(A)\} }$ selects a strategy for $A$ and plays it against 
all strategies of the other agents, and because the operator $\mathbf{D_{e}}$ allows all agent's strategies to vary, 
the inner operator prefix $\neg \mathbf{D_{e}} \neg D_{\{e,\strat(A),\strat(B)\} }$ 
drops the selected strategy of $A$, and selects a fresh strategy for $A$ and $B$ together to play against all strategies of other agents. 
On the other hand, we can force the strategy of agent $A$ to remain fixed
in the inner choice of strategies by  means of the formula
$$ \neg D_{e} \neg D_{\{e,\strat(A)\} }\,( \always p \land  \neg \mathbf{D_{\{e,\strat(A)\}}}\, \neg D_{\{e,\strat(A),\strat(B)\} } \, \sometimes q)~.$$  
Note that the inner operator $\mathbf{D_{\{e,\strat(A)\}}}$ varies all agent's strategies, except that of $A$. 
Evidently, at any point in a nested formula, our approach gives us the freedom
to choose which players' strategies we wish to vary and which  to fix. 

 
%
A logic with 
revocable strategies is presented in
Brihaye et al. \cite{BCLM2009},  which 
considers the extension of ATL with strategy context, or $ATL_{sc}$. 
Formulas are evaluated with respect to a context which is a group strategy $\gamma_G$ for some group $G$. 
The logic has modalities $\cdot\rangle H \langle\cdot\phi$, and $ \langle\cdot H \cdot\rangle\phi$.  
Intuitively,  $\cdot\rangle H \langle\cdot\phi$ reduces the context group $G$ to $G\setminus H$ by restricting $\gamma_G$ to $G \setminus H$. 
The modality $ \langle\cdot H \cdot\rangle\phi$ selects a new group strategy $\gamma_H$ for group $H$, 
and constructs the new context $\gamma_H\circ \gamma_G$ for group $G \cup H$ in which agents $i$ in $H$ play $\gamma_H(i)$, 
and agents $i$ in $G\setminus H$ play $\strat_G(i)$. The formula $\phi$ is then evaluated with respect to context 
 $\gamma_H\circ \gamma_G$ in all 
runs in which $G \cup H$ plays $\gamma_H\circ \gamma_G$ against an arbitrary behaviour of all other agents. 

Evaluation of formulas commences with respect to the empty context, so each subformula 
is evaluated with respect to a context  for a group $G$ that can be determined from the 
operators on the path from the root to that subformula. 
This means that to represent a formula $\phi$ of $ATL_{sc}$, we need to translate it with respect to 
a group $G$; we write the translation as $\phi^G$. 
Roughly, with respect to a context for group $G$, 
the formula $\langle\cdot H \cdot\rangle\phi$ can then be expressed with our logic as 
 $$ 
 (~\langle\cdot H \cdot\rangle\phi ~)^{G} = 
  \neg D_{\{e,\strat(G\setminus H)\}}\, \neg 
 D_{\{e,\strat(G\cup H)\} } \, 
 \phi^{G\cup H}
 $$ 
and  the formula $\cdot\rangle H \langle\cdot\phi$ can be expressed as 
 $$  (~ \cdot\rangle H \langle\cdot\phi ~ )^G = 
 \phi^{G\setminus H}~.
 $$

However, we note that the semantics in \cite{BCLM2009} is based on perfect recall. This explains that the complexity of model checking $ATL_{sc}$ is non-elementary, while the complexity of model checking our logic \ESL\ is EXPSPACE-complete (Theorem~\ref{mcESLupper} and Theorem~\ref{mcESLlower}).

Another work by van der Hoek, Jamroga and Wooldridge 
\cite{CATL} introduces constants that refer to strategies, and adds to ATL a new 
 (counterfactual) modality $C_i(c,\phi)$, with the intended reading ``if it were the case that agent $i$ committed to the strategy denoted by 
 $c$, then $\phi$''. 
 (The meaning of $c$ is bound in the semantic context, and the logic does not allow quantification over $c$.)
 The formula $\phi$ here is not permitted to contain further references to agent $i$ strategies. 
 To interpret the formula $C_i(c,\phi)$ in an environment  $E$, the environment is first updated to a new environment $E'$ by removing all transitions that
 are inconsistent with agent $i$ running the strategy referred to by $c$, and then the formula $\phi$ is evaluated in $E'$. 
In $\ESL$, the assertion that $i$ is running a particular strategy can be made by the formula 
$\lid{\strat(i)}{x}$, 
where $x$ is taken to denote a 
global state in which the local component $\strat(i)$ denotes the strategy denoted by $c$. 
The formula $C_i(c,\phi)$ can then be expressed in our framework as 
$$D_{\{e\} \cup \sigma(\Ags \setminus \{i\})}( 
\lid{\strat(i)}{x}
\rimp  \phi^{+\strat(i)})$$ 
where in the translation $\phi^{+\strat(i)}$ of  $\phi$  we ensure that there is no further deviation from the strategy of agent $i$ 
by adding $\sigma(i)$ to the group of every  knowledge operator occurring later in the translation. 
We remark that because it deletes information from the transition relation, strategy choices made by the construct  $C_i(c,\phi)$ 
are irrevocable, whereas our logic is richer in that it allows revocation of the corresponding choices. 


\subsection{Connections to variants of ATEL} \label{sec:atelvarn}

Alternating temporal epistemic logic (ATEL)
adds 
epistemic operators to ATL \cite{ATEL}. 
As a number of subtleties arise in the formulation of such logics, several variants of ATEL have since
been developed. In this section, we consider a number of such variants 
and argue that our framework is able to express the main strategic concepts 
from these variants. We begin by recalling ATEL as defined in \cite{ATEL}.

The syntax of \ATEL\ is given as follows: 
$$\phi \equiv p ~|~\neg \phi~|~\phi_1 \lor \phi_2~|~\atlop{G}\nxt\phi ~|~~\atlop{G}\always\phi ~|~\atlop{G}(\phi_1 \until \phi_2) ~|~K_i\phi~|~D_G\phi~| ~C_G\phi
$$
where $p \in \Prop$, $i \in \Ags$ and $G\subseteq \Ags$. 
This just adds the operators $K_i, D_G$ and $C_G$ to the syntax for ATL given above. 
As usual, we may define $E_G\phi$ as $\bigwedge_{i\in G} K_i \phi$. 
 The intuitive meaning of the constructs is as in \CTLsK\ above, with additionally, 
$\atlop{G}\phi$ having the intuitive reading that group $G$ has a strategy for assuring that $\phi$ holds.


The relation $\Env,s\models^\GStrats \phi$ is extended from ATL to ATEL by adding the following 
clauses to the inductive definition: 
\begin{itemize}

\item
$\Env,s\models^\GStrats K_i\phi$ if  $\Env,t \models^\GStrats \phi$
 for all $t\in S$ with $t\sim_i s$; 

\item
$\Env,s\models^\GStrats D_G\phi$ if  $\Env,t \models^\GStrats \phi$, for all $t\in S$ with $(s,t) \in  \bigcap_{i\in G}\sim_i$.

\item
$\Env,s\models^\GStrats C_G\phi$ if  $\Env,t \models^\GStrats \phi$
for  all $t\in S$ with $(s,t) \in (\cup_{i\in G} \sim_i)^*$; 
\end{itemize}
where we define, for each $i \in \Ags$, the equivalence relation $\sim_i$ on states $S$, 
by $s\sim_i t$ if and only if $O_i(s)=O_i(t)$.

\newcommand{\detc}{\mathit{det}}  

The specific version of ATEL defined in \cite{ATEL} is obtained from the above definitions 
by taking 
$\GStrats = 
\{\sigma_G~|~G\subseteq \Ags, ~\sigma_G~\text{a deterministic $G$-strategy in}~ \Env\}$. 
That is, following the definitions for ATL, this version works with arbitrary deterministic group strategies, in which 
an agent selects its action 
as if it had full information of the state. 
This aspect of the definition has been 
criticized 
by Jonker \cite{Jonker2003}  and (in the case of ATL without epistemic operators) by Schobbens \cite{Schobbens2004}, 
who argue that this choice is not in the spirit of the epistemic extension, in which observations
are intended precisely to represent that 
agents do not have full information of the state. They propose that the definition instead 
be based on the set 
$
\GStrats 
= \{\sigma_G~|~G\subseteq \Ags, ~\sigma_G~\text{a locally uniform deterministic $G$-strategy in}~ \Env\}$. 
This ensures that in choosing an action, agents are able to use only the information available in their observations. 

We concur that the use of locally uniform strategies is the more appropriate choice, but in either event, 
we now argue that our approach using strategy space is able to express everything
that can be expressed in ATEL. 
We may extend the translation into our logic given above from ATL to ATEL, by adding 
the following rules:
\begin{align*} 
&  \trl{(K_i\phi)} = K_i \trl{\phi} ~~~~~~\trl{(D_G\phi)} = D_G \trl{\phi}~~~~~~~\trl{(C_G\phi)} = C_G\trl{\phi} 
\end{align*}

In order to obtain a correspondence with ATEL, which does not have a notion of initial states, 
we again work with environments in which all states are initial. 
The following result shows that Theorem~\ref{thm:atl} extends from ATL to the logic ATEL. 


\begin{theorem} \label{thm:atel}
For every environment $\Env$ 
in which all states are initial,  for every 
nonempty set of group strategies 
$\GStrats$ 
that is restrictable,
for every state $s$ of $\Env$ and ATEL formula $\phi$, 
we have  $\Env, s \models^\GStrats \phi$ iff for all (equivalently, some)  points $(r,m)$ of  
$\I(\Env,\compl{\GStrats}) $ with $r_e(m) = s$ we have $\I(\Env,\compl{\GStrats}), (r,m)\models \trl{\phi}$.  
\end{theorem}

\begin{proof} 
The proof extends the proof of Theorem~\ref{thm:atl}. 
The argument for the equivalence of the universal and existential quantifications in the 
right hand side 
of the ``iff'' 
continues to apply, even though the translation now contains formulas of the form $K_i\phi$, 
because, by construction in Section~\ref{sec:generatedIS}, $r_e(m) = r'_e(m')$ implies $r_i(m) = r'_i(m')$. 
The remainder of the proof extends the inductive argument. 

Consider $\phi = K_i \psi$. 
Then 
$ \trl{(K_i\psi)} = K_i \trl{\psi}$. 
We suppose first that $\Env, s \models^\GStrats \phi$  
and show that for all  points $(r,m)$ of  
$\I$ with $r_e(m) = s$ we have $\I, (r,m)\models \trl{\phi}$, 
i.e., $\I, (r,m)\models K_i \trl{\psi}$.  
Let $(r,m)$ be a point of  $\I$ with $r_e(m) = s$. 
We need to show that for all points $(r',m')$ of $\I$ 
with $(r,m) \sim_i (r',m')$ we have $\I, (r',m')\models \trl{\psi}$. 
But if $(r,m) \sim_i (r',m')$ then $r'_e(m') \sim_i r_e(m)=s$ in $\Env$. 
Thus, from $\Env, s \models^\GStrats K_i \psi$ it follows that 
$\Env, r'_e(m')  \models^\GStrats \psi$. By the induction hypothesis, 
we obtain that $\I, (r',m')\models \trl{\psi}$, as required. 

Conversely, suppose that  for all  points $(r,m)$ of  
$\I$ with $r_e(m) = s$ we have  $\I, (r,m)\models K_i \trl{\psi}$.  
We show that $\Env, s \models^\GStrats K_i \psi$. 
Let $t$ be any state of $\Env$ with $s\sim_i t$. We have to show $\Env, t \models^\GStrats \psi$. 
First, since $s$ is an initial state of $\Env$, there exists a run $r$ of $\I$ with $r_e(0) = s$, 
and joint strategy equal to any strategy in $\compl{\GStrats}$, so we take $m=0$, 
and we have $\I, (r,m)\models K_i \trl{\psi}$.  
Then for all points $(r',m')$ of $\I$ with 
$r'_e(m') =t$, we have $(r,0) \sim_i (r',m')$, from which it follows that $\I, (r',m')\models \trl{\psi}$.  
By the induction hypothesis, we have $\Env,t \models \psi$, as required. 
This completes the proof for the case of $\phi = K_i \psi$. The argument for the distributed and common knowledge operators 
is similar, and left to the reader. 

\end{proof}

We remark that our translation maps ATEL into 
$\CTLK(\Ags\cup \strat(\Ags), \Prop)$, 
the fragment 
 that 
that we show in Theorem~\ref{mcCTLsK} below 
to have PSPACE-complete model checking complexity.  
This strongly suggests that this fragment  has a strictly
stronger expressive power than ATEL, since the complexity of model checking 
ATEL, 
assuming uniform strategies,  is known to be $P^{NP}$-complete. 
(The class $P^{NP}$ consists of problems solvable by PTIME computations with access to an NP oracle.) 
For ATEL, model checking can be done with  
a polynomial time (with respect to the size of formula) computation with  access to an oracle that is in NP with respect to both the number of states and the number of joint actions. 
In particular,  \cite{Schobbens2004} proves this upper bound and \cite{JD2006} proves  a matching lower bound.

Similar  translation results can be given for other alternating temporal epistemic logics from the literature. 
We sketch a few of these translations here. 

Jamroga and van der Hoek \cite{JvdH2004} 
discuss issue of \emph{de dicto} and \emph{de re} interpretations of ATEL formulas. They 
consider the formula $K_i\atlop{i} \phi$. 
(Note that here $\phi$ is a path formula). 
The ATEL semantics states that for an environment $E$ and a state 
$s$, we have $E, s\models K_i\atlop{i} \phi$ when in every state $t$ consistent with agent $i$'s knowledge, some strategy for agent $i$, depending on $t$, is guaranteed 
to satisfy $\phi$. This is consistent with there being no  \emph{single} strategy for agent $i$ that agent $i$ knows will work to achieve $\phi$
in all such states $t$. To express that a single strategy is known to guarantee $\phi$, 
they formulate a general construct $\atlop{G}^\bullet_{\mathcal{ K}(H)}\phi$ that says, effectively, that there is a strategy  for a group $G$ that another group $H$
knows (for notion of group knowledge $\mathcal{ K}$) to achieve goal $\phi$. 
(Here again, $\phi$ is a path formula.)
The notion of group knowledge $\mathcal{ K}$ could be $E$ for everyone knows, $D$ for  distributed knowledge, or $C$ for common knowledge. 
More precisely%
\footnote{As above, we have generalized the definition to be relative to a set of group strategies $\GStrats$.
The strategies used in \cite{JvdH2004} are imperfect information, perfect recall strategies; we formulate the definition here with imperfect information, imperfect recall strategies.}, 
\begin{align*}
E, s\models^\GStrats 
\atlop{G}^\bullet_{\mathcal{ K}(H)}\phi 
&~~~ \parbox[t]{3.5in}{if there exists a 
locally 
uniform 
group strategy $\alpha\in \GStrats$ 
for group $G$
such that for all states $t$ with $s \sim_H^\mathcal{ K} t$, 
and  for all paths $\rho$ from $t$ that are consistent with $\alpha$, 
we have that $E,\rho \models^\GStrats \phi$.
} 
\end{align*} 
Here $\sim^\mathcal{ K}_H$ is the appropriate epistemic indistinguishability relation on states of $E$. 
The particular case $\atlop{G}^\bullet_{E(G)}\phi$ is also proposed as the semantics for the ATL construct $\atlop{G} \phi$ in 
\cite{Schobbens2004,Jonker2003,JA07}.  

The construct $\atlop{G}^\bullet_{D(H)}\phi$ can be represented in the 
$\CTLK(\Ags \cup \strat(\Ags) , \Prop)$
fragment of $\ESL$  as 
$$ \neg K_{e} \neg D_{H\cup\strat(G)} \phi~.$$
 Intuitively, here the first modal operator 
 $\neg K_e\neg $ 
 switches the strategy of all the agents while maintaining the state $s$,
 thereby selecting a  strategy  $\alpha$ for group $G$ in particular, 
 and the next operator 
 $ D_{H \cup \strat(G)}$ 
 verifies that the group $H$ knows  that the strategy 
 being 
 used by group $G$ guarantees $\phi$.  Similarly, 
 $\atlop{G}^\bullet_{E(H)}\phi$ can be represented as 
$$ \neg K_{e} \neg \bigwedge_{i\in H} D_{\{i\}\cup\strat(G)} \phi~.$$
The precise statement and  proof of these correspondences is  similar to that in Theorem~\ref{thm:atel}.


In the case of the construct $\atlop{G}^\bullet_{C(H)}\phi$, the definition involves the 
common knowledge that a group $H$ of agents would have if they knew a particular strategy being used by another group $G$. 
By analogy with the above cases, one might expect this to be expressible using the formula 
 $\neg K_e \neg C_{H \cup \strat(G)} \phi$. However, this does not give the intended meaning.
Note that the semantics of the formula $C_{H \cup \strat(G)} \phi$  quantifies over points $(r',m')$ 
reachable through chains $(r,m)=(r_0,m_0) \sim_{i_1} (r_1,m_1)  \sim_{i_2} \ldots \sim_{i_n} (r_n,m_n)=(r',m')$, 
where each $i_j$ is in the set $H\cup \strat(G)$. But this loses the connection to common knowledge of group $H$
and fails to fix the strategy of group $G$. Instead, what we would need to capture is chains of the form 
$(r_0,m_0) \sim_{\{i_1\}\cup \strat(H)} (r_1,m_1)  \sim_{\{i_2\}\cup \strat(H)}  \ldots \sim_{\{i_n\}\cup \strat(H)}  (r_n,m_n)=(r',m')$, 
where each $i_j$ is in the set $G$.  
For this, it appears we need to be able to express the greatest fixpoint $X$ of the equation $ X \equiv \bigwedge_{i\in G} D_{\{i\}\cup \strat(H)} (X \land \phi)$. 
The language 
$\CTLK(\Ags\cup\strat(\Ags),\Prop)$ 
does not include fixpoint operators and it does not seem that  this fixpoint is expressible. 
Indeed, the construct $\atlop{G}^\bullet_{C(H)}\phi$  does not appear to be 
expressible using the fragment 
$\CTLK(\Ags \cup\strat(\Ags),\Prop)$. 

On the other hand, common knowledge of group $H$ about the effects of a fixed strategy of group $G$ 
can be expressed with $\ESL$ in a natural way by the formula  
$$C_{H}(\lid{\strat(G)}{x} \rimp \phi)$$
which says that it is common knowledge to the group 
$H$ that $\phi$ holds if the group $G$ 
is running the strategy profile 
captured 
by the variable $x$. 
Using this idea, the 
 construct 
 $\atlop{G}^\bullet_{C(H)}\phi$ 
 can be represented with $\ESL$  as 
$$ \existsg{x}  C_{H}(\lid{\strat(G)}{x} \rimp \phi)~.$$
The following  result states this claim precisely. 

\begin{theorem} \label{thm:atel}
Let $\Env$ be an environment  in which all states are initial,
and let $\GStrats$ be a restrictable and extendable set of group strategies in $E$. 
Let $\I = \I(\Env,\compl{\GStrats})$. 
Assume that $\phi$ 
is a path formula
and that $\trl{\phi}$ is an $\ESL$ formula
without free variables,  
such that  for every path $\rho$ of $\Env$, 
we have  $\Env, \rho \models^\GStrats \phi$ iff for all (equivalently, some)  points $(r,m)$ of  
$\I$ with $r_e[m\ldots] = \rho$ we have $\I, (r,m)\models \trl{\phi}$.  

Then for all states $s$ of $\Env$, 
we have  $\Env, s \models^\GStrats \atlop{G}^\bullet_{C(H)}\phi$ iff for all (equivalently, some)  points $(r,m)$ of  
$\I$ with $r_e(m) = s$ we have $\I, (r,m)\models \existsg{x}  C_{H}(\lid{\strat(G)}{x} \rimp \trl{\phi})$.  
\end{theorem}

\begin{proof} 
The argument for the equivalence between the universal and existential versions of the right hand side of
the iff is similar to that 
in Theorem~\ref{thm:atl}. 

Suppose first that  $\Env, s \models^\GStrats \atlop{G}^\bullet_{C(H)}\phi$. 
 Let $(r,m)$ be a  point of $\I$ with $r_e(m) =s$. We need to prove that  $ \I,(r,m) \models \existsg{x} C_{H}(\lid{\strat(G)}{x} \rimp \trl{\phi})$. 
From  $\Env, s \models^\GStrats \atlop{G}^\bullet_{C(H)}\phi$ it follows that there exists a 
strategy $\alpha\in \GStrats$ for group $G$, such that  for all states $t$ with $s \sim^C_H t$ and paths $\rho$ from $t$
consistent with $\alpha$, 
we have 
$\Env, \rho \models^\GStrats \phi$. 
Let $r'$ be any run with $r'_{\strat(G)}(0) = \alpha$, and define $\Gamma$ to be a context with $\Gamma(x) = r(0)$. 
To prove $ \I,(r,m) \models \existsg{x} C_{H}(\lid{\strat(G)}{x} \rimp \trl{\phi})$, we 
show that $\Gamma, \I,(r,m) \models C_{H}(\lid{\strat(G)}{x} \rimp \trl{\phi})$. 
For this, suppose that 
$(r,m) = (r^0,m^0) \sim_{i_1} (r^1,m^1)  \sim_{i_2} \ldots \sim_{i_k} (r^k,m^k)$, where $i_j \in H$ for $j = 1\ldots k$, 
and assume that $\Gamma, \I,(r^k,m^k) \models \lid{\strat(G)}{x}$. 
We need to show that $\Gamma, \I,(r^k,m^k) \models \trl{\phi}$. 

Note that we have $s = r(m) = r^0_e(m^0) \sim_{i_1} \ldots \sim_{i_k} r^k_e(m^k)$.  
Since $\Gamma, \I,(r^k,m^k) \models \lid{\strat(G)}{x}$, we have that $r^k_{\strat(G)}(m^k) = 
\Gamma(x)_{\strat(G)} = \alpha$. 
Thus, the sequence $\rho = r^k_e[m^k\ldots]$ is a path of $\Env$ consistent with the 
group strategy $\alpha$. It follows that 
$\Env, \rho \models^\GStrats \phi$. By assumption, this means that $\Gamma, \I,(r^k,m^k) \models  \trl{\phi}$. 

Conversely, let $(r,m)$ be a point of  $\I$ with $r_e(m) = s$ and  $ \I, (r,m)\models \existsg{x}  C_{H}(\lid{\strat(G)}{x} \rimp \trl{\phi})$, 
witnessed by  $\Gamma, \I, (r,m)\models   C_{H}(\lid{\strat(G)}{x} \rimp \trl{\phi})$. 
Note that $\Gamma(x)_{\strat(\Ags)} = \compl{\beta}$ where $\beta\in \GStrats$ is a group strategy for some group 
$G'$.  For agents $i \in G\setminus G'$, we have that $\Gamma(x)_{\strat(i)}$ is the random strategy  $\rand_i$.  
It follows that any path consistent with  $\Gamma(x)_{\strat(G\cap G')}$ is also consistent with $\Gamma(x)_{\strat(G)}$. 
Let $\alpha \in \GStrats$ be any group strategy for group $G$  with 
$\alpha\restrict (G\cap G') = \Gamma(x)_{\strat(G\cap G')}$. 
Such a strategy exists by the fact that $\GStrats$ is restrictable and extendable: 
we may take $\alpha$ to be an extension of $\beta\restrict (G\cap G')$. 
 Then we have that any path consistent with $\alpha$ is 
consistent with $\Gamma(x)_{\strat(G)}$.

We show $\Env, s \models^\GStrats \atlop{G}^\bullet_{C(H)}\phi$, with $\alpha$ as the witnessing strategy for group $G$. 
For this, let 
$s = s_0 \sim_{i_1} s_1  \sim_{i_2} \ldots \sim_{i_k} s_k$, 
where $i_j \in H$ for $j = 1\ldots k$, 
and let $\rho$ be a path from $s_k$ consistent with $\alpha$. 
We show $E, \rho \models^\GStrats \phi$. 
By the observation above, $\rho$ is 
also consistent with $ \Gamma(x)_{\strat(G)}$. 
Let $r^k$ be a run  with $r^k_e[0\ldots] = \rho$, and  $r^k_{\strat(G)}(0) =  \Gamma(x)_{\strat(G)}$. 
(We can take $r^k_{\strat(\Ags)} = \compl{\beta \restrict (G\cap G')}$, which is in $\compl{\GStrats}$.)  
Then $\Gamma,\I,(r,0)\models \lid{\strat(G)}{x}$. 
Moreover, for each $j = 1 \ldots k-1$, let $r^j$ be any run with $r^j_e(0) = s_j$. 
Then $(r,m) = (r^0,m^0) \sim_{i_1} (r^1,0) \sim_{i_2} \ldots  \sim_{i_k} (r^k,0)$. 
It follows from $\Gamma, \I, (r,m)\models   C_{H}(\lid{\strat(G)}{x} \rimp \trl{\phi})$ that 
$\Gamma, \I, (r^k,0)\models  \trl{\phi}$, 
and in fact $\I, (r^k,0)\models  \trl{\phi}$, since $\trl{\phi}$ 
has no free variables. 
Since $r^k[0\ldots ] = \rho$, by assumption, we have $E,\rho \models \phi$. 
This proves  $\Env, s \models^\GStrats \atlop{G}^\bullet_{C(H)}\phi$. 
\end{proof}

 The above equivalences give 
 a reduction of the complex operators of  \cite{JvdH2004}
 that makes their epistemic content more explicit by expressing this using standard epistemic operators. 

An alternate 
 approach to decomposing the operators $\atlop{G}^\bullet_{\mathcal{ K}(H)}$ is proposed in 
 \cite{JA07}. 
 By comparison with our
 standard approach to the semantics of the epistemic operators, this  proposal  uses  ``constructive knowledge'' 
 operators which require a nonstandard semantics in which formulas are evaluated at sets of states rather than 
 at individual states. 
 Evaluation at single world $q$ is treated as equivalent to evaluation at the set $\{q\}$. 
 For each standard (group) epistemic operator $\K = E,D,C$, there is a constructive version $\hat{\K} = \mathbb{E}, \mathbb{D}, \mathbb{C}$. 
 Atomic propositions $p$ are evaluated at sets of states $Q$ by 
 $$ \Env, Q \models p ~~\text{ if for all states $q\in Q$ we have $\Env,q \models p$. } $$ 
 (As above we define the semantics on environments $\Env$ rather than ATEL models.) 
For the constructive epistemic operators, 
 $$ \Env, Q \models \hat{\K}_G\phi ~~\text{ if  
 $\Env,\{q'\in Q~|~\exists q\in Q~(q\sim^\K_G q') \}  \models \phi$ } $$  
 and for the ATL operator $\atlop{G}\phi$ we have 
 \begin{quote}
  $ \Env, Q \models \atlop{G}\phi $  if  there exists a strategy $\alpha$ for group $G$ such that $\phi$ holds in all runs 
  starting in a state in $Q$ in which group $G$ plays the strategy $\alpha$. 
  \end{quote}
Note that the ATEL formula 
  $\K_G\atlop{G}\phi $ says that at each world considered possible (in the appropriate sense for $\K$) by 
group $G$, there exists a (possibly different) strategy for  $G$ that achieves $\phi$. 
By contrast, $\K_G\atlop{G}\phi $ says that there exists a \emph{single} strategy for $G$ 
that achieves $\phi$ from each world considered possible (in the appropriate sense for $\K$) by $G$. 

This logic is shown in \cite{JAvH2008} to have a normal form, in which every subformula starting with a 
constructive knowledge operator $\hat{\K}^1_{G}$ is of the form $\hat{\K}^1_{G_1} ... \hat{\K}^n_{G_n}\phi$ where $\phi$ starts with a strategy modality and each 
$\K^i \in \{E,D,C\}$. Such a normal form subformula, evaluated at a single state, can be represented
in $\ESL$ as 
 $$ \existsg{x}  \K^1_{G_1} ... \K^n_{G_n}(\lid{\strat(H)}{x} \rimp \phi)~. $$
 Precise formulation and proof of the claim are similar to the proofs above and left to the reader.

\newcommand{\str}{{\tt Str}}
\newcommand{\trlsl}[2]{#2^{#1}} 
\newcommand{\assignedAgents}{*}
\newcommand{\slvar}{{\tt Var_{SL}}}

\subsection{Strategy Logic}\label{sec:sl}

Chatterjee et al's strategy logic~\cite{CHP10}, which we call CHP-SL, 
 following the convention in \cite{MogaveroMV10}, is an extension of ATL* for two-player games. Let $x,y$ be two variables 
 ranging 
 over Player 1 and Player 2's strategies. 
 The logic allows these variables to be quantified: if $\phi $ is a formula then $\exists x.\phi$ and $\forall x.\phi$ are formulas. 
Additionally the effects of a particular combination of player strategies can be expressed using the formula $\phi(x,y)$, which 
says that $\phi$ holds if player 1 plays strategy $x$ and player 2 plays strategy $y$. 
 Thus, the ATL* formula $\atlop{1}{\phi}$ can be expressed in CHP-SL with $(\exists x)(\forall y)\phi(x,y)$. 
 
Strategy logic (SL) \cite{MogaveroMV10} generalises CHP-SL, with the syntax as follows: 
$$
\phi \equiv p ~|~\neg \phi~|~ \phi \land \phi ~|~ \nxt \phi~|~\phi \until \phi~|~\atlop{v}\phi~|~\atluop{v}\phi~|~(i,v)\phi
$$
where $v\in \slvar$ such that $\slvar$ is a set of strategy variables, and $i\in\Ags$ is an agent. 
Intuitively, $\atlop{v}\phi$ says that there exists a strategy $v$ such that $\phi$, 
formula $\atluop{v}\phi$ says that $\phi$ holds for all strategies $v$, 
and $(i,v)\phi$ says that $\phi$ holds if agent $i$ plays  strategy $v$. 
A formula is a \emph{sentence} if every occurrence of $(a,x)$ is within the scope of an occurrence of $\atlop{x}$ or $\atluop{x}$, 
and every temporal subformula $ \nxt \phi$ or  $\phi \until \phi$ occurs within the context of some binding $(i,x)$, for every agent $i$. 
The ATL* formula $\atlop{1}{\phi}$ can be expressed in SL as  $\atlop{x}\atluop{y}(1,x) (2,y) \phi$. 

Let $\str$ be a set of agent strategies, 
and $\chi:\Ags\cup \slvar\rightarrow \str$ be a 
partial 
mapping from agents and variables to the set of strategies. 
Then, the semantics 
can be formulated with respect to our environments $E$
as follows.%
\footnote{We make some simplifications; \cite{MogaveroMV10}  distinguish between path and state formulas.} 
\begin{itemize}
\item $E,\chi,(r,m)\models \atlop{v}\phi $ iff there exists a strategy $\strat \in \str$ such that $E,\chi[v\mapsto \strat], s\models \phi$; 
\item $E,\chi,(r,m)\models \atluop{v}\phi $ iff for all strategies $\strat \in \str$ it holds that $E,\chi[v\mapsto \strat], s\models \phi$; 
\item $E,\chi,(r,m) \models (i,v)\phi $ iff $E,\chi[i\mapsto \chi(v)],(r',m)\models \phi$ for all runs $r'$ where $r(m)=r'(m)$ and 
$r'$ is a run consistent with $\chi[i\mapsto \chi(v)]$ from time $m$. 
\end{itemize}
Atomic, boolean and temporal formulas are handled as usual. We remark that because (1) the transition relation is assumed in SL to be deterministic, 
i.e., $\trans$ can be written as a function of type $S \times \Acts \rightarrow S$,
and (2)  temporal operators in a sentence appear only in contexts where every agent is bound to a strategy,  the final binding $(i,v)$ before temporal 
operators are evaluated in fact quantifies over just a single run.  	

Given an SL formula $\phi$, we let $V(\phi)$ be the set of variables in the operators $\atlop{}$ or $\atluop{}$. 
 SL allows the assignment of a strategy to multiple agents, e.g., in formula $\atlop{v}((i,v)\phi_1\land (j,v)\phi_2)$ the agents $i$ and $j$ have the same strategy represented in the variable $v$. For this 
to make sense 
in an imperfect information system, 
without allowing implausible bindings or artificially complex interpretations of quantification, 
all agents need to have the same actions and
the same observations. This does not match the setting of our framework particularly well.  We remark that CHP-SL does not allow this expressivity, as players 1 and 2 are associated with their dedicated strategy variables $x$ and $y$, respectively. 

In the following, we consider the fragment of SL in which 
every variable $v_i$ is uniquely associated with an agent $i\in\Ags$, so that $v_i$ occurs only in bindings $(j,v_i)\psi $ with $ j=i$. 
 Then we can 
translate a SL formula $\phi$ into an $\ESL$ formula $\trlsl{\assignedAgents}{\phi}$ as follows. 
\begin{align*} 
 & \trlsl{\assignedAgents}{p} = p \\
 &  \trlsl{\assignedAgents}{(\neg \phi)} = \neg \trlsl{\assignedAgents}{\phi} \\
 & \trlsl{\assignedAgents}{(\phi_1\land \phi_2)} = \trlsl{\assignedAgents}{\phi_1} \land \trlsl{\assignedAgents}{\phi_2} \\ 
 & \trlsl{\assignedAgents}{(\nxt \phi)} = \nxt \trlsl{\assignedAgents}{\phi} \\ 
 & \trlsl{\assignedAgents}{(\phi_1 \until \phi_2)} = \trlsl{\assignedAgents}{\phi_1} \until \trlsl{\assignedAgents}{\phi_2} \\
& \trlsl{\assignedAgents}{(\atlop{v_i}\phi)} = \exists v_i \trlsl{\assignedAgents}{\phi} \\
& \trlsl{\assignedAgents}{(\atluop{v_i}\phi)} = \forall v_i \trlsl{\assignedAgents}{\phi} \\
& \trlsl{\assignedAgents}{((i,v_i)\phi)} = D_{e\cup\strat(\Ags \setminus \{i\}) } ( \lid{\strat(i)}{v_i} \rimp  \trlsl{\assignedAgents}{\phi} ) \\
\end{align*} 
Intuitively, to decide if $\atlop{v_i}\phi$, we need to determine the existence of a
strategy $v_i$ with respect to 
which the formula $\phi$ is satisfied. 
In the $\ESL$ translation, $v_i$ refers to a global state rather than a strategy, but the only component of this global state that is used in 
the remainder of the evaluation is the component $\strat(i)$, which picks out a strategy for agent $i$. 
Similarly for $\atluop{v_i}\phi$. 
To decide if $(i,v_i)\phi$, we need to 
satisfy 
$\phi$ on 
(all) runs where agent $i$'s strategy is switched to that represented in $v_i$. 
The translation handles this using the operator $D_{e\cup\strat(\Ags \setminus \{i\}) } $, which 
refers to points in which the state of the environment and the strategies of all agents are fixed, 
while the strategy of agent $i$ is allowed to vary. The assertion $ \lid{\strat(i)}{v_i} $ checks that the
strategy of agent $i$ is in fact switched to that represented in the global state $v_i$. 

Similarly, CHP-SL formulas can be translated into \ESL\ formulas as follows. 
\begin{align*} 
& \trlsl{\assignedAgents}{(\exists x.\phi)} = \exists 
x 
\trlsl{\assignedAgents}{\phi} \\
& \trlsl{\assignedAgents}{(\forall x.\phi)} = \forall 
x
\trlsl{\assignedAgents}{\phi} \\
& \trlsl{\assignedAgents}{(\phi(x,y))} = D_{e} ( 
\lid{\strat(1)}{x}\land \lid{\strat(2)}{y} 
\rimp  \trlsl{\assignedAgents}{\phi} ) \\
\end{align*} 

Finally, we remark that both the CHP-SL semantics in \cite{CHP10}  and the SL semantics in \cite{MogaveroMV10}  are  for perfect recall. 
Since we have formulated $\ESL$ for imperfect recall, we leave the above translations as indicative rather than attempting a formal proof.

\subsection{Game Theoretic Solution Concepts}\label{sec:games} 
 
It has been shown for a number of logics for strategic reasoning that they are expressive enough
to state a variety of game theoretic solution concepts, e.g.,  \cite{CATL,CHP10} show that Nash Equilibrium is expressible. 
We  now sketch the main ideas required to show that the fragment $\CTLK(\Ags\cup \strat(\Ags)\cup \{e\}, 
Prop)$ of 
our framework also has this expressive power. We assume two players $\Ags=\{0,1\}$ in a normal form 
perfect information 
game, 
and assume that these agents play a deterministic strategy. 
The results in this section can be easily generalized to multiple players and extensive form games.

Given a game $\mathcal{ G}$ we construct an environment $\Env_\mathcal{ G}$ that represents the game. 
Each player has a set of actions that correspond to the moves that the player can make. 
We assume that  $\Env_\mathcal{ G}$ is constructed to model the game so that play happens in the first step from a 
unique initial state,  and that subsequent transitions do not change the state. 
We let agents have perfect information in $E_\mathcal{ G}$, i.e., we define the observation of agent $i$ in state $s$ 
by $O_i(s) = s$. (Consequently, although we use uniform strategies $\Strats^{\unif,\det}$ below, the uniformity constraint 
is vacuous in these environments.) 

We write  $-i$ to denote the adversary of player $i$. 
Let $u_i$ for $i\in\{0,1\}$ be a variable denoting the utility gained by player $i$ when play is finished. 
Let $V_i$ be the set of possible values for $u_i$, 
and let $V = V_0\cup V_1$. 
We work with the following atomic propositions. 
Atomic proposition $u\leq v$, where $u,v\in V$, expresses 
the ordering on utilities. 
Atomic proposition $u_i =v$, where $i\in \{0,1\}$ and $v\in V_i$, expresses that 
player $i$'s utility has value $v$. 
We use formula  $$U_i(v)= \nxt(u_i=v)$$
to express that value $v$ is player $i$'s utility once play finishes.  

\noindent 
{\bf Nash equilibrium (NE)}   is a solution concept  that states 
that 
no player can gain by unilaterally changing their strategy.
We may write $$BR_i(v) =  U_i(v)\land K_{\strat(-i)}\bigwedge_{v'\in V_{i}}(U_{i}(v')\rimp 
v'\leq v)$$
to express that, given the current 
strategy $\strat(-i)$ of the adversary of $i$, 
the value $v$ attained by player $i$'s current strategy is the best possible utility attainable by player $i$, 
i.e., the present strategy of player $i$ is a best response to the adversary.  
Thus $$BR_i = \bigvee_{v\in V_i} BR_i(v)$$ says that player $i$ 
 is playing a best-response to the adversary's strategy. 
The following statement then expresses the existence of a 
(pure)  
Nash equilibrium 
for the game $\mathcal{ G}$:
$$\Env_\mathcal{ G}, \Strat^{\unif, \detstrat}(\Env_\mathcal{ G})\models  \neg D_\emptyset \neg (BR_0\land BR_1)~.$$
That is, in a Nash equilibrium, each player is playing a best response to the other's strategy. \\ 

\noindent 
{\bf Perfect cooperative equilibrium (PCE)} 
is a solution concept intended to overcome deficiencies of Nash equilibrium for explaining cooperative behaviour \cite{HR2010}. 
It says that each player does at least as well as she would if the other player were best-responding.  The following formula
$$BU_i(v) =  D_\emptyset \left( \bigwedge_{v'\in V_i}((BR_{-i} \land U_i(v')) \rimp 
v'\leq v)\right)$$
states that $v$ is as good as any utility that $i$ can obtain if the adversary always best-responds to whatever $i$ plays. 
Thus,  
$$BU_i =  \bigvee_{v\in V_i} (U_i(v) \land BU_i(v))$$ 
says that $i$ is currently getting a utility as good 
as 
the best utility that 
$i$ 
can obtain if the adversary is a best-responder. 
Now, the following formula expresses the existence of perfect cooperative  equilibrium  for the game $\mathcal{ G}$: 
$$\Env_\mathcal{ G}, \Strat^{\unif, \detstrat}(\Env_\mathcal{ G})
\models \neg D_\emptyset \neg ( BU_0 \land BU_1 )$$ 
That is, in a PCE, no player has an incentive to change their strategy, on the assumption that the 
adversary will best-respond to any change.

\subsection{Computer Security Example: Erasure policies} \label{sec:erasure}

Formal definitions of computer security frequently involve reference to 
the strategies available to the players, and to 
agents' reasoning based on these
strategies. In this section we sketch an example that illustrates 
how our framework might be applied in this context. 

Consider the scenario  depicted in the  following diagram: 

\begin{figure}[h]  
\centerline{\includegraphics[height=4cm]{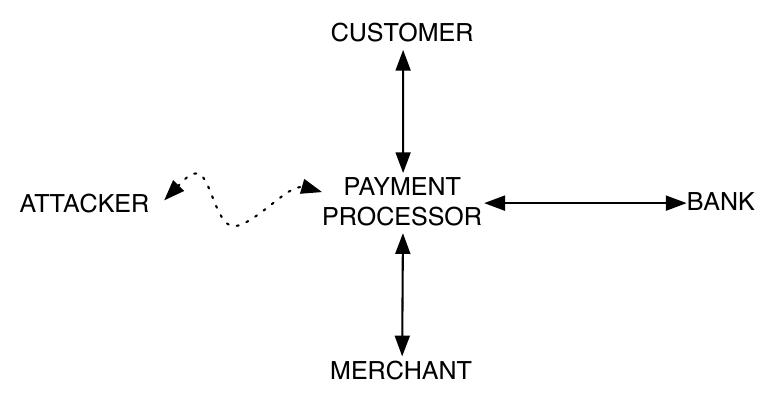} } 
\end{figure} 

\noindent
A customer $C$ can purchase 
items at a web merchant $M$. Payment is handled by a trusted payment processor $P$ (this could be a service or device), 
which interacts with the customer, merchant, and a bank $B$ 
to securely process the payment. (To keep the example simple, we suppose that the customer and merchant use the same bank). 
One of the guarantees provided by the 
payment processor is to protect the customer from attacks on the customer's 
credit card by the merchant: the specification for the protocol that runs the transaction 
requires that the merchant should not obtain the customer's credit card number. 
In fact, the specification for the payment processor is that after the transaction has 
been successfully completed, the payment processor should {\em erase} the 
credit card data, to ensure that even the payment 
processor's state does not contain information about the customer's credit card number. 
The purpose of this constraint is to protect the customer against subsequent attacks 
by an attacker $A$, who may be able to use vulnerabilities in the payment processor's
software to obtain access to the payment processor's state. 

We sketch how one might use our framework to express the specification. 
To capture reasoning about all possible behaviours of the agents, and
what they can deduce from knowledge of those behaviours, we work in $\Iunif(\Env)$ for a 
suitably defined environment $\Env$. To simplify matters, we take $\Ags = \{C, M, P, A\}$. 
We exclude the strategy of the bank from consideration: this amounts to assuming that the bank has no actions and is 
trusted to run a fixed protocol. We similarly assume that the payment processor $P$ has no actions, 
but in order to talk about what information is encoded in the payment processor's local state, 
we do allow that this agent has observations. The customer $C$ 
 may have actions such as entering the
credit card number in a web form, pressing a button to submit the form
to the payment processor, and pressing a button to approve or cancel the transaction. 
The customer observes variable $\cc$, which records the credit card number drawn from a set $\CCN$, and
boolean variable $\done$ which records whether the transaction is complete (which could mean either committed or aborted).  

We assume that the attacker $A$ has some set of exploit actions, as well as some 
innocuous actions (e.g., setting a local variable or performing $\skipp$). 
The effect of the exploit actions is to exploit a vulnerability in the payment processor's software and
copy parts of the local state of the payment processor to variables that are observable by the 
attacker. We include in the environment state a boolean variable $\exploited$, 
which records whether the attacker has executed an exploit action at some time in the past. 
The merchant $M$ may have actions such as sending cost information to the 
payment processor and acknowledging a receipt certifying that payment has been 
approved by the  bank (we suppose this receipt  is transmitted from the 
bank to the merchant via the payment processor). 

We may then capture the statement that the system is {\em potentially} vulnerable to an 
attack that exploits an erasure flaw in the implementation of the 
payment processor, by the following formula: 
\[ 
\neg D_\emptyset \neg  (  \done \land  \bigvee_{x\in \CCN}  K_{P} (  \cc \neq  x)) 
\] 
where $ \cc \neq  x$ is an atomic proposition for each $x\in\CCN$, with the obvious meaning
that the customer's credit card number is not $x$. 
This says that there exist behaviours of the agents, which can 
(at least on some points in some runs) leave the payment processor in a state 
where the customer has received confirmation that the transaction is done, but 
in which the payment processor's local state somehow still encodes {\em some} information  about the 
customer's credit card number. This encoding could be direct (e.g., 
by having a variable $\mathit{customer\_cc}$ that still stores the credit 
card number) or indirect (e.g. by the local state including both a 
symmetric encryption key $K$ and an encrypted version of the credit card
number, $\mathit{enc\_customer\_cc}$, with value $\mathtt{Encrypt}_K(\cc)$
that was used for secure transmission to the bank). 
Note that for a breach of security, it 
is only required that 
the information suffices to {\em rule out } some credit card number 
(so that,  e.g.,  knowing the 
first digit of the number would constitute a vulnerability) 

The vulnerability captured by this formula is only potential, because
it does not necessarily follow that the attacker is able to obtain the credit 
card information. Whether this is possible can be checked 
using the formula 
\[ 
\neg D_\emptyset \neg  (  \done \land  \neg \exploited \land 
E\Diamond
\bigvee_{x\in \CCN}  D_{\{ A, \strat(A)\}} (  \cc \neq x)) 
\] 
which says that it is possible for the attacker to obtain information about the credit card number 
even after the transaction is done. (To focus on erasure flaws, we deliberately wish to exclude here
the possibility that the attack occurs during the processing of the transaction.) 
Note that here we assume that the attacker knows 
his
own strategy when making deductions from the 
information obtained in the attack. This is necessary, because the attacker can typically write 
his
own local variables, so it needs to be able to distinguish between a value it wrote itself and
a value it copied from the payment processor. 

However, even this formula may not be sufficiently strong. Suppose that the payment processor
implements erasure by 
writing,   to its variable $\mathit{customer\_cc}$, a random value. 
Then, even if the attacker obtains a copy of this value, and it happens to be equal to the customer's
actual credit card number, the attacker would not have any knowledge about the credit card number, 
since, as far as the attacker knows, it could be looking at a randomly assigned number. 
However, there may still be vulnerabilities in the system. Suppose that the implementation of 
the payment processor operates so that the customer's credit card data is not erased by randomization until 
the merchant has acknowledged the receipt of payment from the bank, but to 
avoid annoying the customer with a hanging transaction, the customer is advised that the
transaction is approved (setting $\done$ true) if the merchant does not respond within a certain 
time limit. It is still the case that on  observing the copied value of $\mathit{customer\_cc} $, 
the attacker would not be able to deduce that this is the customer's credit card number, since it might be 
the result of erasure in the case that the merchant responded promptly. However, if the attacker knows that the merchant has not
acknowledged the receipt, the attacker can then deduce that the value  is not due to erasure. 
One way in which the attacker might know that the merchant has not acknowledged receipt
is that the attacker is in collusion with the merchant, who has agreed to omit sending 
the required acknowledgement messages. 

This type of attack can be captured by
replacing the term $D_{\{ A, \strat(A)\}} (  \cc \neq x)$
by $D_{\{ A, \strat(A), \strat(M)\}} (  \cc \neq x)$, 
capturing  that the attacker reasons using knowledge of both its own 
strategy as well as the strategy of the merchant, 
or even $D_{\{ A, \strat(A), \strat(M), M\}} (  \cc \neq x)$ for a collusion in which the
merchant shares information observed. Similarly, to focus on erasure flaws in the 
implementation of the payment gateway, independently of the attackers capability, 
we would replace the term $K_P( \cc \neq x)$ above by 
$D_{\{P, \sigma(M)\}}(\cc \neq x)$. 

We remark that in the case of the attacker's knowledge, 
it would be appropriate to work with a perfect recall semantics of 
knowledge, but when using knowledge operators to express information 
in the payment gateway's state for purposes of reasoning about erasure 
policy, the more appropriate semantics of knowledge is imperfect recall.

This example illustrates some of the subtleties that arise in the 
setting of reasoning about security and the way that our framework helps to 
represent them. Erasure policies have previously been studied in the 
computer security literature, beginning with \cite{ChongM05}, though generally without 
consideration of strategic behaviour by the adversary. 


\subsection{Reasoning about Knowledge-Based Programs} \label{sec:KBP} 

Knowledge-based programs \cite{FHMV1997} are a form of specification 
of a multi-agent system in the form of a program structure that 
describes
how an agent's  actions are related to its knowledge. They have been
shown to be a useful abstraction for several areas of application, including
the development of optimal protocols for distributed systems \cite{FHMV1997}, 
robot motion planning \cite{BrafmanLMS1997}, and game theoretic reasoning \cite{HalpernMosesIJCAI}. 

Knowledge-based programs cannot be directly executed, since there is a circularity in their semantics: 
which actions are performed depends on what the agents know, which in turn 
depends on which actions the agents perform. The circularity is not
vicious, and can be resolved by means of a fixed point semantics, 
but it means that a  knowledge-based program may have multiple 
distinct implementations (or none), and the problem of reasoning about these implementations 
is quite subtle. In this section, we show that our framework can capture
reasoning about the set of possible implementations of a knowledge-based program.

\newcommand{\did}{\mathit{did}} 

We consider joint knowledge-based programs $P$ (as defined by \cite{FHMV1997}) where for each 
agent $i$ we have a knowledge-based program 
$$P_i = {\bf do} ~\phi^i_1 \rightarrow a^i_1~ [] ~ \ldots [] ~\phi^i_{n_i } \rightarrow a^i_{n_i}~{\bf od}$$
where each $\phi^i_j$ is a formula of $\CTLsK(\Ags, \Prop)$ of the form $K_i \psi$, and each $a_i$ appears just once.
The formulas $\phi^i_j$ are called the \emph{guards} of the knowledge-based program.%
\footnote{The guards in 
\cite{FHMV1997} are allowed to be boolean combinations of formulas $K_i \psi$ and propositions $p$ local to the agent: 
since for such propositions $p \dimp K_ip$, 
and the operator $K_i$ satisfies positive and negative introspection, our form for the guards is equally general. 
They do not require that $a_i$ appears just once, but the program can always be put into this 
form by aggregating clauses for $a_i$ into one and taking the disjunction of the guards.} 
Intuitively, this program says to repeat forever the following operation: nondeterministically execute one of the 
actions $a^i_j$ such that the corresponding guard $\phi^i_j$ is true. To ensure that 
it is always the case that at least one action enabled, we assume that $\phi^i_1 \lor \ldots \lor \phi^i_{n_i }$ 
is a valid formula; this can always be ensured by taking the last condition $\phi^i_{n_i }$ to be 
the ``otherwise'' condition $K_i \neg (\phi^i_1 \lor \ldots \lor \phi^i_{n_i -1})$, which is equivalent to 
$\neg( \phi^i_1 \lor \ldots \lor \phi^i_{n_i -1})$ by introspection. 
In general,  the guards in a knowledge based program may contain common knowledge 
operators $C_G$, but we assume for technical reasons (explained below) that no $\phi^i_j$ contains 
such an operator.  

We present a formulation of semantics for knowledge-based programs that 
refactors the definitions of \cite{FHMV1997}, following the approach 
of \cite{MeydenTARK96} which uses the 
notion of environment defined above rather than the original notion of {\em context}. 
A {\em potential implementation} of a knowledge-based program $P$ in an environment $\Env$
is a joint strategy $\sgy$  in $\Env$. 
(Recall that we use ``joint strategy'' to refer to a group strategy for the group of all agents.)
Given  a potential implementation $\sgy$  in $\Env$, 
we can construct the interpreted system 
$\I_\sgy= \I(\Env,\{\sgy\})$, 
which captures the possible runs of $\Env$ when 
the agents choose their actions according to the single possible joint strategy $\sgy$. 
Given this interpreted system, we can now interpret the epistemic guards in $P$. 
Say that a state $s$ of $\Env$ is $\sgy$-reachable if there is a point $(r,m)$ of 
$\I_\sgy$ 
with $r_e(m) = s$. 
We note that for a formula $K_i \phi$, and a point $(r,m)$ of  $\I_\sgy$, 
the statement $\I_\sgy,(r,m) \models K_i \phi$ depends only on the state $r_e(m)$
of the environment at $(r,m)$. 
Recall that $r_e(m)$ determines $r_i(m)$ for $i \in \Ags$. 
For an $\sgy$-reachable state $s$ of  $\Env$, 
it therefore makes sense to define satisfaction of $K_i\phi$ at $s$ rather than at a point, by 
$\I_\sgy,s \models K_i \phi$ if  $\I_\sgy,(r,m) \models K_i \phi$ for all $(r,m)$ with $r_e(m) = s$. 
We define 
a joint strategy 
$\sgy$ to be an {\em implementation} of $P$ in $\Env$ if for all 
$\sgy$-reachable states $s$ of $\Env$ and agents $i$, we have 
$$ \sgy_i(s) = \{ a^i_j ~|~  1\leq j\leq n_i,~~\I_\sgy,s \models \phi^i_j\}~.$$
Intuitively, the right hand side of this equation is the set of actions that are 
enabled at $s$ by $P_i$ when the tests for knowledge are interpreted using the system 
obtained by running the strategy $\sgy$ itself. 
The condition states that the strategy is an implementation 
if it enables precisely this set of actions at every reachable state. 
It is easily checked that a strategy $\alpha_i$ satisfying the above equation is uniform. 

We now show that our framework for strategic reasoning can express the 
same content as a knowledge-based program by means of a formula, and that this 
enables the framework to be used for reasoning about knowledge-based program 
implementations. 
In general, implementations $\alpha$ of a knowledge based program $P$ can be hard to find, and
there may be one, many or no implementations of a given knowledge based program. 
We therefore work in strategy space $\I(E, \Strat^\unif)$, which contains all candidate 
implementations, and develop a formula $\mathbf{imp}(P)$ such that 
for a given run $r$, the formula $\mathbf{imp}(P)$  holds at a point of $r$ iff 
the joint strategy encoded in $r$ is an implementation of $P$ in $E$.

We need one constraint on the environment. 
Say that an environment $\Env$ is {\em action-recording} if 
for all agents $i$,  
for each $a\in \Acts_i$
there exists an  atomic  proposition $\did_i(a)$ such that  for $s\in I$ we have $\did_i(a) \not \in \pi(s)$ 
and for all states $s,t$ and joint actions $a$ such that $(s,a,t) \in \trans$, 
we have $\did_i(b) \in \pi(t)$ iff $b= a_i$.
Intuitively, this means that we can determine from a non-initial state the joint action 
that was executed in reaching that state. 
It is easily seen that any environment can be made action-recording, just by adding 
a component to the states that records the latest joint action. 

We can now express  knowledge-based program 
implementations as follows.  The main issue that we need to 
deal with is that the semantics of knowledge formulas in knowledge-based programs 
is given with respect to a system $\I_\alpha$, in which it 
is {\em common knowledge} that the joint strategy in use is $\alpha$. 
In general, strategies are not common knowledge in the strategy space 
$\I(\Env, \Strat^\unif)$
within which we wish to reason about knowledge-based program implementations.  
We handle this by means of a transformation of formulas. 

\newcommand{\trlk}[1]{#1^\$}

For a formula 
$\phi$ of $\CTLsK(\Ags,\Prop)$, 
not containing common knowledge operators,
write $\trlk{\phi}$ for the formula of $\ESL$ 
(in fact, of 
$\CTLsK(\Ags \cup \strat(\Ags), \Prop)$)
obtained from the
following recursively defined transformation: 
\begin{align*} 
& \trlk{p} = p \\ 
&  \trlk{(\neg \phi)} = \neg \trlk{\phi} \\
& \trlk{(\phi_1\land \phi_2)} = \trlk{\phi_1} \land \trlk{\phi_2} \\ 
&  \trlk{(D_G\phi)}  = D_{G\cup \strat(\Ags)}\, \trlk{\phi}  \\ 
& 
\trlk{(A \phi)} = A \trlk{\phi} \\
& 
\trlk{ (\nxt \, \phi) } = \nxt ~ \trlk{\phi} \\
& \trlk{(\phi_1 \until \phi_2)} =  (\trlk{\phi_1}\,  \until \, \trlk{\phi_2})
\end{align*}

Intuitively, this substitution says that knowledge operators  in $\phi$ are to be interpreted as if 
it is  known that the current joint strategy is being played. In the case of an operator 
$D_G$, 
which includes the special case 
$K_i = D_{\{i\}}$, 
the 
translation handles this by adding $\strat(\Ags)$ to the set of agents that are kept fixed when moving through the
indistinguishability relation.
%
%

Let  $${\bf imp}(P) = D_{\strat(\Ags)} ( \bigwedge_{i \in \Ags, j =1 \ldots n_i} ((\phi^i_j)^\$ \dimp 
E\nxt
\did_i(a^i_j))).$$
Intuitively, this  formula says that the current joint strategy gives an implementation of the knowledge-based program $P$. 
More precisely, we have the following: 

\begin{propn} \label{prop:kbp}
Suppose that $P$ is a knowledge-based program
in which the guards do not contain common knowledge operators. 
Let $\sgy$ be a locally uniform joint strategy in $\Env$ and let $r$ 
be a run of 
$\I(E,\Strat^\unif(E))$, 
in which the agents are running 
joint 
strategy $\sgy$, i.e., 
$r(0) = (s, \alpha)$ for some state $s$. 
Let $m\in \nat$.
Then 
$$\I(E,\Strat^\unif(E)), (r,m) \models {\bf imp}(P)$$
iff the strategy $\sgy$ is an implementation of knowledge-based program $P$ in $\Env$. 
\end{propn} 

\begin{proof} 
For brevity, we write just $\I$ for $\I(E,\Strat^\unif(E))$. 
First, we claim that for a formula $\phi$ not containing common knowledge operators, 
we have  $\I,(r,m) \models \trlk{\phi}$ iff $\I_\alpha, (r,m) \models \phi$, where $r(m) = (s,\alpha)$. 
The proof is by induction on the construction of $\phi$. The base case of atomic propositions, and the cases
for boolean and linear temporal operators are straightforward. 

Consider the case  $\phi = A\psi$, where we have  $\trlk{(A\psi)} = A(\trlk{\psi})$. 
Observe that if $r$ and $r'$ are runs of $\I$, with $r[0\ldots m] = r'[0\ldots m]$, then 
$r$ and $r'$ encode the same strategy $\alpha = r_{\strat(\Ags)}(0)$. 
Now $\I,(r,m) \models A(\trlk{\psi})$ iff $\I,(r',m)\models \trlk{\psi}$ for all runs $r'$ of $\I$ with 
 $r[0\ldots m] = r'[0\ldots m]$. By the observation, this is 
 equivalent to 
$\I,(r',m)\models \trlk{\psi}$ for all runs $r'$ of $\I_\alpha$ with 
 $r[0\ldots m] = r'[0\ldots m]$.
 By induction, the latter is equivalent to 
 $\I_\alpha,(r',m)\models \psi$ for all runs $r'$ of $\I_\alpha$ with 
 $r[0\ldots m] = r'[0\ldots m]$, i.e., to $\I_\alpha,(r,m) \models A\psi$.  
 Hence $\I,(r,m) \models \trlk{(A\psi)}$ iff $\I_\alpha,(r,m) \models A\psi$.

Finally, consider the case  $\phi = D_G\psi$, where we have  $\trlk{(D_G\psi)} = D_{G\cup \strat(\Ags)} (\trlk{\psi})$. 
Observe that if $(r,m)$ and $(r',m')$ are points of $\I$ with $(r,m) \sim_{G\cup\strat(\Ags)}  (r',m')$, then 
$r$ and $r'$ encode the same strategy $\alpha = r_{\strat(\Ags)}(0)$ and $(r,m) \sim_{G}  (r',m')$. 
Conversely, if $(r,m)$ and $(r',m')$ are points of $\I_\alpha$, i.e., both encode joint strategy $\alpha$, 
then  $(r,m) \sim_{G}  (r',m')$ implies $(r,m) \sim_{G\cup\strat(\Ags)}  (r',m')$. 
Now $\I,(r,m) \models D_{G\cup \strat(\Ags)} (\trlk{\psi})$ iff $\I,(r',m')\models \trlk{\psi}$ for all points $(r',m')$ of $\I$ with 
 $(r,m) \sim_{G\cup\strat(\Ags)}  (r',m')$. By the observation, this is 
 equivalent to 
$\I,(r',m')\models \trlk{\psi}$ for all points $(r',m')$ of $\I_\alpha$ with 
 $(r,m) \sim_{G}  (r',m')$.
 By induction, this is equivalent to 
 $\I_\alpha,(r',m')\models \psi$ for all points $(r',m')$ of $\I_\alpha$ with 
 $(r,m) \sim_{G}  (r',m')$, i.e., to $\I_\alpha,(r,m) \models D_G\psi$.  
 
This completes the proof of the claim.  Next, note that, for a point $(r,m)$ with $r(m) = (s,\alpha)$, 
for action $\alpha^i_j$ of agent $i$, we have $\I,(r,m) \models E\nxt \did_i(a^i_j)$ iff $a^i_j \in \alpha_i(s)$. 

Suppose that $\alpha$ is an implementation of $P$ in $E$, and let $(r,m)$ be a point of $\I$ with $r(m) = (s,\alpha)$, 
We show that $\I,(r,m) \models \mathbf{imp}(P)$. For this, we let $(r',m')$ be a 
point with $(r',m') \sim_{\strat(\Ags)} (r,m)$, and show that 
$\I,(r',m') \models  \bigwedge_{i \in \Ags, j =1 \ldots n_i} ((\phi^i_j)^\$ \dimp  E\nxt \did_i(a^i_j))$.
From $(r',m') \sim_{\strat(\Ags)} (r,m)$ it follows that $r'(m') = (t,\alpha)$ for some state $t$ of $E$. 
Thus, from what was noted above, $\I,(r',m') \models E\nxt \did_i(a^i_j)$ iff $a^i_j \in \alpha_i(t)$. 
Since $\alpha$ is an implementation of $P$ in $E$, this holds iff 
$\I_\alpha,(r',m') \models \phi^i_j$. By the claim proved above, $\I_\alpha,(r',m') \models \phi^i_j$ is 
equivalent to $\I,(r',m') \models \trlk{(\phi^i_j)}$.
Thus, we have that  $\I,(r',m') \models \trlk{(\phi^i_j)} \dimp E\nxt \did_i(a^i_j)$. 
It follows that $\I,(r,m) \models \mathbf{imp}(P)$. 

Conversely, suppose that $\I,(r,m) \models \mathbf{imp}(P)$, and let $r(m) = (s,\alpha)$. 
We show that $\alpha$ is an implementation of $P$ in $E$. Let
$t$ be any $\alpha$-reachable state, with, in particular, $(r',m')$ a point of $\I_\alpha$ 
with $r'(m') = (t,\alpha)$. 
We need to show that for all agents $i$, we have 
$$ \sgy_i(t) = \{ a^i_j ~|~  1\leq j\leq n_i,~~\I_\sgy,t \models \phi^i_j\}$$
i.e., that for all $i,j$ we have $a^i_j \in \alpha_i(t)$ iff $\I_\sgy,t \models \phi^i_j$.
Note that $(r,m) \sim_{\strat(\Ags)} (r',m')$, so we have that 
$$\I,(r',m') \models  \bigwedge_{i \in \Ags, j =1 \ldots n_i} (\trlk{(\phi^i_j)} \dimp E\nxt \did_i(a^i_j)).$$
As in the previous paragraph, 
$a^i_j \in \alpha_i(t)$ iff $\I,(r',m') \models E\nxt \did_i(a^i_j)$, 
which is equivalent to  $\I,(r',m') \models \trlk{(\phi^i_j)}$, 
and by the claim proved above, equivalent to 
 $\I_\alpha,(r',m') \models \phi^i_j$, 
 i.e.,  $\I_\alpha,t \models \phi^i_j$. 
 Thus $a^i_j \in \alpha_i(t)$ iff $\I_\alpha,t \models \phi^i_j$, for all $i,j$, which is what we needed
 to prove.  
\end{proof} 

In particular, as  a consequence of this result, it follows that several properties 
of knowledge-based programs  (that do not make use of common knowledge operators)
can be expressed in the system 
$\I(E,\Strat^\unif(E))$: 
\be 
\item The statement that there exists an implementation of $P$ in $\Env$ can be expressed by 
$$ 
\I(E,\Strat^\unif(E))
\models \neg D_{\emptyset}\neg  {\bf imp}(P) $$ 
 \item The statement that  all implementations of $P$ in $\Env$ 
 guarantee that formula $\phi$ of $\CTLsK(\Ags, \Prop)$ (which may contain knowledge operators) holds at all times can be expressed 
 by $$  
 \I(E,\Strat^\unif(E))
 \models D_{\emptyset}(  {\bf imp}(P) \rimp 
 \phi^\$)$$ 
 \ee

We remark that as a consequence of these encodings and 
Theorem~\ref{mcCTLsK} (in section~\ref{sec:mc} below) that 
$\CTLsK(\Ags \cup \strat(\Ags), \Prop)$ 
model checking 
in strategy space 
is in PSPACE,  
we obtain the following result: 

\begin{cor} \label{cor:kbpcomplex}
The following are in PSPACE: 
\begin{enumerate} 
\item Given a finite environment $\Env$ and a knowledge based program $P$, determine if $P$ has an implementation in $\Env$. 
\item Given a finite environment $\Env$ and a knowledge based program $P$ and a  $\CTLsK(\Ags, \Prop)$  formula $\phi$, 
determine if $\I_\alpha \models \phi$ for all implementations $\alpha$ of $P$ in $\Env$. 
\end{enumerate} 
\end{cor} 


For testing existence 
(part 1 of Corollary~\ref{cor:kbpcomplex}), 
this result was known \cite{FHMV1997}, 
but the result on verification 
(part 2 of Corollary~\ref{cor:kbpcomplex}), 
has not previously been noted
(though it  could also have been shown using 
the techniques in~\cite{FHMV1997}.)  

One might expect that Proposition~\ref{prop:kbp} can be extended to 
knowledge based programs in which formulas may contain common knowledge 
operators, simply by adding the condition  
$$ \trlk{(C_G \phi)} = C_{G\cup \strat(\Ags)}\, \trlk{\phi}$$ 
to the transformation of formulas. However, this does not
work, because the interpretation of $C_G\phi$ in a subsystem $\I_\alpha$ 
is based on  chains of points $(r_0,m_0) \sim_{i_1}   (r_1,m_1) \sim_{i_2}  \ldots \sim_{i_k}  (r_k,m_k)$, 
such that $r_j$ is a run of joint strategy $\alpha$ for all $j=1\ldots k$. 
By contrast, the semantics  of $C_{G\cup \strat(\Ags)} \trlk{\phi}$ in $\I$ 
involves chains of points which are not required to preserve the joint strategy: 
rather each step preserves the local state of one of the agents in $G$ \emph{or} the strategy of one of the agents.  
Neither does it work to use the translation 
$$ \trlk{(C_G\phi)} = \exists x( \lid{\strat(\Ags)}{x} \land C_G(\lid{\strat(\Ags)}{x} \rimp \trlk{\phi}))$$
since the operator $C_G$ similarly does not preserve the joint strategy, and it is not 
enough to test only at the end of the chain that the joint strategy has been preserved. 

It is not clear that the translation we require for common knowledge 
is expressible in $\ESL$. What would work is to generalize the common knowledge 
operator to the form $C_X \phi$, where $X$ is a set of sets of agents (instead of a set of agents), 
and to define the semantics of this more general form as the greatest fixpoint of equation 
  $$ C_X \phi =  \bigwedge_{G\in X} D_G (\phi \land C_X\phi)~. $$
We could then use the translation 
$$ \trlk{(C_G \phi)} = C_{\{ \{i\} \cup \strat(\Ags) ~|~ i \in G\}}\, \trlk{\phi}~.$$ 
Here the semantics involves chains of points in which we preserve
the joint strategy and one of the agents in $G$. 
While this is an interesting extension, that we consider worthy of study, we do not pursue 
this as an ad hoc extension here, leaving it for future consideration in a broader context, such as a logic that
extends $\ESL$ by mu-calculus operators. 
   
\section{Model Checking}\label{sec:mc}

\newcommand{\stratdesc}{sd}
\newcommand{\stratconf}{\gamma}

Model checking is the problem of computing  whether a formula of a logic holds in a given model. 
We now consider the problem of model checking ESL and various of its fragments.

The model checking problem is to  determine whether $\Cont, \Env,\Strat\models \phi$ for a finite state environment $\Env$, a set $\Strats$ of strategies 
and a context $\Cont$, where  $\phi$ is an $\ESL$ formula. 

For purposes of results concerning the complexity of model checking, we need a measure
of the size of a finite environment. Conventionally, the size of a model is taken to be the 
length of a string that lists its components, and typically, this is polynomial in the 
number of states of the model. 
We note that in the 
case of environments, the set of labels $\Acts$ of the transition relation 
is an $n$-fold cartesian product, where $n = \Ags$, so 
(if the number of agents is a variable in the class of environments we consider) the size of the transition relation 
may be exponential in the number of 
agents.\footnote{
For certain classes of environments, we could address this 
by allowing  that the transition relation $\trans$ is presented in some notation with the property
that (1) given states $s,t$ and a joint action $a$, 
the representation of $\ptrans{}$ has size polynomial in the size of $|S|$ and $|\Acts|$, 
and (2) 
determining whether $s\ptrans{a} t$ is 
in PTIME
given $s,a,t$ and the representation of $\ptrans{}$. 
One example of a presentation format with this property is the 
class of {\em turn-based environments}, where at each state $s$, there exists an 
agent $i$ such that if $s\ptrans{a} t$ for a joint action $a$, then for all joint actions $b$
with $a_i = b_i$ we have $s\ptrans{b} t$. That is,  the set of states reachable
in a single transition from $s$ depends only on the action performed by agent $i$. 
In this case, the transition relation can be presented more succinctly 
as a subset of $S\times (\cup_{i\in \Ags} A_i)\times S$. 
While it would be interesting to consider the effect of such optimized representations on 
our complexity results, we do not pursue this here.}

However there is a more severe issue with respect to the parameter $\Strat$ 
of the model checking problem. A strategy for a single agent is a mapping from 
states to sets of actions of the agent. Hence the number of strategies we may need to 
list to describe $\Strat$  explicitly could be exponential in the 
number of \emph{states} of the environment, even in the case of a single agent. 
To address this issue, we abstract 
the strategy set 
 $\Strats$ to a parameterized 
class such that for each environment $\Env$, the set $\Strats(\Env)$ is a set of strategies for $\Env$. 
When $\mathcal{ C}$ is a complexity class, 
we say that the parameterized class $\Strats$  {\em can be presented in $\mathcal{ C}$}, 
 if the problem of determining,  
given an 
environment  $\Env$ and a joint strategy $\alpha$ for  $\Env$, whether  $\alpha \in \Strats(\Env)$,  is in complexity class $\mathcal{ C}$. 
For example, the class 
$\Strats$ of all strategies for  $\Env$ can be PTIME-presented, as can $\Strats^\unif$,  $\Strats^{\detstrat}$ and  $\Strats^{\unif,\detstrat}$. 

We first consider the complexity of model checking the full language $\ESL$.
The following result gives an upper bound of EXPSPACE for this problem.

\begin{theorem}\label{mcESLupper}
Let $\Strats$  be a parameterized   class of strategies 
that can be presented in  EXPSPACE.
The complexity of deciding, given an environment $\Env$, an \ESL\ formula $\phi$ 
and a context $\Cont$ for 
$\I(E, \Strats(E))$, defined on the free variables of $\phi$,  
whether $\Cont, \Env,\Strat(\Env)\models \phi$, is in EXPSPACE. 
\end{theorem}

\begin{proof} 
The problem can be reduced to that of model checking the temporal epistemic logic $\CTLsK$ 
obtained by omitting the constructs $\exists$ and $\lid{i}{x}$ from the language $\ESL$. 
This is known to be PSPACE-complete.%
\footnote{The result is stated explicitly in \cite{EGM2007:LFCS}, but techniques sufficient for a 
proof (involving guessing a labelling of states by knowledge  subformulas in order
to reduce the problem to LTL model checking and also verifying the guess
by LTL model checking) were already present in \cite{Vardi96tark-kbp}.
The branching operator $A$ can be treated as a knowledge operator for purposes of the proof.} 
The reduction 
involves 
an exponential blowup of size of both the formula and the environment, 
so we obtain an EXPSPACE upper bound. 

Model checking for temporal epistemic logic takes as input a formula and a
structure that is like an environment, except that its transitions are not based on 
a set of actions for the agents. More precisely, 
an \emph{epistemic transition system} for a set of agents $\Ags$ 
is a tuple $\ets = \langle S,I, \trans, \{O_i\}_{i \in \Ags}, \pi\rangle$, 
where $S$ is a set of states, $I\subseteq S$ is the set of initial states, 
$\trans \subseteq S\times S$ is a state transition relation, for 
each $i \in \Ags$, component $O_i:S\rightarrow L_i$ is a 
function giving an observation in some set $L_i$ for the agent $i$ at each state, 
and $\pi: S\rightarrow \powerset{\Prop}$ is a propositional assignment. 
A {\em run} of $\ets$ is a sequence $r: \nat \rightarrow S$
such that $r(0) \in I$ and $r(k) \trans r(k+1)$ for all $k \in \nat$. 
To ensure that every partial run can be completed to a run, we assume that 
the transition relation is {\em serial}, i.e., that for all states $s$ there exists a 
state $t$ such that $s\trans t$.

Given an epistemic transition system $\ets$, we define an interpreted system $\I(\ets) = (\R, \pi')$  
as follows. For a run $r:\nat \rightarrow S$ of $\ets$, define the lifted run $\hat{r}: \nat \rightarrow S \times \Pi_{i\in \Ags} L_i$ (here $L_e = S$), by 
$\hat{r}_e(m) = r(m)$ and  $\hat{r}_i(m) = O_i(r(m))$ for $i\in \Ags$. Then we take $\R$ to be the set of 
lifted runs $\hat{r}$ with $r$ a run of $\ets$. The assignment $\pi'$ is given by $\pi'(r,m) = \pi(r(m))$. 
The model checking problem for temporal epistemic logic 
$\CTLsK$
is to decide, 
given an epistemic transition system $\ets$ and a formula 
$\phi\in \CTLsK$, 
whether 
 $\I(\ets),(r,0)\models \phi$ for all  runs $r$ of $\I(\ets)$.

We now show how to reduce $\ESL$ model checking to $\CTLsK$ model checking. 
Given an environment $\Env =  \langle S, I, \Acts, \trans, \{O_i\}_{i\in \Ags}, \pi\rangle$
for  
$\ESL(\Ags, \Prop, \SVar)$, we first introduce a set of new propositions
$\Prop^*= \{p_{(s,\alpha)} ~|~ s\in S, ~\alpha\in \Sigma(E)\}$
which will be interpreted at global states of the generated interpreted system. 
Each proposition $p_{(s,\alpha)}$ will be true only at the global state $(s,\alpha)$.
These propositions will help to eliminate the constructs $\lid{i}{x}$ and $\exists x$. 
We then define the epistemic transition system 
$\ets =  \langle S^*, I^*, \trans^*, \{O^*_i\}_{i\in \Ags}, \pi^*\rangle$
for the language 
$\CTLsK(\Ags\cup \strat(\Ags), \Prop\cup \Prop^*, \SVar)$, 
in which the propositions have been extended by the set 
$\Prop^*$, as follows: 
\be
\item 
$S^* = \{ (s,\alpha) \in S\times \Sigma(\Env)~|~ \text{$s$ is reachable in $E$ using $\alpha$}\}$, 

\item $I^* = I\times \Sigma(\Env)$, 
\item $(s, \alpha) \trans^* (t, \beta)$ iff $s\ptrans{a}t$ 
(in $\Env$)
for some joint action $a$ 
and $\beta = \alpha$, 
\item $O_i^*(s,\alpha) = O_i(s)$ 
and $O_{\strat(i)}^*(s,\alpha) = \alpha_i$, 
for $i\in\Ags$,

\item $\pi^*(s,\alpha) = \pi(s) \cup \{p_{(s,\alpha)}\}$.  
\ee 
We can treat the states $(s,\alpha)\in S^*$ as tuples indexed by 
$\Ags\cup \strat(\Ags)\cup \{e\}$ by taking 
$(s,\alpha)_i = O_i(s)$ and  $(s,\alpha)_{\strat(i)} = \alpha_i$ for $i\in \Ags$, and $(s,\alpha)_e = s$. 

Note that  a joint strategy for an environment $\Env$ can be represented in space 
$\Sigma_{i\in \Ags}  |S|\times |\Acts_i|$, 
and the number of strategies is exponential in the space requirement. 
Thus, the size of $\ets$ is $O(2^{poly(|\Env|)})$. 
Note also that the construction of $\ets$ can be done in EXPSPACE
so long as verifying whether an individual strategy $\alpha$ is in $\Strat(\Env)$ 
can be done in EXPSPACE.

We also need a transformation of the formula. 
Given a formula $\phi$ of $\ESL$ and a context $\Cont$ for $\Env$, 
we define a formula $\phi^\Cont$, inductively, by 
\be
\item $p^\Cont = p$, for $p\in \Prop$, 
\item $\lid{i}{x}^\Cont = 
\bigvee \{ p_g ~|~ {g\in S^*, g_i =\Cont(x)_i} \}$  
\item $(\neg \phi)^\Cont = \neg \phi^\Cont$, ~ $(\phi_1\land \phi_2)^\Cont = \phi_1^\Cont \land \phi_2^\Cont$, 
\item $(\nxt \,\phi)^\Cont = \nxt \,(\phi^\Cont)$, ~ $(\phi_1\until \phi_2)^\Cont = (\phi_1^\Cont )\until (\phi_2^\Cont), 
~~(A\phi)^\Cont = A(\phi^\Cont)$ 
\item $(D_G\phi)^\Cont = D_G \phi^\Cont$, $(C_G\phi)^\Cont = C_G \phi^\Cont$,
\item $\exists x(\phi)^\Cont = 
\bigvee \{ \phi^{\Cont[g/x]} ~|~ {g\in S^*}\} $. 
\ee 
Plainly the size of $\phi^\Cont$ is $O(2^{poly(|E|,|\phi|)})$, and this
formula is in $\CTLsK(
\Ags
\cup \strat(\Ags), \Prop\cup \Prop^*)$. 
A straightforward inductive argument based on the semantics shows that \\
$\Cont, \Env, \Strat(\Env) \models \phi$ iff $\I(\ets) \models \phi^\Cont$. 
It therefore follows from the fact that  model checking $\CTLsK$ 
with respect to the observational semantics for knowledge is in 
PSPACE that $\ESL$ model checking is in EXPSPACE. 
\end{proof}

The following result shows that a restricted version of the model checking problem, 
where we consider systems with just one agent and uniform deterministic strategies is already EXPSPACE hard. 

\newcommand{\blnk}{\bot}  

\begin{theorem}\label{mcESLlower}
The problem  of deciding,  given an environment $\Env$ for a single agent, and an \ESL\ sentence $\phi$, 
whether $\Env,\Strat^{\unif,\det}(\Env)\models \phi$, is EXPSPACE-hard. 
\end{theorem}
 
\begin{proof} 
We show how polynomial size inputs to the problem can simulate exponential space deterministic Turing machine computations. 
Let $T= \langle Q, q_0, q_f, 
q_r,  
 A_I,A_T, \delta\rangle$ be a one-tape Turing machine solving an EXPSPACE-complete problem, 
with states $Q$, initial state $q_0$, final (accepting) state $q_f$, 
final (rejecting) state $q_r$, 
input alphabet $A_I$, tape alphabet $A_T\supseteq A_I$, 
and transition function $\delta: Q\times A_T\rightarrow Q\times A_T\times \{L,R\}$.   We assume that $T$ runs in space 
 $2^{p(n)}-2$ for a polynomial $p(n)$, and that the transition relation is defined so that the machine idles in 
 its final state 
 $q_f$
 on accepting, and idles in 
 state $q_r$
 on rejecting. The tape alphabet $A_T$ is assumed to 
 contain the blank symbol $\blnk$. 
 
Define
 $\C_{T,Q} = A_T \cup (A_T\times Q)$ to be the set of ``cell-symbols'' of $T$. 
We may represent a configuration of $T$ as a finite sequence 
over the set $C_{T,Q}$, containing exactly one element $(x,q)$ of $A_T\times Q$, 
representing a cell containing symbol $x$ where the machine's head 
is positioned, with the machine in state $q$.  For technical reasons, 
we pad configurations with a blank symbol to the left and right (so configurations 
take space $2^{p(n)}$), so that the initial configuration has the head at the second tape cell 
and,  without loss of generality, assume that the machine is designed so that it never moves the head to the initial 
or final 
padding blank. 
This means that the transition function $\delta$ can also be represented as a set of 
tuples 
$\Delta \subseteq C_{T,Q}^4$, 
such that $(a,b,c,d)\in \Delta$ iff, whenever the machine is in a configuration 
with $a,b,c$ at cells at positions $k-1,k,k+1$, respectively, the next configuration has $d$ at the cell at position $k$. 

Given the TM $T$  and a number $N = p(n)$ 
(for some polynomial $p$) 
we construct an environment $\Env_{T,N}$ such that
for every input word $w$, with $|w| = n$,  
there exists a sentence $\phi_w$  of size polynomial in $n$ 
such that $\Env_{T,N},\Strat^\unif(\Env_{T,N})\models \phi_w$ iff $T$ accepts $w$. 
The idea of the simulation,  
depicted in Figure~\ref{fig:EXPSPACE}, 
is to represent a 
computation 
of the Turing machine, using space $2^{N}$, 
by representing the  sequence of configurations of $T$ for the computation consecutively along a run $r$
of the environment $\Env_{T,N}$. 
(These runs travese a set of states we call $s/c$-states.) 
Each cell of a configuration will be encoded as a  block of 
$N+1$ consecutive moments of time in $r$.
In a block, the first of these moments represents the cell-symbol of the cell, and the 
remaining $N$ moments represent the position of the cell in the configuration, in binary. 
Not all runs of $\Env_{T,N}$ will correctly
encode a computation of the machine, so we use the formula to 
check whether a 
computation 
of $T$ has been correctly encoded in a given
run of $\Env_T$. 
In order to do so, the main difficulty is to check that corresponding cells of successive configurations represented along a run are updated correctly according to the 
yields relation of the Turing machine. For this, we need to be able to identify these corresponding cells, i.e. the cells with the same position number in the binary representation. 
For this, we use the behaviour of a strategy on an additional set of states ($t$-states) to give an alternate representation of a binary number, 
one that may be accessed in a formula by means of existential quantification. The formula then compares the representations of the binary number at two locations in the 
the $s/c$-run with the representation of the binary number in the strategy, in order to check that the numbers represented at the two locations in the $s/c$-run are the same. 
Details are given below.

\begin{figure}[h]  
\centerline{\includegraphics[height=12cm]{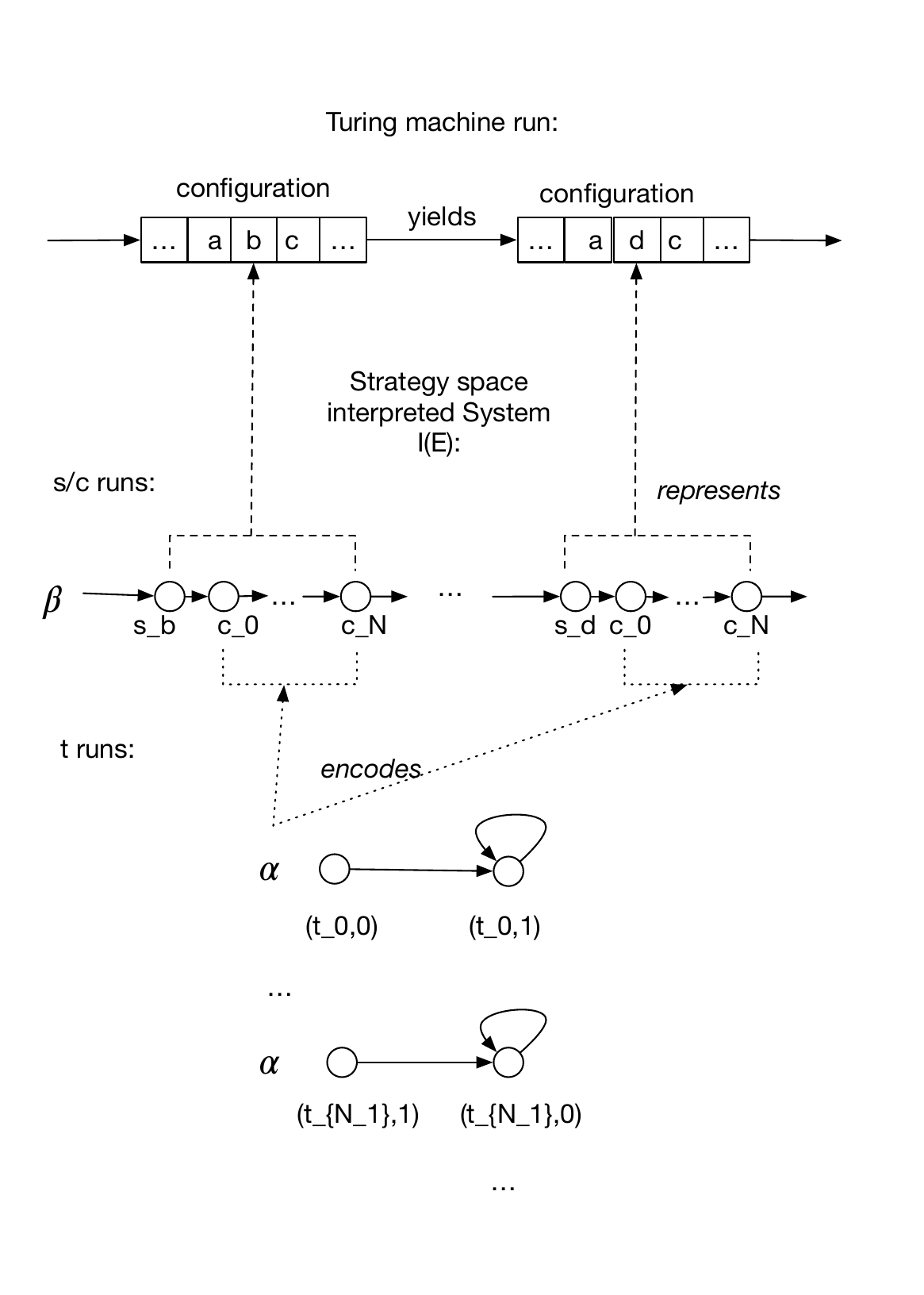} } 
\caption{Structure of the encoding
\label{fig:EXPSPACE}}
\end{figure} 

\newcommand{\cbit}{\mathtt{c}} 

The environment $\Env$ has propositions 
$C_{T,Q}\cup \{\cbit\}\cup \{t_0, \ldots t_{N-1}\}$.  
Propositions from 
$C_{T,Q}$ are used to represent cell  elements, 
and $\cbit$ is used to represents the bits of a counter that indicates
the position of the cell being represented. In particular, 
 a cell in a configuration,  at position 
$b_{N-1}\ldots b_{0}$, in binary, and containing symbol 
$a\in C_{T,Q}$, 
will be represented by a sequence of $N+1$ states,  
the first of which satisfies proposition $a$, such that for $i = 0\ldots N-1$, 
element $i+2$ of the sequence satisfies $\cbit$ iff $b_i =1$. 
(Thus, low order bits are represented to the left in the run.) 

We take the set of states of the environment to be 
$$S= \{ s_x~|x \in C_{T,Q} \}\cup \{c_0,c_1\}\cup \{(t_i,j)~|~ i = 0  \ldots N-1, ~j\in \{0,1\}\}~. $$ 
The set of initial states of the environment is defined to be $I = \{ s_\blnk \} \cup 
\{(t_i,j)~|~ i = 0  \ldots N-1, ~j\in \{0,1\}\}$. 
We define the assignment $\pi$ so that 
$\pi(s_a) = \{a\}$ for 
$a \in C_{T,Q}$, 
$\pi(c_0) = \emptyset$, $\pi(c_1) = \{\cbit\}$
and $\pi((t_i,0)) = \{t_i\}$ and $\pi((t_i,1)) = \{t_i\}  \cup \{\cbit\}$.

We take the set of actions of the single agent to be the set $\{\mathtt{a}_0,\mathtt{a}_1\}$.
The transition relation $\trans$ is defined so that for 
the only transitions are 
$$
\begin{array}{ll} 
 s_x \ptrans{a_k} c_i  & \\ 
 c_i \ptrans{a_k} c_j &  \\ 
 c_i \ptrans{a_k} s_x & \\ 
 (t_m,j) \ptrans{a_k} (t_m,k) \\ 
 \end{array}$$
 for  $x\in C_{T,Q}$ and  
$i,j,k\in \{0,1\}$ and   $m\in \{ 0\ldots N-1\}$. 
Intuitively, this forces the runs starting at state $s_\blnk$ to alternate between selecting 
a symbol from $C_{T,Q}$ and a sequence of bits $\{0,1\}$ 
for the counter. 
Note that for every sequence $\rho$ in $\bot \cdot \{c_0,c_1\}^+ \cdot  (C_{T,Q} \cdot \{c_0,c_1\}^+)^\omega$, 
and for every strategy $\alpha$ for the single agent, 
there exists a run $r$ with $r_{\strat(1)} = \alpha$ and $r_e[0 \ldots] = \rho$. 
For each 
$i = 0,\ldots N-1$, 
the 
states of the form $(t_i,j)$ for $j\in \{0,1\}$ 
form an isolated component in the 
transition relation, and are used to ensure that there is a sufficiently rich 
set of strategy choices for strategies to encode counter values.

The length of the counter sequence segments of a run generated by this transition
system can vary within the run, but we can use a formula of length 
$O(N)$
to state that these segments always 
have length $N$ wherever they appear in the run; let $\phi^N_{clock}$ be the formula  
$$ \Box (\alpha_{T,Q}  \rimp (\nxt^{N+1} \alpha_{T,Q} \land \bigwedge_{i=1\ldots N} \nxt^i \neg \alpha_{T,Q} ))$$
where we write $\alpha_{T,Q}$ for $\bigvee_{x \in C_{T,Q}} x$. 
By definition of the transition relation, this formula holds on  a run starting in state $s_\bot$ just when 
it consists of states of the form $s_x$ alternating with sequences of states of the form $c_i$
of length exactly $N$. 

The transition system generates arbitrary such sequences of states $c_i$ of length $N$, intuitively constituting
a guess for the correct counter value. 
Note that a temporal formula of length $O(N^2)$ can say that these guesses for the
counter values are correct, in that the counter values encoded along the run are $0, 1, 2,\ldots  2^N-1, 0, 1, 2,\ldots 2^N-1$ (etc). 
Specifically, this is achieved by the following formula $\phi^N_{count}$: 
$$ 
\phi_{zero}~~\land 
\Box \left ( 
\begin{array}{l} 
\alpha_{T,Q} \rimp 
 
\left(
\begin{array}{l} 
( \phi_{max}  \rimp \nxt^{N+1}(\phi_{zero})) \land \\[5pt] 
\bigwedge_{i=1\ldots N} ((\nxt c \land \ldots \nxt^{i-1} c \land \nxt^i \neg c) \rimp \\ 
~~~~~~~~~~~~~~~~~
\nxt^{N+1} ( \nxt \neg c \land \ldots  \land \nxt^{i-1} \neg c \land  \nxt^i c)  \\
~~~~~~~~~~~~~~~~~~~~~~~~~~~~~\land\bigwedge_{j=i+1\ldots N} ( (\nxt^j c) \dimp (\nxt^{j+N+1} c)))\\ 

\end{array}\right )  
\end{array}\right ) 
$$ 
where $\phi_{zero} = \bigwedge_{i=1\ldots N} \nxt^i \neg c$ and $\phi_{max} = \bigwedge_{i=1\ldots N} \nxt^i  c$. 
Intuitively, the first line of the inner formula handles the steps from $2^N-1$ to $0$, and the remainder of the inner formula 
uses the fact that, in binary, $x01^i + 1 = x10^i$.
(Recall that in the run, low order bits are represented to the left.) 

The following formula $\phi^w_{init}$ then says that the run is initialized with word $w=a_1\dots a_n$ 
$$ \blnk \land \nxt^{N+1} ((q_0,a_1) \land \nxt^{N+1}( a_2 \land  \nxt^{N+1} (\ldots  \nxt^{N+1} (a_n \land \nxt \, 
((\alpha_{T,Q} \rimp (\blnk\land \neg \phi_{zero})) 
\until (\alpha_{T,Q}\land \phi_{zero})))\ldots ))$$ 
where $\blnk$ is the blank symbol.   This formula has size $O(N\cdot |w|) = O(p(n)\cdot n)$. 
Intuitively, the formula says that the sequence of symbols $w$ is followed by a sequence of $\bot$ symbols 
until 
the first time that the counter has value zero (this corresponds to the start of the second configuration). 

We now need a formula that expresses that whenever we consider two consecutive configurations $C, C'$
encoded in a run, $C'$ is derived from $C$ by a single step of the TM $T$. The padding 
blanks are easily handled by the following formula $\phi_{pad}$: 
$$ 
\Box ((\alpha_{T,Q} \land ( \phi_{zero} \lor \phi_{max})) \rimp \bot )
$$  
For the remaining cell positions, we need to express that for each cell position $k=1\ldots 2^N-2$, the cell value at 
position $k$ in $C'$ is determined from the cell value at positions $k-1, k,k+1$ in $C$ according
to the transition relation encoding $\Delta$. This means that we need to be able to identify the 
corresponding positions $k$ in $C$ and $C'$. 
To capture the counter value at a given position in the run, 
we represent counter values using a strategy for the single agent, 
as follows. 

We define the observation function $O_1$ for the single agent in 
$\Env_{T,N}$, so that 
observation 
$O_1((t_i,j)) = i$ 
for 
$i=0\ldots N-1$. 
(The values of the 
observation function on other states are not used in the encoding, and can be defined arbitrarily.)  
The number with binary representation $B= b_{N-1} \ldots b_0$ can then be represented 
by the  strategy $\alpha_B$ such that 
$\alpha_B(t_i,j) = a_{b_i}$,
 for $i= 0\ldots N-1$
 and  $j\in \{0,1\}$, and $\alpha_B(s) = a_0$ for all other states $s$. (Note that this strategy is uniform, and conversely, for any 
 uniform strategy $\alpha$ there exists a unique binary number $b_{N-1} \ldots b_0$ such that  
$\alpha_B(t_i,j) = a_{b_i}$,  for $i= 0\ldots N-1$ and  $j\in \{0,1\}$.) 
Comparing this representation with the encoding of numbers along runs, 
the following formula $\phi_{num}(x)$ expresses 
that the number  encoded at the present position in the run is the 
same as the number encoded in the strategy of agent $1$ in the global state denoted by
variable $x$: 
$$
\alpha_{T,Q} \land \bigwedge_{i=0\ldots N-1} (\nxt^{i+1} c) \dimp \neg D_\emptyset \neg ( \lid{\sigma(1)}{x} \land t_i \land \nxt c) 
$$ 
Note that, by the definition of the transition system, 
the value of $\nxt c$ at a state where $t_i$ holds encodes
whether the strategy selects $a_0$ or $a_1$ on observation $i=0\ldots N-1$. 
Note also that since all states of the form $(t_i,j)$ are initial, for every strategy $\alpha$,  the value of $\alpha(t_i,j)$
is represented in this way at some point of some run. 
We may now check that the transitions of the TM are correctly computed along the 
run by means of the following formula $\phi_{trans}$: 
$$
  \Box \left(
  \begin{array}{l} 
  \bigwedge_{(a,b,c,d) \in \Delta} 
  (a \land \neg \phi_{max} \land \nxt^{N+1} (b \land \neg \phi_{max} \land 
  \nxt^{N+1} c))
  \rimp \\
  ~~~~~~~~~~~~~~~~~~~~~~~~~
   \nxt^{N+1} \exists x \, [\phi_{num}(x) \land  \nxt \, ( (\neg \phi_{num}(x)) \until (\phi_{num}(x) \land d )]
   \end{array}
   \right)
$$
Intuitively, here $x$ captures the number encoded at the cell containing the symbol $b$, 
and the $\until$ operator is used to find the next occurrence in the run of this number. 
The occurrences of $\phi_{max}$ ensure that the three positions considered in the formula
do not span across a boundary between two configurations. 
In Figure~\ref{fig:EXPSPACE}, the bottom part represents a strategy 
$\alpha$
of agent 1 encoded in 
some global state $x$. The behaviour of this strategy at the $t$-state runs represents a number, using 
the statement  $\lid{\sigma(1)}{x}$ in the formula $\phi_{num}$. 
The formula $\phi_{num}$ is used to assert that this representation of a binary number in $\alpha$ encodes 
the counter values at a position in a run. Asserting that two positions have the same counter number by this device
allows us to check the yields relation at corresponding positions in the run representation of a computation of the Turing machine.

To express that the machine accepts we just need to assert that the accepting state is reached; this 
is done by the formula 
$\phi_{accept} =  \Diamond  \bigvee_{a\in A_T}(a,q_f)$. 

Combining these pieces, we get that the TM accepts input $w$ 
if and only if 
$$\Env_{T,N} \models (\phi^N_{clock} \land \phi^N_{count} \land  \phi^w_{init} \land  \phi_{trans}) \rimp  \phi_{accept} $$
holds, i.e., when every run that correctly encodes a computation of the machine is accepting. 

\end{proof}

Combining Theorem~\ref{mcESLupper} and Theorem~\ref{mcESLlower} we
obtain the following characterization of the complexity of ESL model checking.

\begin{cor} 
Let 
$\Strats$  be an EXPSPACE presented class of strategies for environments, containing $\Strats^{\unif,\det}$. 
The complexity of deciding, given an environment $\Env$, an \ESL\ formula $\phi$ 
and a context $\Cont$ for the free variables in an $\ESL$  formula $\phi$ relative
to $\Env$ and $\Strats(\Env)$,  whether $\Cont, \Env,\Strat(\Env)\models \phi$, is EXPSPACE-complete. 
\end{cor}

The high  complexity for \ESL\  model checking motivates the consideration of 
fragments that have lower model checking complexity. We demonstrate
two orthogonal fragments for which the complexity of model checking is in a lower complexity class.  
One is the fragment $\ESL^-$, where we allow the operators $\existsg{x} \phi$ and $\lid{i}{x}$, 
but restrict the use of the temporal operators to be those of the branching-time temporal logic CTL. 
In this case, we have the following result: 

\begin{theorem} \label{mcESLminus} 
Let $\Strat$ be a PSPACE-presented class of strategies. 
The problem of deciding, given an environment $\Env$, a formula $\phi$ of $\ESL^-$, and a context $\Cont$ for the free variables of $\phi$ relative to $\Env$ and $\Strat(\Env)$, whether
$\Cont, \Env, \Strat(\Env) \models \phi$, is  in  PSPACE. 
\end{theorem}

\begin{proof} 
We observe that the following fact follows straightforwardly from the semantics for formulas $\phi$ of $\ESL^-$: 
for a context $\Cont$ for the free variables of $\phi$ relative to $\Env$ and $\Strat(\Env)$, 
and for two points $(r,n)$ and $(r',n')$ of $\I(\Env, \Strat(\Env))$
with $r(n)= r'(n')$, we have that 
$\Cont, \I(\Env, \Strat(\Env)), (r,n) \models \phi$ iff 
$\Cont, \I(\Env, \Strat(\Env)), (r',n') \models \phi$. 
That is, satisfaction of a formula relative to a context at a point depends only on the 
global state at the point, and not on other details of the run containing the point. 
For a global  state $(s,\alpha)$ of $ \I(\Env, \Strat(\Env))$, 
define the boolean $SAT(\Cont, \Env, \Strat, (s,\alpha), \phi)$ 
to be TRUE just when $\Cont, \I(\Env, \Strat(\Env)), (r,n) \models \phi$ holds for 
some point $(r,n)$ of $\I(\Env, \Strat(\Env))$ with $r(n) = (s,\alpha)$. 
By the above observation, we have that 
$\Cont, \Env, \Strat(\Env) \models \phi$ iff
 $SAT(\Cont, \Env, \Strat, (s,\alpha), \phi)$ holds for all initial states $s$ of $\Env$ 
 and all strategies $\alpha\in \Strat(\Env)$. 
 Since we may check these conditions one at a time, strategies $\alpha$ can be represented in space 
 linear in 
 $|\Env|$,  and deciding $\alpha\in \Strat(\Env)$ is in PSPACE, 
 it suffices to show that $SAT(\Cont, \Env, \Strat, (s,\alpha), \phi)$ is 
 decidable in PSPACE. 
 
 We proceed by describing an APTIME algorithm for $SAT(\Cont, \Env, \Strat, (s,\alpha), \phi)$, 
 and using the fact that APTIME = PSPACE \cite{alternation}. The algorithm operates
 recursively, with the following cases: 
 \be 
 \item If $\phi = p$, for $p \in \Prop$, then return TRUE if $p \in \pi(s)$, else return FALSE. 
  \item If $\phi =\lid{i}{x}$,  then return TRUE if $(s,\alpha)_i = \Cont(x)_i$, else return FALSE. 
  \item If $\phi =\phi_1 \land \phi_2$,  then  universally call $SAT(\Cont, \Env, \Strat, (s,\alpha), \phi_1)$
  and \\
  $SAT(\Cont, \Env, \Strat, (s,\alpha), \phi_2)$.
  \item If $\phi =\neg \phi_1$,  then  return the complement of $SAT(\Cont, \Env, \Strat, (s,\alpha), \phi_1)$. 
  \item If $\phi = A\,\nxt \phi_1$ then universally choose a state $t$ such that $s \ptrans{a}t$ for some 
  for some joint action $a$, 
  and call $SAT(\Cont, \Env, \Strat, (t,\alpha), \phi_1)$. 
  The other temporal operators from CTL are handled similarly. (In the case of operators using $\until$, we need
  to run a search for a path through the set of states of $\Env$ generated by the strategy $\alpha$, but this is easily handled in APTIME.) 
  \item If $\phi =D_G \phi_1$,  then  
  universally choose a global state 
  $(t,\beta)$ such that $(s,\alpha) \sim^D_G (t,\beta)$ 
  and universally 
  \be 
  \item decide 
  whether $\beta \in \Strat(\Env)$, and 
   \item call $\mathit{REACH}(t,\beta)$, and   
  \item  call  $SAT(\Cont, \Env, \Strat, (t,\beta), \phi_1)$. 
  \ee 
(Here $\mathit{REACH}(t,\beta)$ decides whether 
state $t$ is reachable in $\Env$ from some initial state when the agents
run the joint strategy $\beta$;  this is trivially in PSPACE.
Deciding $\beta \in \Strat(\Env)$ is in PSPACE by the assumption that $\Strat$ is PSPACE-presented.)  
\item 
 If $\phi =C_G \phi_1$,  then    universally guess a global state 
  $(t,\beta)$ and universally do the following: 
  \be 
  \item Decide 
   whether 
   $(s,\alpha) \sim^C_G (t,\beta)$  using an existentially branching binary search for a  path of length at most 
  $|S|\times |\Strat(\Env)|$.  
  For all states $(u,\gamma)$ on this path it should 
  be 
  verified that 
  $\mathit{REACH}(u, \gamma)$  and that $\gamma \in \Strat(\Env)$. 
  The maximal length of the path is in the worst case exponential in $|\Env|$, but the binary search 
  can handle this in APTIME. 
     \item  call  $SAT(\Cont, \Env, \Strat, (t,\beta), \phi_1)$. 
  \ee 
\item If $\phi = \exists x(\phi_1)$, then 
existentially guess a global state $(t,\beta)$, and universally
 \be 
  \item decide if $\beta \in \Strat(\Env)$, and 
   \item call $\mathit{REACH}(t,\beta)$, and   
  \item  call  $SAT(\Cont[(t,\beta)/x], \Env, \Strat, (s,\alpha), \phi_1)$. 
  \ee 
 \ee 
A straightforward argument based on the semantics of the logic shows that the above correctly 
computes SAT. 

We remark that a more efficient procedure  for checking that $(s,\alpha) \sim^C_G (t,\beta)$  is possible in the 
typical case where 
$\Strats(E)$ 
is a cartesian product of sets of strategies for each of the agents. 
In this case, if there exists  a witness chain then there is one of length at most $|S|$. 
Let $G=G_1\cup \strat(G_2)$ such that $G_1,G_2\subseteq \Ags$. The number of steps 
through the relation $\cup_{i\in G} \sim_i$ 
required to witness $(s,\alpha) \sim^C_G (t,\beta)$
depends on the sets $G_1,G_2$ as follows: 
\be
\item If $G_1= G_2 = \emptyset$ then we must have $(s,\alpha)=(t,\beta)$ and a chain of length 0 suffices. 
\item If $G_1$ is nonempty and $G_2 = \emptyset$ then we must have $s ~(\cup_{i\in G_1} \sim_i)^* ~t$, 
but $\beta$ can be arbitrary, and this component can be changed in any step. A path of length 
$|S|$ suffices in this case. 
\item If $G_1 = \emptyset$ and $G_2 = \{i\}$ is a singleton, then 
we must have $\alpha_i = \beta_i$, but $s$ and $t$ can be arbitrary. 
A path of length one suffices in this case. 

\item If $|G_1|\geq 1$, say $i \in G_1$, and $G_2 = \{j\}$ is a singleton, then 
$(\cup_{i\in G} \sim_i)^*$ is the universal relation and a path of length 2 suffices. 
In particular, for any $(s,\alpha),(t,\beta)$ we have 
$(s, \alpha) \sim_i (s,\beta) \sim_{\strat(j)} (t, \beta)$. 

\item If $|G_2 |\geq 2$ then $(\cup_{i\in G} \sim_i)^*$ is the universal relation and a path of length 2 suffices. 
In particular, for any $(s,\alpha),(t,\beta)$ and 
distinct 
$i,j\in G_2$, 
 there exists  $\alpha'$ such that $\alpha'_i = \alpha_i$
 and
 $\alpha'_k = \beta_k$ for all 
 $k \in \Ags$ with 
 $k\neq i$, and $(s, \alpha) \sim_{\strat(i)} (s,\alpha') \sim_{\strat(j)} (t, \beta)$. 

\ee
\end{proof} 

The following result shows that the PSPACE upper bound from this result is tight, already
for formulas that use strategy 
indices
in the CTLK operators, but make no direct uses of 
the constructs $\exists x$ and $\lid{i}{x}$. 

\begin{theorem} \label{mcCTLKlower}
The problem of deciding, given an environment $\Env$ for two agents and a
formula $\phi$ of $\CTLK(
\Ags
\cup \strat(\Ags), \Prop)$, 
whether
$\Env, \Strat^{\unif, \detstrat}(\Env)\models \phi$ is PSPACE hard.   
\end{theorem}

\begin{proof} 
We proceed by a reduction 
from the satisfiability of Quantified Boolean Formulae (QBF). 
An instance of QBF is a formula  $\phi$  of form 
$$ Q_1x_1...Q_nx_n(\gamma)$$
where $Q_1,...,Q_n\in \{\exists,\forall\}$ and $\gamma$ is a formula of propositional logic over 
propositions $x_1,...,x_n$. 
The QBF problem is to decide, given a QBF instance  $\phi$,  whether it is true. 
We construct an environment $\Env_\phi$ and a formula $\phi^*$ of 
$\CTLK$ using strategic 
indices
$\strat(i)$ 
such that the QBF formula $\phi$ is true iff we have $\Env_\phi, \Strat^{\unif, \detstrat}(\Env_\phi) \models \phi^*$. 

Given the QBF  formula $\phi$, we construct  the environment $\Env_\phi =  \langle S, I, \{\Acts_i\}_{i\in \Ags}, \trans, \{O_i\}_{i\in \Ags}, \pi\rangle$
 for $2$ agents $\Ags = \{1, 2\}$ and propositions $\Prop = \{ p_0,\ldots, p_n, q_1,
 q_2\}$  as follows.
\be 
\item The set of states $S = \{ s_0 \} \cup \{ s_{t,j,k} ~|~ t \in \{1\ldots n\}, j,k\in \{0,1\}\}$. 
\item The set of initial states is $I = \{ s_0\}$. 
\item The actions of agent $i$ are $A_i = \{0,1\}$, for each $i \in \Ags$. 
\item The transition relation is defined
to consist of the following transitions, where $j,j',k,k' \in \{0,1\}$
$$
\begin{array}{cl} 
s_0 \ptrans{(j',k')} s_{1,j',k'} \\ 
s_{t,j,k}  \ptrans{(j',k')} s_{t+1,j',k'} &  \text{~~~for $t= 1 \ldots n-1$} \\
s_{n,j,k}  \ptrans{(j',k')} s_{n,j,k}~.
\end{array} 
$$ 
\item Observations are defined so that  $O_i(s_0) =0$ and  $O_i(s_{t,j,k}) = t$.  
\item 
The assignment $\pi$ is defined by $\pi(s_0) = \{p_0\}$, 
and 
$$\pi(s_{t,j,k})= \{p_t \}\cup \{q_1 ~|~j=1\}\cup  \{q_2~|~k=1\}$$
for $t = 1\ldots n$.

\ee 
Intuitively, the model sets up $n+1$ moments of time 
$t= 0,\ldots, n$, 
with $s_0$ the only possible state at time $0$ and $s_{t,j,k}$ for $j,k\in \{0,1\}$ the possible
states at times  $t= 1,\ldots, n$. 
Both agents observe only the value of the moment of time, so that for each agent, a 
strategy selects an action $0$ or $1$ at each moment of time. 
We may therefore encode an assignment to the proposition variables $x_1 \ldots x_n$ 
by the actions chosen by an 
agent's strategy 
at times $0, \ldots n-1$. The action chosen by each agent at  time $t\in \{0\ldots n-1\}$ is 
recorded in the indices of the state at time $t+1$, i.e. if the state at time $t+1$ 
is $s_{t+1,j,k}$ then agent 1 chose action $j$ at time $t$, 
and agent 2 chose action $k$. 

We work with two agents, each of whose strategies is able to 
encode an assignment, in order to alternate between
the two encodings. At each step, one of the strategies is 
assumed to encode an assignment to the variables $x_1, \ldots x_{m}$. 
This strategy is fixed, and we universally or existentially guess the 
other strategy in order to obtain a new value for the variable $x_{m+1}$. 
We then check that the guess has maintained the values of the existing
assignment to $x_1, \ldots x_{m}$ by comparing the two strategies. 

More precisely, let $val_i(x_j)$ be the formula 
$K_{\{\strat(i)\}}(p_{j-1} \rimp 
E\nxt
(q_i))$ for $i=1,2$ and $j = 1\ldots n$. 
This states that at the current state, the strategy of agent
$i$ selects action $1$ at time $j-1$, so it encodes an assignment making $x_j$ true. 
For $m = 1\ldots n$, let $agree(m)$ be the formula 
$$ \bigwedge_{j=1\ldots m} D_{\{\strat(1), \strat(2)\}}(p_{j-1} \rimp (
E\nxt
(q_1) \dimp 
E\nxt 
(q_2))) $$
This says that the assignments encoded by the strategies of the two agents agree on the values
of the variables $x_1 \ldots, x_m$.  
Assuming, without loss of generality, that $n$ is even, 
and that the quantifier sequence in $\phi$ is $(\exists\, \forall)^{n/2}$, 
given the QBF formula $\phi$, 
define the formula $\phi^*$ to be 
\begin{tabbing} 
$\neg D_\emptyset\neg ( D_{\{\strat(1)\}}($\=$agree(1) \rimp$ \\ 
\> $\neg D_{\{\strat(2)\}}\neg ($\=$agree(2) \land$  \\ 
\>\>$D_{\{\strat(1)\}}($\=$agree(3) \rimp$ \\ 
\>\> \> $\neg D_{\{\strat(2)\}}\neg ( agree(4)\land $\= $ \ldots  $ \\ 
\>\>\> \> $\vdots$ \\
\>\>\>\> $D_{\{\strat(1)\}}( agree(m-1) \rimp \gamma^{+})\ldots )$
\end{tabbing}
where $\gamma^+$ is the formula obtained by replacing each 
occurrence of a variable $x_j$ in $\gamma$ by the formula $val_2(x_j)$. 
Intuitively, the first operator $\neg D_\emptyset\neg  $
existentially chooses a value for variable $x_1$, encoded in $\strat(1)$, the
next operator 
$D_{\{\strat(1)\}}$ 
remembers this strategy while
encoding a universal choice of value for variable $x_2$ in $\strat(2)$, and
the formula $agree(1)$ checks that the existing choice for $x_1$ 
is preserved in $\strat(2)$. Continued alternation between the
two strategies adds universal or existential choices 
for variable values while preserving previous choices. 
It can then be shown that the QBF formula $\phi$ is true 
iff $\Env_\phi, \Strat^{\unif, \detstrat}\models \phi^*$. 
\end{proof} 

Combining Theorem~\ref{mcESLminus} and Theorem~\ref{mcCTLKlower}, we obtain the following: 

\begin{cor} 
Let $\Strat$ be a PSPACE-presented class of strategies. 
The problem of deciding if $\Cont, \Env, \Strat(\Env) \models \phi$,  
given an environment $\Env$, a formula $\phi$ of $\ESL^-$
and a context $\Cont$ for the free variables of $\phi$ relative to $\Env$ and $\Strat(\Env)$, 
is PSPACE complete. 
\end{cor}

Since PSPACE is strictly contained in EXPSPACE, this result shows a strict improvement in 
complexity as  a result of the restriction to the CTL-based fragment. 
We remark that, by a trivial generalization of the standard state labelling algorithm for model 
checking \CTL\ to handle the knowledge operators, 
the problem of model checking the logic $\CTLK(\Ags, \Prop)$ 
in the system 
$\I(\ets)$ 
generated by an epistemic 
transition system 
$\ets$ is in PTIME. Thus, there is a jump in complexity from \CTLK\  
as a result of the move to the strategic setting, even without the addition of the operators $\existsg{x} \phi$ and $\lid{i}{x}$.  
However, this jump is not so large as the jump to the  the full logic $\ESL$. 

An orthogonal restriction of $\ESL$ is to retain the $\CTLs$ temporal basis, i.e., to allow full use of 
LTL operators, but to 
allow
epistemic operators and 
strategy 
indices, 
but omit use of the operators $\existsg{x} \phi$ and $\lid{i}{x}$. This gives the logic 
$\CTLsK(
\Ags 
\cup \strat(\Ags), \Prop)$. For this logic we also see an improvement in the complexity of model 
checking compared to full $\ESL$, as is shown in the following result. 

\begin{theorem} \label{mcCTLsK}
Let $\Strats(E)$ be a PSPACE presented class of strategies for environments $E$.  
The complexity of deciding, given an environment $E$
 and a \CTLsK\ formula $\phi$ for agents 
 $\Ags(E)^+ \cup \strat(\Ags(E))$, 
whether $\Env,\Strats(\Env) \models\phi$, is PSPACE-complete. 
\end{theorem}
\begin{proof} The lower bound is straightforward from the fact that linear time temporal logic LTL is a 
sublanguage of \CTLsK, and model checking LTL is already PSPACE-hard \cite{SistlaC85}. 
For the upper bound, we describe an alternating PTIME algorithm, and invoke the fact that APTIME = PSPACE \cite{alternation}. 
We abbreviate $\I(E, \Strats(E))$ to $\I$. 

For a formula $\phi$, write $\ksubf(\phi)$ for the maximal epistemic subformulas of $\phi$, defined to be the 
set of subformulas of the form $A\psi$ or 
$C_G\psi$ or $D_G\psi$ for some set $G$ of basic and strategic 
indices, 
which are themselves not a subformula of a larger subformula of $\phi$ of one of these forms. 
Note that $A\psi$ can be taken to be epistemic because it is equivalent to $D_{\{e\}\cup \strat(\Ags)}\psi$; 
in the following we assume that $A\psi$ is written in this form. 
Also note that for epistemic formulas $\psi$, satisfaction at a point depends 
only on the global state, i.e.,  for all  points $(r,m)$ and $(r',m')$ of $\I$, we have
that if $r(m) = r'(m')$ then $\I, (r, m) \models  \psi$ iff $\I, (r', m') \models  \psi$. 
Thus, for global states $(s,\alpha)$  of 
$\I$,
we may write $\I, (s,\alpha) \models \psi$
to mean that  $\I, (r, m) \models  \psi$ for some point $(r,m)$ with $r(m) = (s, \alpha)$. 

Define a  {\em $\phi$-labelling} of $E$ to be a mapping $L: S\times \ksubf(\phi) \rightarrow \{0,1\}$, 
giving a truth value for each maximal epistemic subformula of $\phi$. 
A $\phi$-labelling can be represented in space $|S|\times |\phi|$. 
Note that if we treat the maximal epistemic subformulas of $\phi$ as if they were atomic 
propositions, evaluated at the states of $E$ using the $\phi$-labelling $L$, 
then $\phi$ becomes an LTL formula, evaluable on any path in $E$ 
with respect to the labelling $L$. 
Verifying that all $\alpha$-paths from a state $s$ satisfy $\phi$ with respect to $L$ 
is then exactly the problem of LTL model checking, for which there exists an 
APTIME procedure $\ASAT(E,(s, \alpha), L, \phi)$
since model checking LTL is in PSPACE \cite{SistlaC85} and APTIME = PSPACE \cite{alternation}. 
For this to correspond to model checking in $\I$, 
we require that 
the $\phi$-labelling 
$L$ gives the correct answers for the truth value of the formula
at each state $(s,\alpha)$, i.e., 
that $L(\psi) = 1$ iff $\I, (s,\alpha) \models \psi$. 
 We handle this by means of a guess and verify technique. 

To handle the verification, an alternating PTIME algorithm  $\KSAT(E,\Strats, (s,\alpha), \phi)$  is defined, for 
$\phi$ an epistemic formula, such that $\KSAT(E,\Strats, (s,\alpha), \phi)$ returns TRUE 
iff  $\I, (s,\alpha) \models  \phi$.  The definition is recursive and uses a call to the  procedure ASAT. 
Specifically, 
$\KSAT(E,\Strats, (s,\alpha), D_G\phi)$ operates as follows: 
\be 
\item universally guess a state $t$ of $E$ and a joint strategy $\beta$ in $E$, then

\item  verify that $t$ is reachable in $E$ using joint strategy $\beta$, 
that $(s,\alpha) \sim_G (t,\beta)$, and that $\beta$ is in $\Strats(E)$, then 

\item existentially guess a $\phi$-labelling $L$ of $E$,  then 

\item universally, 
\be \item 
call $\ASAT(E,(t,\beta),L,\phi)$, and 
\item for each state $w$ and formula $\psi \in \ksubf(\phi)$, 
call $\KSAT(E,\Strats, (w,\beta), \psi)$.  
\ee
\ee
Note that  step 4(b) verifies that the 
$\phi$-labelling $L$ is correct. 

For $\KSAT(E,\Strats, (s,\alpha),  C_G\phi)$, 
the procedure is similar, except that instead of verifying that $(s,\alpha) \sim_G (t,\beta)$ in the 
second step, we need to verify that $(s,\alpha) ~(\cup_{i\in G} \sim_i)^* ~(t,\beta)$. 
This is  easily handled in APTIME by a standard recursive procedure that guesses a midpoint of the 
path 
and branches 
universally to verify the existence of the left and right halves of the chain. 
(See the proof of Theorem~\ref{mcESLminus} for some further discussion on this point.) 

To solve the model checking problem in $\I$, 
we can now apply the following alternating procedure: 
\be 
\item universally guess a global state $(s,\alpha)$ of $\I$, then 
branch existentially to the following cases: 
\be 
\item if $s$ is an initial state of $\Env$ return FALSE, else return TRUE, 
\item if $\alpha\in \Strats(\Env)$, return FALSE, else return TRUE, 
\item call $\KSAT(E, \Strat, (s,\alpha), A\phi)$. 
\ee 
\ee
Evidently, each of the alternating procedures runs in polynomial time internally, 
and the number of recursive calls is $O(|\phi|)$. It follows that the entire
computation is in APTIME = PSPACE. 
\end{proof}

It is interesting to note that, although $\CTLsK(
\Ags 
\cup \strat(\Ags))$ is  significantly richer than the temporal logic LTL, 
the added expressiveness comes without an increase in complexity: model checking LTL is already PSPACE-complete \cite{SistlaC85}.

\section{Conclusion}\label{sec:concl}

We now discuss some related work and remark upon some questions for future research. 
The sections above have already made some references and 
comparisons to related work on each of the  topics that we cover. 
Beside these references, the following are also worth mentioning.

Semantics that explicitly encode strategies in runs  have been used 
previously in the literature 
on knowledge in information flow security
\cite{HalpernOneill}; what is novel in our approach is
to develop a logic that enables explicit reference  to these strategies.

A variant of propositional dynamic logic (PDL) for describing strategy profiles 
in normal form games subject to preference relations is introduced in 
\cite{vanEijck13}. This work does not cover temporal aspects as we have done 
in this paper. Another approach based on PDL is given in 
\cite{RamSimon08}, which describes strategies by means of formulas. 

A very rich generalization of ATEL for probabilistic environments is described in \cite{Schnoor10}. 
This proposal includes variables that refer to {\em strategy choices}, and  strategic operators that may 
refer to these variables, so that statements of the form ``when coalition A runs the strategy represented
by variable S1, and coalition B runs the strategy represented by variable S2, 
and the remaining agents behave arbitrarily, then the probability that $\phi$ holds
is at least $\delta$'' can be expressed. Here a strategy choice maps each 
state, coalition and formula to a uniform imperfect recall strategy for the coalition. 
There are a number of syntactic restrictions compared to our logic. 
The epistemic operators in this approach  apply only to 
state formulas rather than path formulas (in the sense of this distinction from $\CTL^*$.) 
Moreover, the strategic variables may be quantified, but only in the prefix of the formula. 
These constraints imply  that notions such as ``agent $i$ knows that there exists a strategy by which it can achieve
$\phi$'' and  ``agent $i$ knows that it has a winning response to every strategy chosen by agent $j$'' 
cannot be naturally expressed. 

The extended temporal epistemic logic $\ETL$ we have introduced, of which our epistemic 
strategy logic $\ESL$ is an instantiation with respect to a particular semantics, uses constructs that 
resemble constructs  from {\em hybrid logic} \cite{BS98}.
Hybrid logic is an approach to the extension of modal logics that 
uses ``nominals'', i.e., propositions $p$ that hold at a 
single world. These can be  used in combination with operators
such as  $\exists p$, which marks an arbitrary world
as the unique world at which nominal $p$ holds. Our construct $\exists x$ is 
closely related to the hybrid construct  $\exists p$, but we work in a setting
that is richer in both syntax and semantics than previous works. 
There have been a few works using hybrid logic ideas in the context of epistemic logic \cite{Hansen11,Roy09a} 
but none are  concerned with temporal logic. Hybrid temporal logic has seen a larger
amount of study \cite{BozzelliL10,FranceschetRS03,FranceschetR06,SchwentickW07}, 
with variances in the semantics used for the model checking problem. 

We note that if we were to view the variable $x$ in our logic as a propositional constant, 
it would be true at a set of points in the system $\I(\Env, \Strats)$, hence not a nominal in that system.
Results in \cite{BozzelliL10}, where a hybrid linear time temporal logic
formula is checked in all paths in a given model, suggest that 
a variant of $\ESL$ in which $x$ is treated  as a nominal in $\I(\Env, \Strats)$
would have a complexity of model checking at least non-elementary, compared to 
our EXPSPACE and PSPACE complexity results. 

Our PSPACE model checking result for $\CTLK(
\Ags
\cup \strat(\Ags))$ seems to be more closely related to the 
result in \cite{FranceschetR06} that 
model checking a logic HL$(\exists, @, F,A)$ is PSPACE-complete. 
Here $F$ is essentially a branching  time future operator and $A$ is a universal operator (similar to 
our $D_{\emptyset}$), the construct $@_p\phi$ says that $\phi$ holds at the world marked by the 
nominal $p$, and $\exists p(\phi)$ says that $\phi$ holds after marking some world by $p$. 
The semantics in this case does not unfold the model 
into either a tree or a set of linear structures before checking the formula, 
so the semantics of the hybrid existential $\exists$ is close to our idea
of quantifying over global states. Our language, however, has a richer set of 
operators, even in the temporal dimension, and introduces the strategic dimension in the 
semantics. It would be an interesting question for future work to 
consider fragments of our language to obtain 
a more precise statement of the
relationship with hybrid temporal logics. 

Strategy Logic \cite{CHP10} is a (non-epistemic) generalization of  ATL for perfect information 
strategies in which strategies may be explicitly named and quantified. 
Strategy logic  has a non-elementary model checking problem. Work on identification of more efficient variants 
of quantified strategy logic includes \cite{MogaveroMV10}, who formulate a variant 
with a 2-EXPTIME-complete model checking problem. In both cases, strategies are perfect recall strategies, rather than the imperfect recall 
strategies that form the basis for our PSPACE-completeness result for model checking. 

Most closely related to this paper are a number of independently developed 
works that consider epistemic extensions of variants 
of strategy logic.  Belardinelli \cite{Belardinelli14} develops a logic, 
based on linear time temporal logic with epistemic operators, that adds an operator 
$\exists x_i$, the semantics of which existentially modifies the 
strategy associated to agent $i$ in the current strategy profile. 
It omits the binding operator from \cite{MogaveroMV10}, so 
provides no other way to refer to the variable $x$. 
The logic is shown to have nonelementary model checking 
complexity. This complexity is higher than the results we have presented
because the semantics for strategies allows agents to have perfect 
information and perfect recall (though the semantics for the 
knowledge operators is based on imperfect information 
and no recall), whereas we have assumed imperfect information 
and no recall for strategies.  

Another extension of strategy logic with epistemic operators has been 
independently developed by \v{C}erm\'{a}k et al \cite{CLMM2014,Cermak}. 
Their syntax and semantics differs from ours in a number of respects. 
Although the syntax appears superficially in the
form of an extension of LTL, it is more like CTL in some regards. 
The transition relation is deterministic in the 
sense that for each joint action, each state has a unique successor
when that action is performed. 
Strategies are also assumed to be deterministic (whereas we allow nondeterministic strategies.)  
This means that, like CTL, the semantics of a formula
depends only on the current global state and the current strategy profile, 
whereas for LTL it is generally the case that the future structure of the run from a given global state can vary, and 
the truth value of the formula depends on how it does so. 
Although it seems that non-determinism could be modelled, as is commonly done, through the choice of actions 
of the environment, treated as an agent, the fact that strategies are deterministic, uniform and 
memoryless means that the environment must choose the 
same alternative each time a global state occurs in a run. This means
that this standard approach to modelling of non-determinism does not work for this 
logic. The syntax of the logic moreover prevents epistemic operators
from being applied to formulas with free strategy variables, whereas we
allow fully recursive mixing of the constructs of our logic. 
Consequently, epistemic notions from our logic like $D_{\{i, \strat(i)\}}$,  expressing an agent's knowledge 
about the effects of its own strategy, which are used in 
several of our applications, do not appear to be expressible in this logic.  
Finally, the notion of ``interpreted system'' in this work, which corresponds
most closely to our notion of ``environment'', also seems less general than our notion of 
environment because it  defines the accessibility relations for the knowledge operators in a way that makes the 
environment state known to all agents. 

In another paper ~\cite{AAAI2014}, we have implemented a symbolic algorithm that handles model checking 
for the fragment $\CTLK(
\Ags 
\cup \strat(\Ags))$, which, as shown above, encompasses the expressiveness of ATEL. 
Existing algorithms described in the literature for ATEL model checking~\cite{MCMASATL2006,CSS2010,BPQR2013} 
are based either on  explicit-state model checking or are only partially symbolic in that they
iterate over all strategies, explicitly represented.  Our 
experimental results in~\cite{AAAI2014} show that by comparison with the partially-symbolic 
approach, a fully-symbolic algorithm can greatly improve the performance and therefore scalability of model checking. 
The approach to model checking epistemic strategy logic implemented in \cite{CLMM2014,Cermak} 
is fully symbolic, but as already mentioned, this logic has a more limited expressive
power than ours and its semantics does not permit representation of a nondeterministic environment. 
(It does not seem that the semantics could be extended to allow nondeterminism while retaining correctness
of their algorithm.) 

Our focus on this paper has been on an observational, or imperfect recall, 
semantics for knowledge. Other semantics for knowledge are also 
worth considering, 
but are left for future work. 
We note one issue in relation to the connection to ATEL that we
have established, should we consider a perfect recall version of our logic. 
ATEL operators effectively allow reference to situations in which agents switch their strategy
after some actions have already been taken, whereas in our model an agent's strategy is fixed for the 
entire run. When switching to a new strategy, there is the possibility that the given state is not 
reachable under this new strategy. 
We have handled this issue in our translation by assuming that all states are initial, so that 
the run can be reinitialized if necessary to make the desired state reachable. 
This is consistent with an imperfect recall interpretation of ATEL, but it is not
clear that this approach is available on a perfect recall interpretation. 
We leave a resolution of this issue to future work.  

\bibliographystyle{plain}
\bibliography{atl}

\begin{thebibliography}{10}

\bibitem{AgotnesGJ07}
Thomas $\AA$gotnes, Valentin Goranko, and Wojciech Jamroga.
\newblock Alternating-time temporal logics with irrevocable strategies.
\newblock In {\em Proc. of the 11th Conf. on Theoretical Aspects of Rationality
  and Knowledge (TARK-2007)}, pages 15--24, 2007.

\bibitem{ATLJACM}
Rajeev Alur, Thomas~A. Henzinger, and Orna Kupferman.
\newblock {Alternating-Time Temporal Logic}.
\newblock {\em Journal of the ACM}, 49(5):672--713, 2002.

\bibitem{Belardinelli14}
Francesco Belardinelli.
\newblock Reasoning about knowledge and strategies: Epistemic strategy logic.
\newblock In {\em Proc. 2nd Int. Workshop on Strategic Reasoning, {SR} 2014,
  Grenoble, France, April 5-6, 2014.}, pages 27--33, 2014.
\newblock arXiv:1404.0837v1.

\bibitem{BS98}
Patrick Blackburn and Jerry Seligman.
\newblock What are hybrid languages?
\newblock In M.~de~Rijke, H.~Wansing, and M.~Zakharyaschev, editors, {\em
  Advances in Modal Logic}, volume~1, pages 41--62. CSLI Publications, 1998.

\bibitem{BozzelliL10}
Laura Bozzelli and Ruggero Lanotte.
\newblock Complexity and succinctness issues for linear-time hybrid logics.
\newblock {\em Theoretical Computer Science}, 411(2):454--469, 2010.

\bibitem{BrafmanLMS1997}
Ronen~I. Brafman, Jean-Claude Latombe, Yoram Moses, and Yoav Shoham.
\newblock {Applications of a Logic of Knowledge to Motion Planning under
  Uncertainty}.
\newblock {\em JACM}, 44(5):633--668, 1997.

\bibitem{BCLM2009}
Thomas Brihaye, Arnaud~Da Costa, Fran{\c c}ois Laroussinie, and Nicolas Markey.
\newblock {ATL} with strategy contexts and bounded memory.
\newblock In {\em International Symposium on Logical Foundations of Computer
  Science}, pages 92--106, 2009.

\bibitem{BPQR2013}
Simon Busard, Charles Pecheur, Hongyang Qu, and Franco Raimondi.
\newblock Reasoning about strategies under partial observability and fairness
  constraints.
\newblock In {\em 1st International Workshop on Strategic Reasoning (SR2013)},
  pages 71--79, 2013.

\bibitem{CSS2010}
Jan Calta, Dmitry Shkatov, and Bernd-Holger Schlingloff.
\newblock Finding uniform strategies for multi-agent systems.
\newblock In {\em Computational Logic in Multi-Agent Systems (CLIMA XI)}, pages
  135--152, 2010.

\bibitem{Cermak}
Petr {\v C}erm\'{a}k.
\newblock A model checker for strategy logic.
\newblock MEng Individual Project thesis, Department of Computing, Imperial
  College London, June 2014.

\bibitem{CLMM2014}
Petr {\v C}erm{\'a}k, Alessio Lomuscio, Fabio Mogavero, and Aniello Murano.
\newblock Mcmas-slk: A model checker for the verification of strategy logic
  specifications.
\newblock In {\em 26th International Conference, CAV 2014}, pages 525--532,
  2014.

\bibitem{alternation}
Ashok~K. Chandra, Dexter Kozen, and Larry~J. Stockmeyer.
\newblock Alternation.
\newblock {\em Journal of the ACM}, 28(1):114--133, 1981.

\bibitem{CHP10}
Krishnendu Chatterjee, Thomas~A. Henzinger, and Nir Piterman.
\newblock Strategy logic.
\newblock {\em Information and Computation}, 208(6):677--693, 2010.

\bibitem{ChongM05}
Stephen Chong and Andrew~C. Myers.
\newblock Language-based information erasure.
\newblock In {\em IEEE Computer Security Foundations Workshop}, pages 241--254,
  2005.

\bibitem{CES1986}
Edmund~M. Clarke, E.~Allen Emerson, and A.~Prasad Sistla.
\newblock Automatic verification of finite-state concurrent systems using
  temporal logic specifications.
\newblock {\em ACM Transactions on Programming Languages and Systems},
  8(2):244--263, 1986.

\bibitem{DimaATLpr}
Catalin Dima and Ferucio~Laurentiu Tiplea.
\newblock Model-checking {ATL} under imperfect information and perfect recall
  semantics is undecidable.
\newblock {\em CoRR}, abs/1102.4225, 2011.

\bibitem{EH1986}
E.~Allen Emerson and Joseph~Y. Halpern.
\newblock ``sometimes'' and ``not never'' revisited: on branching versus linear
  time temporal logic.
\newblock {\em Journal of the ACM}, 33(1):151--178, 1986.

\bibitem{EGM2007:LFCS}
Kai Engelhardt, Peter Gammie, and Ron van~der Meyden.
\newblock Model checking knowledge and linear time: {PSPACE} cases.
\newblock In {\em Proc.\@ Symposium on Logical Foundations of Computer
  Science}, LNCS, pages 195--211. Springer, June 2007.

\bibitem{FHMVbook}
Ronald Fagin, Joseph~Y. Halpern, Yoram Moses, and Moshe~Y. Vardi.
\newblock {\em {Reasoning About Knowledge}}.
\newblock The MIT Press, 1995.

\bibitem{FHMV1997}
Ronald Fagin, Joseph~Y. Halpern, Yoram Moses, and Moshe~Y. Vardi.
\newblock Knowledge-based programs.
\newblock {\em Distributed Computing}, 10(4):199--225, 1997.

\bibitem{FranceschetR06}
Massimo Franceschet and Maarten de~Rijke.
\newblock Model checking hybrid logics (with an application to semistructured
  data).
\newblock {\em Journal of Applied Logic}, 4(3):279--304, 2006.

\bibitem{FranceschetRS03}
Massimo Franceschet, Maarten de~Rijke, and Bernd-Holger Schlingloff.
\newblock Hybrid logics on linear structures: Expressivity and complexity.
\newblock In {\em 10th International Symposium on Temporal Representation and
  Reasoning / 4th International Conference on Temporal Logic (TIME-ICTL 2003)},
  pages 166--173, 2003.

\bibitem{HM90}
Joseph~Y. Halpern and Yoram Moses.
\newblock {Knowledge and Common Knowledge in a Distributed Environment}.
\newblock {\em Journal of the ACM}, 37(3):549--587, 1990.

\bibitem{HalpernMosesIJCAI}
Joseph~Y. Halpern and Yoram Moses.
\newblock {Characterizing Solution Concepts in Games Using Knowledge-Based
  Programs}.
\newblock In {\em {the 20nd International Joint Conference on Artificial
  Intelligence (IJCAI2007)}}, pages 1300--1307, 2007.

\bibitem{HalpernOneill}
Joseph~Y. Halpern and Kevin~R. O'Neill.
\newblock {Secrecy in Multiagent Systems}.
\newblock {\em ACM Transactions on Information and System Security}, 12(1),
  Article No. 5 2008.

\bibitem{HR2010}
Joseph~Y. Halpern and Nan Rong.
\newblock {Cooperative Equilibrium (Extended Abstract)}.
\newblock In {\em {9th International Joint Conference on Autonomous Agents and
  Multiagent Systems (AAMAS'10)}}, pages 1465--1466, 2010.

\bibitem{Hansen11}
Jens~Ulrik Hansen.
\newblock A hybrid public announcement logic with distributed knowledge.
\newblock {\em Electronic Notes in Theoretical Computer Science}, 273:33--50,
  2011.

\bibitem{CATL}
{Wiebe van der} Hoek, Wojciech Jamroga, and Michael Wooldridge.
\newblock {A logic for strategic reasoning}.
\newblock In {\em Proceedings of the fourth international joint conference on
  Autonomous agents and multiagent systems (AAMAS'05)}, pages 157--164, 2005.

\bibitem{ATEL}
{Wiebe van der} Hoek and Michael Wooldridge.
\newblock {Tractable multiagent planning for epistemic goals}.
\newblock In {\em {Proceedings of the First International Joint Conference on
  Autonomous Agents and Multiagent Systems (AAMAS'02)}}, pages 1167--1174,
  2002.

\bibitem{Horty2001}
John~F. Horty.
\newblock {\em Agency and deontic logic}.
\newblock Oxford University Press, 2001.

\bibitem{HM-SR14}
Xiaowei Huang and Ron van~der Meyden.
\newblock An epistemic strategy logic (extended abstract).
\newblock In {\em 2nd International Workshop on Strategic Reasoning, SR 2014},
  pages 35--41, 2014.

\bibitem{AAAI2014}
Xiaowei Huang and Ron van~der Meyden.
\newblock Symbolic model checking epistemic strategy logic.
\newblock In {\em Twenty-Eighth AAAI Conference on Artificial Intelligence
  (AAAI-14)}, pages 1426--1432, 2014.

\bibitem{HvdM2014}
Xiaowei Huang and Ron van~der Meyden.
\newblock A temporal logic of strategic knowledge.
\newblock In {\em 14th International Conference on Principles of Knowledge
  Representation and Reasoning (KR2014)}, 2014.

\bibitem{Jamroga2003}
Wojciech Jamroga.
\newblock {Some Remarks on Alternating Temporal Epistemic Logic}.
\newblock In {\em {Proceedings of Formal Approaches to Multi-Agent Systems
  (FAMAS 2003)}}, 2003.

\bibitem{JA07}
Wojciech Jamroga and Thomas $\AA$gotnes.
\newblock Constructive knowledge: what agents can achieve under imperfect
  information.
\newblock {\em Journal of Applied Non-Classical Logics}, 17(4):423--475, 2007.

\bibitem{JAvH2008}
Wojciech Jamroga, Thomas $\AA$gotnes, and Wiebe van~der Hoek.
\newblock A simpler semantics for abilities under uncertainty.
\newblock Tech. Report IfI-08-01, Department of Informatics, Clausthal
  University of Technology, 2008.

\bibitem{JD2006}
Wojciech Jamroga and J{\"u}rgen Dix.
\newblock Model checking abilities under incomplete information is indeed
  delta2-complete.
\newblock In {\em the 4th European Workshop on Multi-Agent Systems (EUMAS'06)},
  2006.

\bibitem{JvdH2004}
Wojciech Jamroga and Wiebe van~der Hoek.
\newblock {Agents that Know How to Play }.
\newblock {\em Fundamenta Informaticae}, 62:1--35, 2004.

\bibitem{Jonker2003}
Geert Jonker.
\newblock {Feasible strategies in alternating-time temporal}.
\newblock Master's thesis, University of Utrech, The Netherlands, 2003.

\bibitem{MCMASATL2006}
Alessio Lomuscio and Franco Raimondi.
\newblock {Model Checking Knowledge, Strategies, and Games in Multi-Agent
  Systems}.
\newblock In {\em {the proceedings of the 5th international joint conference on
  Autonomous agents and multiagent systems (AAMAS 2006)}}, pages 161--168,
  2006.

\bibitem{MeydenTARK96}
Ron van~der Meyden.
\newblock {Knowledge Based Programs: On the Complexity of Perfect Recall in
  Finite Environments}.
\newblock In {\em {6th Conference on Theoretical Aspects of Rationality and
  Knowledge (TARK 1996)}}, pages 31--49, 1996.

\bibitem{MogaveroMV10}
Fabio Mogavero, Aniello Murano, and Moshe~Y. Vardi.
\newblock Reasoning about strategies.
\newblock In {\em IARCS Annual Conference on Foundations of Software Technology
  and Theoretical Computer Science (FSTTCS 2010)}, pages 133--144, 2010.

\bibitem{Parikh83}
Rohit Parikh.
\newblock Propositional game logic.
\newblock In {\em IEEE Symp. on Foundations of Computer Science}, pages
  195--200, 1983.

\bibitem{RamParikh84}
Rohit Parikh and Ramaswamy Ramanujam.
\newblock {Distributed Processes and the Logic of Knowledge}.
\newblock In {\em {Logics of Programs 1985}}, pages 256--268, 1985.

\bibitem{Pauly2002}
Marc Pauly.
\newblock A modal logic for coalitional power in games.
\newblock {\em Journal of Logic and Computation}, 12(1):149--166, 2002.

\bibitem{Pnueli1977}
Amir Pnueli.
\newblock {The Temporal Logic of Programs}.
\newblock In {\em {Symp. on Foundations of Computer Science}}, pages 46--57,
  1977.

\bibitem{RamSimon08}
Ramaswamy Ramanujam and Sunil~Easaw Simon.
\newblock Dynamic logic on games with structured strategies.
\newblock In {\em Eleventh International Conference on Principles of Knowledge
  Representation and Reasoning (KR2008)}, pages 49--58, 2008.

\bibitem{Roy09a}
Olivier Roy.
\newblock A dynamic-epistemic hybrid logic for intentions and information
  changes in strategic games.
\newblock {\em Synthese}, 171(2):291--320, 2009.

\bibitem{Schnoor10}
Henning Schnoor.
\newblock Explicit strategies and quantification for {ATL} with incomplete
  information and probabilistic games.
\newblock Technical Report 1008, Institut f\"{u}r Informatik,
  Christian-Albrechts Universit\"{a}t zu Kiel, Aug. 2010.

\bibitem{Schobbens2004}
Pierre-Yves Schobbens.
\newblock Alternating-time logic with imperfect recall.
\newblock {\em Electronic Notes in Theoretical Computer Science}, 85(2):82--93,
  2004.

\bibitem{SchwentickW07}
Thomas Schwentick and Volker Weber.
\newblock Bounded-variable fragments of hybrid logics.
\newblock In {\em Proc. STACS 2007, 24th Annual Symposium on Theoretical
  Aspects of Computer Science}, volume 4393 of {\em Springer LNCS}, pages
  561--572, 2007.

\bibitem{SistlaC85}
A.~Prasad Sistla and Edmund~M. Clarke.
\newblock The complexity of propositional linear temporal logics.
\newblock {\em Journal of the ACM}, 32(3):733--749, 1985.

\bibitem{sutherland_86}
D.~Sutherland.
\newblock A model of information.
\newblock In {\em Proceedings of the 9th National Computer Security
  Conference}, pages 175--183, 1986.

\bibitem{vanEijck13}
Jan van Eijck.
\newblock {PDL} as a multi-agent strategy logic.
\newblock In {\em Proc. Conf. on Theoretical Aspects of Reasoning about
  Knowledge}, 2013.
\newblock published in CoRR, http://arxiv.org/abs/1310.6437.

\bibitem{vOJ2005}
Sieuwert van Otterloo and Geert Jonker.
\newblock {On Epistemic Temporal Strategic Logic}.
\newblock {\em Electronic Notes in Theoretical Computer Science}, 126:77--92,
  2005.

\bibitem{Vardi96tark-kbp}
Moshe~Y. Vardi.
\newblock Implementing knowledge-based programs.
\newblock In {\em the Sixth Conference on Theoretical Aspects of Rationality
  and Knowledge}, pages 15--30, 1996.

\bibitem{WJ90}
J.~Todd Wittbold and Dale~M. Johnson.
\newblock Information flow in nondeterministic systems.
\newblock In {\em Proc. IEEE Symp. on Security and Privacy}, pages 144--161,
  1990.

\end{thebibliography}

\end{document}